\documentclass[journal,twocolumn,10pt]{IEEEtran}
\usepackage{amsmath,amssymb}
\usepackage{amsthm}
\usepackage{booktabs}
\usepackage{array}
\usepackage{tikz}
\usepackage{pgfplots}
\pgfplotsset{compat=1.16}
\usetikzlibrary{arrows.meta,positioning,calc,trees,shapes.geometric,decorations.pathreplacing}
\tikzset{
  inode/.style={circle,draw,fill=white,inner sep=0pt,minimum size=2.2mm},
  leaf/.style={rectangle,draw,fill=white,inner sep=2pt,minimum height=4.2mm,font=\footnotesize},
  sub/.style={isosceles triangle,draw,fill=white,shape border rotate=90,
    isosceles triangle stretches,minimum width=13mm,minimum height=10mm,
    inner sep=0pt,anchor=apex,font=\scriptsize,align=center},
  nlab/.style={font=\scriptsize,align=left,inner sep=1pt},
  bit/.style={font=\scriptsize,inner sep=3pt},
  every edge from parent/.style={draw,thin},
  parent anchor=south, child anchor=north}
\pgfplotsset{paper/.style={
  tick label style={font=\scriptsize}, label style={font=\footnotesize},
  title style={font=\footnotesize}, grid=major, grid style={gray!20},
  axis on top, clip=true}}

\newtheorem{theorem}{Theorem}
\newtheorem{lemma}{Lemma}
\newtheorem{proposition}{Proposition}
\newtheorem{corollary}{Corollary}
\theoremstyle{definition}
\newtheorem{definition}{Definition}
\newtheorem{example}{Example}

\newcommand{\lgg}{\log_2}
\newcommand{\CC}{C}
\newcommand{\lam}{\lambda}
\newcommand{\bp}{\mathbf p}
\newcommand{\bq}{\mathbf q}
\newcommand{\bx}{\mathbf x}
\newcommand{\Ent}{\mathrm H}

\newcommand{\Rmax}{R_{\max}}
\newcommand{\hk}{\bigl(1-h(\tfrac13)\bigr)}

\title{A Gallager-Type Redundancy Bound for Binary Fano Coding}
\author{Kamila~Szewczyk%
\thanks{K. Szewczyk is with the Algorithmic Bioinformatics group,
Saarland University, Saarbr\"ucken, Germany
(e-mail: szewczyk@cs.uni-saarland.de).}}

\begin{document}
\maketitle

\begin{abstract}
Binary Fano codes build a complete prefix code by repeatedly
splitting the sorted list of symbol probabilities into two halves whose
total probabilities are as close to each other as possible. Kraj\v{c}i, Liu,
Mike\v{s}, and Moser proved in 2015 that the redundancy $R$ of this code
satisfies $R\le1-p_n$, where $p_n$ is the smallest symbol probability. We
prove that $R<U(p_1)$, where $U$ is a piecewise function of the largest
symbol probability $p_1$ with $U<1$. In particular,
$R<p_1+4-\tfrac53\lgg5$ for every source, and $R<\tfrac52-\tfrac56\lgg5$
when $p_1<\tfrac12$. The bound equals the supremum of the redundancy for
$p_1\ge\tfrac12$ and on an interval just below $\tfrac13$, and we give a
lower bound for the other values of $p_1$. Unlike the bounds for Huffman
codes, the proofs do not use the sibling property. They work directly
with the Fano split rule, mainly through the facts that every split is
nearly balanced and that the redundancy is a sum of one term per internal
node of the tree.
\end{abstract}


\IEEEpeerreviewmaketitle

\section{Introduction}\label{sec:intro}

\subsection{Fano coding and its redundancy}

\IEEEPARstart{B}{inary} Fano coding \cite{Shannon48,Fano49} begins by sorting the
source symbols by decreasing probability. The sorted list is then split into two
halves of consecutive symbols such that their total probabilities are as nearly
equal as possible. Every symbol in the first half is assigned the code bit $0$,
the remaining symbols are assigned the code bit $1$. Each half is recursively
split in the same way until every block is a single symbol. Unlike Huffman coding
\cite{Huffman52}, Fano coding may not produce an optimal code. Because the
splits are contiguous, the codewords follow the same order as the sorted
symbols. Consequently, a Fano code is an alphabetic code \cite{GilbertMoore59} for
this order, and its expected length is at least that of the optimal
alphabetic code, the Hu--Tucker code \cite{HuTucker71}.

Let $\bp=(p_1,\dots,p_n)$ with $p_1\ge\cdots\ge p_n>0$ be a source, and
let $\ell_i$ be the length of the codeword of symbol $i$. The redundancy is
given by the expected codeword length minus the entropy,
\begin{equation}\label{eq:Rintro}
  R(\bp)\ =\ \sum_{i=1}^{n}p_i\ell_i+\sum_{i=1}^{n}p_i\lgg p_i .
\end{equation}
Gallager \cite{Gallager78} proved the following bound for the redundancy of Huffman codes in
terms of the largest probability $p_1$:
\begin{equation}\label{eq:gallager}
\begin{aligned}
  &R(\bp)\ \le\ p_1+\sigma,\\
  &\sigma\ :=\ 1-\lgg e+\lgg\lgg e\ \approx\ 0.0861.
\end{aligned}
\end{equation}
Ye and Yeung \cite{YeYeung02} gave a simple upper bound for the case that
the probability of one symbol is known. Johnsen \cite{Johnsen80}
improved the upper and lower bounds for $p_1\ge0.4$. Capocelli,
Giancarlo, and Taneja \cite{CapocelliGiancarloTaneja86} and Capocelli and
De~Santis \cite{CapocelliDeSantis89,CapocelliDeSantis91} gave sharper
upper bounds, Montgomery and Abrahams \cite{MontgomeryAbrahams87} gave
tight lower bounds for all $p_1$, and Manstetten \cite{Manstetten92}
determined the largest redundancy of Huffman codes as a function of
$p_1$. These proofs rely on the sibling property of Huffman trees
\cite{Gallager78}. The problem may thus be regarded as solved for Huffman codes.

Fano codes have received less attention. The tree produced by Fano's algorithm is a
weight-balanced tree. Rissanen \cite{Rissanen73} and Horibe
\cite{Horibe77} bounded the cost of such trees. Nakatsu \cite{Nakatsu91}
bounded the redundancy of alphabetic codes. Bruno, De~Prisco, De~Santis,
and Vaccaro \cite{BrunoDePriscoDeSantisVaccaro24} recently gave
linear-time constructions of nearly optimal alphabetic codes and new
bounds on the cost of optimal binary search trees. Kraj\v{c}i, Liu, Mike\v{s},
and Moser \cite{Krajci15} showed that every binary Fano code satisfies
\begin{equation}\label{eq:klmm}
  R(\bp)\ \le\ 1-p_n\ <\ 1\qquad\text{(Fano codes)}.
\end{equation}
Cicalese and Rossi \cite{CicaleseRossi20} studied the redundancy of
$D$-ary Fano codes.
No bound in terms of $p_1$, analogous to \eqref{eq:gallager}, was known for Fano codes.

\subsection{Main result}

We bound the redundancy of binary Fano coding by a piecewise function $U(p_1)$ of
the largest probability. It is defined separately on seven intervals
(Theorem~\ref{thm:main}). Fig.~\ref{fig:envelope} shows it for
$0<p_1<1$. The main results are:
\begin{itemize}
\item Every source satisfies $R(\bp)<U(p_1)<1$. For $p_1<\tfrac12$, this gives
  $R(\bp)<K=\tfrac52-\tfrac56\lgg5$. For every source,
  $R(\bp)<p_1+\CC$, where $\CC=4-\tfrac53\lgg5$.
\item On $[\tfrac12,1)$ and on $[b_2,\tfrac13)$, with $b_2\approx0.3196$
  defined in Section~\ref{sec:main}, the bound equals the supremum of $R$ over all
  sources with largest probability $p_1$ (Theorem~\ref{thm:exact}).
\item On the rest of $(0,1)$ we give a lower bound $L(p_1)$ on this
  supremum (Corollary~\ref{cor:L}).
  Fig.~\ref{fig:envelope} shows $L$, and Table~\ref{tab:gaps} lists the
  difference $U-L$ at selected points.
\end{itemize}

\begin{figure*}[t]
\centering
\pgfmathsetmacro{\bone}{0.3190251}
\pgfmathsetmacro{\btwo}{0.3196324}
\pgfmathsetmacro{\bthree}{0.4157982}
\pgfplotsset{
  envelope/.style={paper,
    declare function={
      ee(\x)=ifthenelse(\x>0,-\x*ln(\x)/ln(2),0);
      CC=4-5/3*log2(5);
      KK=(1+CC)/2;
      uone(\x)=\x+CC;
      utwo(\x)=1.5+CC/2-\x-ee(\x)-ee(0.5-\x);
      uthree(\x)=(4+2*\x)/3+2*(1-\x)*KK/3-ee(\x)-ee((1-\x)/3)-ee(2*(1-\x)/3);
      ufour(\x)=3-3*\x-3*ee(\x)-ee(1-3*\x);
      ufive(\x)=1-ee(\x)-ee(1-\x)+(1-\x)*KK;
      usix(\x)=(5-\x)/3+(2-\x)*CC/3-ee(\x)-ee((1-2*\x)/3)-ee((2-\x)/3);
      useven(\x)=2-\x-ee(\x)-ee(1-\x);
      lthree(\x)=6*\x+3*(1-3*\x)-3*ee(\x)-ee(1-3*\x);
      ltwo(\x)=max(3*\x+3*(1-2*\x)-2*ee(\x)-ee(1-2*\x),1.5-\x/2-ee(\x)-ee(1-\x));
    }},
  ubound/.style={black,line width=0.9pt},
  uexact/.style={black,line width=2.2pt},
  lbound/.style={black,densely dotted,line width=1.1pt},
  prev/.style={black!60,dashed,line width=0.9pt},
  jumpval/.style={only marks,mark=*,mark size=1.5pt,black},
  jumplim/.style={only marks,mark=*,mark size=1.5pt,mark options={fill=white,draw=black}}}
\begin{tikzpicture}
\begin{axis}[envelope, name=full,
  width=0.36\textwidth, height=0.3\textwidth,
  xmin=0, xmax=1, ymin=0, ymax=1.12,
  xtick={0,0.25,0.33333,0.4157982,0.5,0.75,1},
  xticklabels={$0$,$\tfrac14$,$\tfrac13$,$b_3$,$\tfrac12$,$\tfrac34$,$1$},
  ytick={0,0.2,0.4,0.6,0.8,1},
  xlabel={$p_1$}, ylabel={redundancy (bits)},
  title={(a) $0<p_1<1$},
  legend to name=envlegend, legend columns=4,
  legend style={font=\footnotesize,draw=none,column sep=4mm},
  legend cell align=left]
\addplot[prev,domain=0:1,samples=2]{1};
\addlegendentry{previous bound $R<1$}
\addplot[ubound,domain=0:0.25,samples=2]{uone(x)};
\addlegendentry{bound $U$}
\addplot[uexact,domain=0.5:1,samples=80]{useven(x)};
\addlegendentry{pieces where $U$ is the supremum}
\addplot[lbound,domain=0.5:1,samples=60]{2-x-ee(x)-ee(1-x)};
\addlegendentry{lower bound $L$}
\addplot[ubound,domain=0.25:\bone,samples=40]{utwo(x)};
\addplot[ubound,domain=\bone:\btwo,samples=4]{uthree(x)};
\addplot[uexact,domain=\btwo:0.33333,samples=30]{ufour(x)};
\addplot[ubound,domain=0.33333:\bthree,samples=40]{ufive(x)};
\addplot[ubound,domain=\bthree:0.5,samples=40]{usix(x)};
\addplot[lbound,domain=0.25:0.33333,samples=40]{lthree(x)};
\addplot[lbound,domain=0.33333:0.5,samples=90]{ltwo(x)};
\addplot[jumpval] coordinates {(0.25,0.315060) (0.33333,0.458411) (0.5,0.5)};
\addplot[jumplim] coordinates {(0.25,0.380120) (0.33333,0.415037) (0.5,0.565060) (0,0.130120) (1,1)};
\end{axis}
\begin{axis}[envelope, name=zoom,
  at={(full.south east)}, anchor=south west, xshift=9mm,
  width=0.36\textwidth, height=0.3\textwidth,
  xmin=0.24, xmax=0.51, ymin=0.15, ymax=0.62,
  xtick={0.25,0.33333,0.4157982,0.5},
  xticklabels={$\tfrac14$,$\tfrac13$,$b_3$,$\tfrac12$},
  extra x ticks={0.3190251,0.3196324}, extra x tick labels={},
  extra x tick style={grid=none,major tick length=4pt},
  ytick={0.2,0.3,0.4,0.5,0.6},
  xlabel={$p_1$},
  title={(b) $\tfrac14\le p_1\le\tfrac12$}]
\addplot[ubound,domain=0.24:0.25,samples=2]{uone(x)};
\addplot[ubound,domain=0.25:\bone,samples=60]{utwo(x)};
\addplot[ubound,domain=\bone:\btwo,samples=4]{uthree(x)};
\addplot[uexact,domain=\btwo:0.33333,samples=40]{ufour(x)};
\addplot[ubound,domain=0.33333:\bthree,samples=60]{ufive(x)};
\addplot[ubound,domain=\bthree:0.5,samples=60]{usix(x)};
\addplot[uexact,domain=0.5:0.51,samples=4]{useven(x)};
\addplot[lbound,domain=0.25:0.33333,samples=60]{lthree(x)};
\addplot[lbound,domain=0.33333:0.5,samples=120]{ltwo(x)};
\addplot[lbound,domain=0.5:0.51,samples=4]{2-x-ee(x)-ee(1-x)};
\addplot[jumpval] coordinates {(0.25,0.315060) (0.33333,0.458411) (0.5,0.5)};
\addplot[jumplim] coordinates {(0.25,0.380120) (0.33333,0.415037) (0.5,0.565060)};
\draw[black!60,thin] (axis cs:0.3165,0.262) rectangle (axis cs:0.3215,0.286);
\end{axis}
\begin{axis}[envelope, name=inset,
  at={(zoom.south east)}, anchor=south west, xshift=9mm,
  width=0.36\textwidth, height=0.3\textwidth,
  xmin=0.3165, xmax=0.3215, ymin=0.262, ymax=0.286,
  xtick={0.3190251,0.3196324}, xticklabels={$b_1$,$b_2$},
  extra x ticks={0.317,0.321},
  extra x tick labels={$0.317$,$0.321$},
  ytick={0.265,0.27,0.275,0.28,0.285},
  yticklabels={$0.265$,$0.270$,$0.275$,$0.280$,$0.285$},
  scaled x ticks=false, scaled y ticks=false,
  xlabel={$p_1$},
  title style={align=center},
  title={(c) the pieces $U_2$, $U_3$, $U_4$\\ near $b_1$ and $b_2$}]
\addplot[ubound,domain=0.3165:\bone,samples=20]{utwo(x)};
\addplot[ubound,domain=\bone:\btwo,samples=6]{uthree(x)};
\addplot[uexact,domain=\btwo:0.3215,samples=30]{ufour(x)};
\addplot[lbound,domain=0.3165:0.3215,samples=40]{lthree(x)};
\end{axis}
\node[anchor=north] at ($(zoom.south)+(0,-11mm)$) {\pgfplotslegendfromname{envlegend}};
\end{tikzpicture}
\caption{The bound $U$ of Theorem~\ref{thm:main}, the previous bound
$R<1$, and the lower bound $L$ of Corollary~\ref{cor:L}, which is
drawn for $p_1\ge\tfrac14$ and discussed for $p_1<\tfrac14$ in
Section~\ref{ss:gaps}. The points $b_1<b_2<b_3$ are defined in
Lemma~\ref{lem:boundaries}. Panel (b) enlarges the interval $[\tfrac14,\tfrac12]$, and
panel (c) enlarges the small rectangle in (b). Filled dots give the value
of $U$ at its jumps, open dots give one-sided limits, thick pieces have $L=U$.}
\label{fig:envelope}
\end{figure*}

\subsection{Motivation}\label{ss:motiv}

We bound $R$ in terms of $p_1$ because the largest probability largely
determines the top of the Fano tree. For example, the split at the root
separates symbol $1$ from the others, so that $\ell_1=1$, if and only if
$2p_1+p_2\ge1$ (Lemma~\ref{lem:iso}). This always happens
for $p_1\ge\tfrac12$ and never for $p_1<\tfrac13$. If
$\tfrac14\le p_1<\tfrac13$, then $\ell_1=2$
(Lemma~\ref{lem:noburial-q}). So the shape of the Fano tree near the root
depends on the range of $p_1$, and the bound has a different form on each
range.

Gallager's bound \eqref{eq:gallager} for Huffman codes is also stated in
terms of $p_1$. Our bound $R<p_1+\CC$ for Fano codes has the same form,
$p_1$ plus a constant, so the two upper bounds can be compared directly
(Section~\ref{ss:gallager}).

A Fano tree is built from the top down and need not have the sibling
property, because the code need not be optimal. So the methods used in bounds
for Huffman codes do not apply. Instead we work directly with the Fano
split rule of Definition~\ref{def:fano}, which determines the top of the
tree (Lemmas~\ref{lem:iso} and~\ref{lem:shapes}). We use it mainly through
two facts. The redundancy is a sum of one term per internal node
(Lemma~\ref{lem:rec}), and every split is nearly balanced
(Lemma~\ref{lem:key}). With these facts we prove the bound by induction on
the number of symbols. Section~\ref{ss:plan} outlines the proofs.
Section~\ref{ss:ext} discusses other codes to which the method may apply.

\subsection{Organization}

Section~\ref{sec:model} defines the code and the redundancy, proves the
recursion, and gives an example of binary Fano coding. Section~\ref{sec:struct} proves
the properties of Fano splits used by the later proofs.
Section~\ref{sec:main} states the results. It derives
the main theorem (Theorem~\ref{thm:main}) from four theorems, one for each range of $p_1$.
Sections~\ref{sec:half} to~\ref{sec:quarter} prove these four theorems.
Section~\ref{sec:disc} gives the lower bounds and lists open questions.
The appendices contain routine computations used in the proofs and the
list of numerical inequalities.

\section{Model}\label{sec:model}

A \emph{source} is a probability vector
\begin{equation}\label{eq:source}
\begin{aligned}
  &\bp=(p_1,\dots,p_n),\qquad p_1\ge p_2\ge\cdots\ge p_n>0,\\
  &\sum_{i=1}^{n}p_i=1,\qquad n\ge2 .
\end{aligned}
\end{equation}
A \emph{block} is a set of consecutive indices $B=\{i,\dots,j\}$ with
$1\le i\le j\le n$. Its \emph{mass} is given by
\begin{equation}\label{eq:mass}
  M_B\ :=\ p_i+\dots+p_j .
\end{equation}

\begin{definition}[Binary Fano code]\label{def:fano}
The code tree is built recursively from the block $\{1,\dots,n\}$. The singleton block
$B=\{i\}$ with one symbol is a leaf. A block $B=\{i,\dots,j\}$ with $i<j$ is
split after a position $k$ with $i\le k<j$. The left half
$B_L=\{i,\dots,k\}$ receives the code bit $0$. The right half
$B_R=\{k+1,\dots,j\}$ receives the code bit $1$. The \emph{cut} $k$
is the smallest position that minimizes the \emph{imbalance}
\begin{equation}\label{eq:imbalance}
  \bigl|M_{B_L}-M_{B_R}\bigr|
  \ =\ \bigl|\,2(p_i+\dots+p_k)-M_B\bigr| .
\end{equation}
The resulting children are then recursively split in the same way.
\end{definition}

The codeword of symbol $i$ is the sequence of bits on the path from the root
to the leaf $\{i\}$, so its length $\ell_i$ is the depth of that leaf.

\subsection{Length, entropy, redundancy}

For a source $\bp$ with codeword lengths $\ell_1,\dots,\ell_n$ let
\begin{align}
  \bar\ell(\bp)\ &:=\ \sum_{i=1}^{n}p_i\,\ell_i ,\label{eq:L}\\
  \Ent(\bp)\ &:=\ -\sum_{i=1}^{n}p_i\lgg p_i ,\label{eq:H}\\
  R(\bp)\ &:=\ \bar\ell(\bp)-\Ent(\bp) .\label{eq:R}
\end{align}
Since every internal node of the tree has two children, the code is complete, so
$\sum_{i=1}^{n}2^{-\ell_i}=1$. The redundancy may be written in terms of the Kullback--Leibler divergence $\mathrm D$ as
\begin{equation}\label{eq:KL}
\begin{aligned}
  R(\bp)&\ =\ \sum_{i=1}^{n}p_i\bigl(\ell_i+\lgg p_i\bigr)\\
  &\ =\ \mathrm D\bigl(\bp\,\big\|\,(2^{-\ell_1},\dots,2^{-\ell_n})\bigr)\ \ge\ 0.
\end{aligned}
\end{equation}

For a block $B=\{i,\dots,j\}$ let
\begin{equation}\label{eq:child}
  \bp^{(B)}\ :=\ \frac{1}{M_B}\,(p_i,\dots,p_j)
\end{equation}
be the source obtained by rescaling the block to total probability one.
The subtree of the code tree below $B$ is the Fano tree of $\bp^{(B)}$.
This holds because Definition~\ref{def:fano} compares only masses within
the block, and scaling all masses by $1/M_B$ does not change which cut
minimizes the imbalance. We measure the redundancy of a block in the
units of the whole source, and write
\begin{align}
  R_B\ &:=\ M_B\,R\bigl(\bp^{(B)}\bigr),\label{eq:RB}\\
  R_{\{i\}}\ &:=\ 0,\label{eq:RBsingle}\\
  R(\bp)\ &=\ R_{\{1,\dots,n\}} .\label{eq:RBroot}
\end{align}

\subsection{Example}\label{ss:example}

We build the code for $\bp=(0.30,\,0.25,\,0.20,\,0.15,\,0.10)$. At the
four possible cuts of the root the left masses are $0.30$, $0.55$,
$0.75$, $0.90$, and the imbalances are $0.40$, $0.10$, $0.50$, $0.80$.
Consequently, the root splits into the left child $\{1,2\}$ of
mass $0.55$ and the right child $\{3,4,5\}$ of mass $0.45$. The block
$\{1,2\}$ has only one possible cut. The block $\{3,4,5\}$ rescales to the
source $(\tfrac49,\tfrac13,\tfrac29)$, whose two cuts have imbalances
$\tfrac19$ and $\tfrac59$. So $\{3,4,5\}$ splits into $\{3\}$ and
$\{4,5\}$, and $\{4,5\}$ splits once more. The codeword lengths are $(2,2,2,3,3)$, which gives
$\bar\ell=2.25$, $\Ent\approx2.228213$ and $R\approx0.021787$ bits.
Fig.~\ref{fig:example} shows the tree obtained by the algorithm.

\begin{figure}[t]
\centering
\begin{tikzpicture}[level distance=15mm,
  level 1/.style={sibling distance=30mm},
  level 2/.style={sibling distance=20mm},
  level 3/.style={sibling distance=10mm}]
\node[inode] (r) {}
  child {node[inode] (l) {}
    child {node[leaf] (s1) {$0.30$} edge from parent node[bit,left]{0}}
    child {node[leaf] (s2) {$0.25$} edge from parent node[bit,right]{1}}
    edge from parent node[bit,above left]{0}}
  child {node[inode] (rr) {}
    child {node[leaf] (s3) {$0.20$} edge from parent node[bit,left]{0}}
    child {node[inode] (rrr) {}
      child {node[leaf] (s4) {$0.15$} edge from parent node[bit,left]{0}}
      child {node[leaf] (s5) {$0.10$} edge from parent node[bit,right]{1}}
      edge from parent node[bit,right]{1}}
    edge from parent node[bit,above right]{1}};
\foreach \n/\c in {s1/00,s2/01,s3/10,s4/110,s5/111}
  \node[font=\scriptsize,below=0.5mm of \n] {\texttt{\c}};
\node[nlab,above=1.5mm of r,align=center] {$M=1$, $a=0.55$\\ term $0.00723$};
\node[nlab,anchor=south east] at ($(l)+(-1mm,1mm)$) {$M=0.55$, $a=\tfrac6{11}$};
\node[nlab,anchor=south west] at ($(rr)+(1mm,1mm)$) {$M=0.45$, $a=\tfrac49$};
\node[nlab,anchor=south west] at ($(rrr)+(1mm,1mm)$) {$M=0.25$, $a=\tfrac35$};
\end{tikzpicture}
\caption{The Fano tree of $\bp=(0.30,0.25,0.20,0.15,0.10)$. Each internal node shows the mass $M$ of its block and the fraction
$a$ of $M$ that lies in the left child.}
\label{fig:example}
\end{figure}

\subsection{Entropy of a nonnegative tuple}\label{ss:tuple}

For a tuple $\bx=(x_1,\dots,x_k)$ of nonnegative numbers, not necessarily
summing to one, we write
\begin{equation}\label{eq:Htuple}
\begin{aligned}
  &\Ent(x_1,\dots,x_k)\ :=\ -(x_1\lgg x_1+\dots+x_k\lgg x_k),\\
  &\text{with }\qquad 0\lgg0:=0 .
\end{aligned}
\end{equation}
This agrees with \eqref{eq:H} when $\bx$ is a source. The binary
entropy function is
\begin{equation}\label{eq:hbin}
  h(x)\ :=\ \Ent(x,1-x)\qquad(0\le x\le1).
\end{equation}
It is concave, with $h(0)=h(1)=0$ and $h(\tfrac12)=1$. Two identities are
used throughout. If $1\le i<k$ and $S=x_{i+1}+\cdots+x_k>0$, then the
grouping identity
\begin{equation}\label{eq:grouping}
\begin{aligned}
  \Ent(x_1,\dots,x_k)\ &=\ \Ent(x_1,\dots,x_i,S)\\
  &\qquad+S\,\Ent\Bigl(\frac{x_{i+1}}{S},\dots,\frac{x_k}{S}\Bigr)
\end{aligned}
\end{equation}
holds. If $\bx$ has total mass $M=x_1+\cdots+x_k>0$, then the scaling
identity
\begin{equation}\label{eq:scaling}
  \Ent\Bigl(\frac{\bx}{M}\Bigr)\ =\ \frac{\Ent(\bx)}{M}+\lgg M
\end{equation}
holds. Both follow from $\lgg(x/S)=\lgg x-\lgg S$.

\subsection{Recursion}\label{ss:recursion}

\begin{lemma}[Recursion]\label{lem:rec}
Let $B$ be a block with at least two symbols. Consider its split into $B_L$ and $B_R$,
and let $a:=M_{B_L}/M_B$. Then
\begin{equation}\label{eq:rec}
  R_B\ =\ M_B\bigl(1-h(a)\bigr)+R_{B_L}+R_{B_R}.
\end{equation}
Consequently, $R(\bp)$ is the sum of the terms $M_B(1-h(a))$ over the
internal nodes of the code tree. Here $B$ is the block at the node, and
$aM_B$ is the mass of its left child.
\end{lemma}

\begin{proof}
We compare the expected length and the entropy of the rescaled block with
those of its two children. Let $\bq=\bp^{(B)}$. Its children are $\bq_L=\bp^{(B_L)}$ and
$\bq_R=\bp^{(B_R)}$, of masses $a$ and $1-a$ within $\bq$. Every symbol of
a child is one level deeper in the tree of $\bq$ than in the tree of the
child. So $\bar\ell(\bq)=1+a\,\bar\ell(\bq_L)+(1-a)\,\bar\ell(\bq_R)$. The grouping identity
\eqref{eq:grouping}, applied twice, gives
$\Ent(\bq)=h(a)+a\,\Ent(\bq_L)+(1-a)\,\Ent(\bq_R)$. Subtracting the two
identities gives $R(\bq)=(1-h(a))+a\,R(\bq_L)+(1-a)\,R(\bq_R)$. Multiplying by $M_B$ and
using $aM_B=M_{B_L}$ and $(1-a)M_B=M_{B_R}$ gives \eqref{eq:rec}. Applying
\eqref{eq:rec} at every internal node, with $R_{\{i\}}=0$ at the leaves,
gives the sum.
\end{proof}

\subsection{Bounds on binary entropy}\label{ss:kappa}

By Lemma~\ref{lem:rec} a node of mass $M$ whose left child has mass $aM$
contributes a term of $M(1-h(a))$ to the total redundancy, which is zero only
for $a=\tfrac12$. Fig.~\ref{fig:kappa} shows $1-h(a)$ and the bounds of
Lemma~\ref{lem:kappa}.

\begin{figure*}[t]
\centering
\begin{tikzpicture}
\pgfplotsset{hbounds/.style={paper,
  declare function={
    ee(\x)=ifthenelse(\x>0,-\x*ln(\x)/ln(2),0);
    kap(\x)=1-ee(\x)-ee(1-\x);
    lam=8+3*log2(3)-5*log2(5);
  }}}
\begin{axis}[hbounds, name=left,
  width=0.48\textwidth, height=0.25\textwidth,
  xlabel={$a$}, title={(a) $0\le a\le1$},
  xmin=0, xmax=1, ymin=0, ymax=1.05,
  xtick={0,0.25,0.5,0.75,1}, xticklabels={$0$,$\tfrac14$,$\tfrac12$,$\tfrac34$,$1$},
  legend to name=hlegend, legend columns=4,
  legend style={font=\footnotesize,draw=none,column sep=5mm}]
\addplot[black,line width=1.1pt,domain=0:1,samples=161]{kap(x)};
\addlegendentry{$1-h(a)$}
\addplot[black,dashed,domain=0:1,samples=3]{abs(2*x-1)};
\addlegendentry{$|2a-1|$}
\addplot[black,densely dotted,line width=1pt,domain=0:1,samples=81]{(2*x-1)^2};
\addlegendentry{$(2a-1)^2$}
\addplot[black!55,line width=1.1pt,domain=0.5:0.6,samples=2]{(lam-1)*(2*x-1)};
\addlegendentry{$(\lam-1)(2a-1)$ on $[\tfrac12,\tfrac35]$}
\draw[black!60,thin] (axis cs:0.5,0) rectangle (axis cs:0.65,0.06);
\end{axis}
\begin{axis}[hbounds, at={(left.south east)}, anchor=south west, xshift=12mm,
  width=0.48\textwidth, height=0.25\textwidth,
  xlabel={$a$}, title={(b) $\tfrac12\le a\le0.65$},
  xmin=0.5, xmax=0.65, ymin=0, ymax=0.06,
  xtick={0.5,0.55,0.6,0.65}, xticklabels={$\tfrac12$,$0.55$,$\tfrac35$,$0.65$},
  ytick={0,0.02,0.04,0.06}, yticklabels={$0$,$0.02$,$0.04$,$0.06$},
  scaled y ticks=false]
\addplot[black,line width=1.1pt,domain=0.5:0.65,samples=61]{kap(x)};
\addplot[black,densely dotted,line width=1pt,domain=0.5:0.65,samples=61]{(2*x-1)^2};
\addplot[black!55,line width=1.1pt,domain=0.5:0.6,samples=2]{(lam-1)*(2*x-1)};
\addplot[only marks,mark=*,mark size=1.5pt,black] coordinates {(0.6,0.029049)};
\end{axis}
\node[anchor=north] at ($(left.south east)+(6mm,-12mm)$) {\pgfplotslegendfromname{hlegend}};
\end{tikzpicture}
\caption{The function $1-h(a)$ and its upper bounds from
Lemma~\ref{lem:kappa}. Panel (b) enlarges the small rectangle in (a). The
line $(\lam-1)(2a-1)$ lies above $1-h(a)$ on $[\tfrac12,\tfrac35]$ and
meets it at $a=\tfrac12$ and at $a=\tfrac35$.}
\label{fig:kappa}
\end{figure*}

\begin{lemma}[Bounds on $1-h(a)$]\label{lem:kappa}\leavevmode
\begin{enumerate}
\item[\textup{(i)}] For $a\in[0,1]$,
\begin{equation}\label{eq:chord}
  1-h(a)\ \le\ |2a-1| ,
\end{equation}
with equality only for $a\in\{0,\tfrac12,1\}$.
\item[\textup{(ii)}] For $a\in[0,1]$,
\begin{equation}\label{eq:quad}
  1-h(a)\ \le\ (2a-1)^2 .
\end{equation}
\item[\textup{(iii)}] On $[0,1]$ the function $t\mapsto1-h(\tfrac{1+t}2)$
is strictly convex and strictly increasing, and its quotient by $t$ is
strictly increasing on $(0,1]$. So, with
$\lam:=1+5\bigl(1-h(\tfrac25)\bigr)$,
\begin{equation}\label{eq:ratio}
  1-h\Bigl(\frac{1+t}2\Bigr)\ \le\ (\lam-1)\,t\quad\text{for }0\le t\le\tfrac15,
\end{equation}
with equality if and only if $t\in\{0,\tfrac15\}$.
\end{enumerate}
\end{lemma}

\begin{proof}
(iii) Write $f(t):=1-h(\tfrac{1+t}2)$. Since
$h''(x)=-1/(x(1-x)\ln2)<0$ on $(0,1)$, the function $h$ is strictly
concave, and $f$ is strictly convex. Let $0<r<t\le1$. Writing
$r=\tfrac rt\,t+(1-\tfrac rt)\cdot0$, strict convexity and $f(0)=0$ give
$f(r)<\tfrac rt\,f(t)$, so $f(r)/r<f(t)/t$. Since $f(t)>0$ for
$t>0$, the same inequality shows that $f$ is strictly increasing on
$[0,1]$. The bound \eqref{eq:ratio} is the quotient inequality at
$t=\tfrac15$, because $f(\tfrac15)=1-h(\tfrac35)=1-h(\tfrac25)$ and
$5f(\tfrac15)=\lam-1$.

(i) On $[\tfrac12,1]$ the convex function $1-h$ lies below its chord from
$(\tfrac12,0)$ to $(1,1)$, which is the line $2a-1$. Equality holds only
at the ends, because $1-h$ is strictly convex. The case $a\le\tfrac12$
follows from $h(a)=h(1-a)$.

(ii) With $t=2a-1$, the Taylor expansion of $h$ at $\tfrac12$ gives
\begin{equation}\label{eq:taylor}
  1-h(a)\ =\ \frac1{\ln2}\sum_{m\ge1}\frac{t^{2m}}{2m(2m-1)} .
\end{equation}
Since $t^{2m}\le t^2$ for $|t|\le1$ and
\begin{equation}\label{eq:lnsum}
  \sum_{m\ge1}\tfrac1{2m(2m-1)}=\sum_{m\ge1}\bigl(\tfrac1{2m-1}-\tfrac1{2m}\bigr)=\ln2,
\end{equation}
it follows that $1-h(a)\le t^2$.
\end{proof}

\subsection{Convexity of entropy expressions}

For a single number $u\ge0$, write
$\Ent(u)=-u\lgg u$, the case $k=1$ of \eqref{eq:Htuple}.

\begin{lemma}[Convexity of entropy expressions]\label{lem:endpoint}
Let $g$ be an affine function plus $\sum_{j=1}^{m}c_j\,\Ent\bigl(u_j(x)\bigr)$ on an
interval $[x_0,x_1]$, where $u_1,\dots,u_m$ are affine, $u_j>0$ on
$(x_0,x_1)$, and $u_j$ has slope $s_j$. Then on $(x_0,x_1)$
\begin{equation}\label{eq:second}
  g''(x)\ =\ -\frac1{\ln2}\sum_{j=1}^{m}\frac{c_j\,s_j^{\,2}}{u_j(x)} .
\end{equation}
If $g''\le0$ on $(x_0,x_1)$, then $g$ is concave and
$g\ge\min\{g(x_0),g(x_1)\}$ on $[x_0,x_1]$. If $g''\ge0$, then $g$ is
convex and $g\le\max\{g(x_0),g(x_1)\}$. In particular, if every $c_j$ is
negative, then $g$ is convex. If the $u_j$ and the affine part are affine
functions of several variables and every $c_j$ is negative, then $g$ is
convex, because it is convex along every line.
\end{lemma}

\begin{proof}
The function $\Ent(u)=-u\lgg u$ has $\Ent''(u)=-1/(u\ln2)$, and the chain rule gives
\eqref{eq:second}. A concave function on an interval attains its minimum
at an endpoint, and a convex function attains its maximum at an endpoint.
The function $g$ is continuous on the closed interval because $\Ent$ is
continuous on $[0,\infty)$.
\end{proof}

\subsection{Constants and numerical inequalities}\label{ss:consts}

The constants used throughout this paper are
\begin{subequations}\label{eq:constants}
\begin{align}
  1-h(\tfrac13)\ &=\ \tfrac53-\lgg3\ \approx\ 0.081704 ,\label{eq:k13}\\
  1-h(\tfrac25)\ &=\ \tfrac75+\tfrac35\lgg3-\lgg5\nonumber\\
  &\approx\ 0.029049 ,\label{eq:k15}\\
  \lam\ &=\ 1+5\bigl(1-h(\tfrac25)\bigr)\nonumber\\
  &=\ 8+3\lgg3-5\lgg5\ \approx\ 1.145247 ,\label{eq:lam}\\
  \CC\ &:=\ \bigl(1-h(\tfrac13)\bigr)+\frac{\lam-1}{3}\nonumber\\
  &=\ 4-\tfrac53\lgg5\ \approx\ 0.130120 ,\label{eq:C}\\
  K\ &:=\ \frac{1+\CC}{2}\ =\ \tfrac52-\tfrac56\lgg5\nonumber\\
  &\approx\ 0.565060 .\label{eq:K}
\end{align}
\end{subequations}
The number $\lam$ is the
coefficient of the smallest probability in Theorem~\ref{thm:aff}, and
$\CC$ is the constant in the bound $R<p_1+\CC$. The uniform source on
three symbols attains equality in Theorem~\ref{thm:aff}, which fixes
$\CC$ once $\lam$ is chosen. The number $K$ is the bound on $R$ for
$p_1<\tfrac12$ (Corollary~\ref{cor:kcap}). The proofs use the following
inequalities between the constants:
\begin{equation}\label{eq:constfacts}
  \frac{\lam-1}{2}\ <\ \CC,\qquad
  1-h(\tfrac13)\ <\ \tfrac23\,\CC,\qquad
  \tfrac13+\CC\ <\ K,
\end{equation}
and further
\begin{equation}\label{eq:lamrange}
  1\ <\ \lam\ \le\ \tfrac65\ <\ 2,
  \qquad
  \CC\ >\ 0 .
\end{equation}

The inequalities between constants can be easily certified computationally.
Appendix~\ref{app:cert} lists every such inequality used in the paper,
labeled (N1), (N2), and so on. Values of
entropy terms at rational points follow from
$\Ent(\tfrac ab)=\tfrac ab(\lgg b-\lgg a)$.

\section{Properties of Fano Splits}\label{sec:struct}

\begin{figure}[t]
\centering
\begin{tikzpicture}[font=\footnotesize,
  sym/.style={draw,minimum height=6mm,anchor=west,inner sep=0pt}]
\node[sym,minimum width=13mm] (a) at (0,0) {$p_i$};
\node[sym,minimum width=9mm] (b) at (a.east) {$\cdots$};
\node[sym,minimum width=13mm] (c) at (b.east) {$p_{k-1}$};
\node[sym,minimum width=13mm,fill=black!15] (d) at (c.east) {$p_k$};
\node[sym,minimum width=13mm] (e) at (d.east) {$p_{k+1}$};
\node[sym,minimum width=9mm] (f) at (e.east) {$\cdots$};
\node[sym,minimum width=10mm] (g) at (f.east) {$p_j$};
\draw[line width=1.4pt] ($(d.east)+(0,-0.55)$) -- ($(d.east)+(0,0.95)$)
  node[above,font=\scriptsize] {chosen cut};
\draw[thick,dashed] ($(c.east)+(0,-0.55)$) -- ($(c.east)+(0,1.6)$)
  node[above,font=\scriptsize] {cut one position to the left};
\draw[-{Latex[length=2mm]},thick] ($(d.north)+(0,0.05)$) to[out=60,in=120]
  ($(e.north)+(0,0.05)$);
\draw[decorate,decoration={brace,mirror,amplitude=3pt}]
  ($(a.west)+(0,-0.7)$) -- ($(d.east)+(-0.05,-0.7)$)
  node[midway,below=3pt] {left child, mass $M_{B_L}$};
\draw[decorate,decoration={brace,mirror,amplitude=3pt}]
  ($(e.west)+(0.05,-0.7)$) -- ($(g.east)+(0,-0.7)$)
  node[midway,below=3pt] {right child, mass $M_{B_R}$};
\end{tikzpicture}
\caption{The proof of Lemma~\ref{lem:key}. Moving the cut one position to
the left moves $p_k$ (shaded) from the left child to the right child, as
the arrow shows. This changes $M_{B_L}-M_{B_R}$ by $-2p_k$.}
\label{fig:cut}
\end{figure}

\begin{lemma}[Imbalance bound]\label{lem:key}
Let a block $B$ be split after position $k$ into $B_L$ and $B_R$. Then
\begin{equation}\label{eq:key}
  \bigl|M_{B_L}-M_{B_R}\bigr|\ \le\ p_k ,
\end{equation}
where $p_k$ is the last and smallest symbol of $B_L$.
\end{lemma}

\begin{proof}
We compare the chosen cut with the cut one position to the left or to the
right, whose imbalance cannot be smaller. Write
$u=M_{B_L}-M_{B_R}$ (Fig.~\ref{fig:cut}). Suppose $u\ge0$. If $B_L$ is a single symbol, then $u=2p_k-M_B\le p_k$ because
$M_B\ge p_k$. Otherwise moving the cut one position to the left moves $p_k$
to the right child and gives the imbalance $|u-2p_k|$. Since the chosen
cut minimizes the imbalance by construction, $u\le|u-2p_k|$. If $u>2p_k$, this reads
$u\le u-2p_k$, which is impossible. So $u\le2p_k$, hence
$|u-2p_k|=2p_k-u$, and $u\le2p_k-u$ gives $u\le p_k$. Suppose $u<0$. Then $B_R$ has at least
two symbols, because $B_R=\{k+1\}$ would give $M_{B_L}\ge p_k\ge p_{k+1}=M_{B_R}$.
So the cut can move one position to the right, which moves $p_{k+1}$ to
the left child, and the same argument gives $-u\le p_{k+1}\le p_k$.
\end{proof}

We have $\ell_1=1$ exactly when the split at the root separates symbol $1$
from all other symbols. For a block $B$ we say that the split of $B$
separates its first symbol if the left child of $B$ is that symbol alone.

\begin{lemma}[Codeword length one for the first symbol]\label{lem:iso}
We have $\ell_1=1$ if and only if
\begin{equation}\label{eq:isocond}
  2p_1+p_2\ \ge\ 1 .
\end{equation}
It holds if $p_1\ge\tfrac12$ and fails if $p_1<\tfrac13$. When it holds,
\begin{equation}\label{eq:isorec}
\begin{aligned}
  R(\bp)\ &=\ \bigl(1-h(p_1)\bigr)+R_{\{2,\dots,n\}}\\
  &=\ \bigl(1-h(p_1)\bigr)+(1-p_1)\,R\bigl(\bp^{(\{2,\dots,n\})}\bigr).
\end{aligned}
\end{equation}
For a block $B=\{i,\dots,j\}$, the split of $B$ separates its first symbol
under the same condition with $p_1,p_2,1$ replaced by $p_i,p_{i+1},M_B$.
\end{lemma}

\begin{proof}
If $n=2$, then $\ell_1=1$ and $2p_1+p_2=1+p_1\ge1$. Let
$n\ge3$. The cut after
position $1$ is at least as good as the cut after position $2$ if and
only if $|2p_1-1|\le|2p_1+2p_2-1|$. Both sides are nonnegative, so we
may square them. This gives $4p_2(2p_1+p_2-1)\ge0$, which is equivalent to
$2p_1+p_2\ge1$ because $p_2>0$. In this case the
difference $M_{B_L}-M_{B_R}$ after position $2$ is
$2p_1+2p_2-1\ge p_2>0$. Since the difference strictly increases as the
cut moves to the right, every later cut is worse. If the cuts after
positions $1$ and $2$ are equally good, the tie rule selects the cut after
position $1$. If $2p_1+p_2<1$, the cut after position $2$ is strictly
better than the cut after position $1$, so $\ell_1\ge2$. For the two special ranges, $p_1\ge\tfrac12$ gives $2p_1+p_2>1$,
and $p_1<\tfrac13$ gives $2p_1+p_2\le3p_1<1$. The identity
\eqref{eq:isorec} is \eqref{eq:rec} at the root with $B_L=\{1\}$,
$a=p_1$ and $R_{\{1\}}=0$. Finally, a block is split in the same way as
its rescaled source $\bp^{(B)}$, which gives the claim for blocks.
\end{proof}

When $\ell_1=1$, the redundancy is $1-h(p_1)$ plus the
redundancy of the right child. We bound the right child with the lemma
of Kraj\v{c}i, Liu, Mike\v{s}, and Moser.

\begin{lemma}[Kraj\v{c}i, Liu, Mike\v{s}, and Moser
\textup{\cite[Theorem~6, eq.~(9)]{Krajci15}}]\label{lem:univ}
For every source, $R(\bp)\le1-p_n$. In particular $R(\bp)<1$.
\end{lemma}

\begin{corollary}\label{cor:isob}
If $\ell_1=1$, in particular if $p_1\ge\tfrac12$, then
\begin{equation}\label{eq:U7bound}
  R(\bp)\ <\ 2-p_1-h(p_1) .
\end{equation}
If the split of a block $B=\{i,\dots,j\}$ separates its first symbol and
$q:=p_i/M_B$, then $R_B<M_B\,\bigl(2-q-h(q)\bigr)$.
\end{corollary}

\begin{proof}
By \eqref{eq:isorec}, $R(\bp)=(1-h(p_1))+R_{\{2,\dots,n\}}$. If $n=2$, the
right child is a single symbol and $R_{\{2\}}=0<1-p_1$. If $n\ge3$, the
right child rescales to a source with $n-1$ symbols, and
Lemma~\ref{lem:univ} applied to it gives
$R_{\{2,\dots,n\}}=(1-p_1)R(\bp^{(\{2,\dots,n\})})<1-p_1$. In both cases
$R(\bp)<(1-h(p_1))+(1-p_1)$. The bound for a block is the same computation
for $\bp^{(B)}$, multiplied by $M_B$.
\end{proof}

Lemma~\ref{lem:shapes} lists the possible root splits for $p_1<\tfrac12$,
labeled by the sizes of the two children (Fig.~\ref{fig:shapes},
Table~\ref{tab:shapes}).

\begin{figure}[t]
\centering
\begin{tikzpicture}[level distance=10mm, sibling distance=16mm]
\begin{scope}[xshift=0cm]
\node[inode]{}
  child {node[leaf]{$p_1$}}
  child {node[inode]{} child {node[leaf]{$p_2$}} child {node[leaf]{$p_3$}}};
\node[font=\footnotesize] at (0,-2.75) {$(1,2)$};
\end{scope}
\begin{scope}[xshift=4cm]
\node[inode]{}
  child {node[leaf]{$p_1$}}
  child {node[sub]{$\ge3$}};
\node[font=\footnotesize] at (0,-2.75) {$(1,\ge3)$};
\end{scope}
\begin{scope}[yshift=-3.5cm,sibling distance=20mm]
\node[inode]{}
  child {node[inode]{} child[sibling distance=12mm] {node[leaf]{$p_1$}} child[sibling distance=12mm] {node[leaf]{$p_2$}}}
  child {node[sub]{$\ge2$}};
\node[font=\footnotesize] at (0,-2.75) {$(2,\ge2)$};
\end{scope}
\begin{scope}[xshift=4cm,yshift=-3.5cm,sibling distance=20mm]
\node[inode]{}
  child {node[sub]{$\ge3$}}
  child {node[sub]{$\ge2$}};
\node[font=\footnotesize] at (0,-2.75) {$(\ge3,\ge2)$};
\end{scope}
\end{tikzpicture}
\caption{The four possible root splits for $p_1<\tfrac12$
(Lemma~\ref{lem:shapes}), labeled by the sizes of the two children. A
triangle stands for a subtree with the indicated number of leaves.}
\label{fig:shapes}
\end{figure}

\begin{lemma}\label{lem:shapes}
Let $p_1<\tfrac12$ and $n\ge3$, and let $a=M_{B_L}$ be the mass of the left
child of the root. Then the right child has at least two symbols, and one
of the following holds:
\begin{enumerate}
\item[\textup{(i)}] the left child is $\{1\}$, $2p_1+p_2\ge1$, and
$p_1\ge\tfrac13$;
\item[\textup{(ii)}] the left child has two symbols, and
$a=p_1+p_2\in[\tfrac25,\tfrac23)$;
\item[\textup{(iii)}] the left child has at least three symbols, and if
$a\ge\tfrac12$, then $a\le\tfrac35$.
\end{enumerate}
\end{lemma}

\begin{table}[t]
\renewcommand{\arraystretch}{1.25}
\caption{The root splits for $p_1<\tfrac12$ and the constraints used in
the induction of Section~\ref{sec:affine}. Here $a$ is the mass of the left
child.}
\label{tab:shapes}
\centering
\footnotesize
\setlength{\tabcolsep}{4pt}
\begin{tabular}{@{}lll@{}}
\toprule
sizes & condition & constraints\\
\midrule
$(1,2)$ & $2p_1+p_2\ge1$, $n=3$ & $\tfrac13\le p_1$\\
$(1,\ge3)$ & $2p_1+p_2\ge1$, $n\ge4$ & $\tfrac13\le p_1$\\
$(2,\ge2)$ & $2p_1+p_2<1$, $n\ge4$ & $a=p_1+p_2\in[\tfrac25,\tfrac23)$\\
$(\ge3,\ge2)$ & $2p_1+p_2<1$, $n\ge5$ & $a\le\tfrac35$ if $a\ge\tfrac12$\\
\bottomrule
\end{tabular}
\end{table}

\begin{proof}
By Lemma~\ref{lem:iso} the left child is $\{1\}$ if and only if
$2p_1+p_2\ge1$, and then $p_1\ge\tfrac13$. Let the left child have $k\ge2$
symbols. If $k=2$, then $a=p_1+p_2$. Since $p_2\le p_1$, we get
$\tfrac32a\le2p_1+p_2<1$, so $a<\tfrac23$. If also $a<\tfrac12$, then
Lemma~\ref{lem:key} gives $1-2a\le p_2\le\tfrac a2$, so $a\ge\tfrac25$. If
$k\ge3$ and $a\ge\tfrac12$, then Lemma~\ref{lem:key} and the ordering give
$2a-1\le p_k\le\tfrac ak\le\tfrac a3$, so $a\le\tfrac35$.

It remains to show that the right child is not a single symbol. If it
were, the cut would be after position $n-1$ with imbalance $1-2p_n$, and
Lemma~\ref{lem:key} would give $1-2p_n\le p_{n-1}\le p_1<\tfrac12$, so
$p_n>\tfrac14$. For $n\ge4$ all symbols would then exceed $\tfrac14$, and
the total mass would exceed $1$. For $n=3$ the cut after position $1$ has
imbalance $1-2p_1\le1-2p_3$. So the cut after position $1$ is chosen, by
the tie rule if $p_1=p_3$, and the right child has two symbols.
\end{proof}

For $\tfrac14\le p_1<\tfrac13$, Lemma~\ref{lem:iso} gives $\ell_1\ge2$,
and the split of the left child of the root separates symbol $1$, so
$\ell_1=2$ (Lemmas~\ref{lem:twolevel} and~\ref{lem:noburial-q},
Fig.~\ref{fig:depthtwo}).

\begin{figure}[t]
\centering
\begin{tikzpicture}[level distance=12mm, level 1/.style={sibling distance=40mm},
  level 2/.style={sibling distance=20mm}]
\node[inode] (r) {}
  child {node[inode] (L) {}
    child {node[leaf]{$p_1$}}
    child {node[sub]{$L'$\\ $x$}}}
  child {node[sub,minimum width=15mm]{$B_R$\\ $1-a$}};
\node[nlab,anchor=south west] at ($(r)+(1mm,1mm)$) {mass $1$,\\ term $1-h(a)$};
\node[nlab,anchor=south east] at ($(L)+(-1mm,1mm)$) {mass $a$,\\ term $a\bigl(1-h(\tfrac{p_1}a)\bigr)$};
\end{tikzpicture}
\caption{The tree of Lemma~\ref{lem:twolevel}. The left child of the root
has mass $a$, and its split separates symbol $1$. The rest $L'$ of the left child has
mass $x=a-p_1$, and the right child $B_R$ has mass $1-a$.}
\label{fig:depthtwo}
\end{figure}

\begin{lemma}\label{lem:twolevel}
Let the left child $B_L$ of the root have mass $a$ and at least two
symbols. Suppose that $2p_1+p_2\ge a$, so that the split of $B_L$
separates symbol $1$ (Lemma~\ref{lem:iso}). Then $\ell_1=2$. Let $L'=B_L\setminus\{1\}$, of mass
$x=a-p_1$. Then
\begin{equation}\label{eq:twolevel}
  R(\bp)\ =\ (1+a)-\Ent(p_1,\,x,\,1-a)+R_{L'}+R_{B_R} .
\end{equation}
\end{lemma}

\begin{proof}
Lemma~\ref{lem:rec} at the root and at $B_L$ gives
$R(\bp)=(1-h(a))+a\bigl(1-h(p_1/a)\bigr)+R_{L'}+R_{B_R}$, and by the
grouping identity \eqref{eq:grouping},
$(1-h(a))+a\bigl(1-h(p_1/a)\bigr)=(1+a)-\Ent(p_1,x,1-a)$.
\end{proof}

\begin{lemma}\label{lem:noburial-q}
Let $\tfrac14\le p_1<\tfrac13$. Then:
\begin{enumerate}
\item[\textup{(i)}] the left child of the root is $B_L=\{1,\dots,k\}$ with
$k\ge2$, and $s:=|1-2a|\le p_k$, where $a=M_{B_L}$;
\item[\textup{(ii)}] the split of $B_L$ separates symbol $1$. So $\ell_1=2$ and
\eqref{eq:twolevel} holds, where $L'=\{2,\dots,k\}$ has mass $x=a-p_1$,
largest symbol $p_2$ and smallest symbol $p_k$, and $B_R=\{k+1,\dots,n\}$
has mass $1-a$ and largest symbol $p_{k+1}\le p_k$.
\end{enumerate}
\end{lemma}

\begin{proof}
(i) Since $2p_1+p_2\le3p_1<1$, $\ell_1\ge2$
(Lemma~\ref{lem:iso}). So $k\ge2$. Lemma~\ref{lem:key} at the root
gives $s=|1-2a|\le p_k$.

(ii) By Lemma~\ref{lem:iso} applied to the block $B_L$, it suffices to
show $2p_1+p_2\ge a$. Suppose instead that $2p_1+p_2<a$, that is,
$p_2<a-2p_1$. By (i),
$s\le p_k\le p_2<a-2p_1$. If $a\le\tfrac12$, then $1-2a<a-2p_1$ gives
$a>\tfrac{1+2p_1}3\ge\tfrac12$, using $p_1\ge\tfrac14$. This is a
contradiction. If $a>\tfrac12$, then $2a-1<a-2p_1$ gives
$a<1-2p_1\le\tfrac12$, again a contradiction. So $2p_1+p_2\ge a$, and
Lemma~\ref{lem:twolevel} applies.
\end{proof}

\section{Main Results}\label{sec:main}

\subsection{Bound pieces}\label{ss:pieces}

The bound is built from the constants \eqref{eq:constants} and the seven
functions
\begin{subequations}\label{eq:pieces}
\begin{align}
  U_1(p_1)\ &:=\ p_1+\CC ,\label{eq:U1}\\
  U_2(p_1)\ &:=\ \tfrac32+\tfrac{\CC}2-p_1-\Ent\bigl(p_1,\tfrac12-p_1\bigr),\label{eq:U2}\\
  U_3(p_1)\ &:=\ \frac{4+2p_1}{3}+\frac{2(1-p_1)}{3}\,K\nonumber\\
  &\qquad-\Ent\Bigl(p_1,\frac{1-p_1}{3},\frac{2(1-p_1)}{3}\Bigr),\label{eq:U3}\\
  U_4(p_1)\ &:=\ 3-3p_1-\Ent(p_1,p_1,p_1,1-3p_1),\label{eq:U4}\\
  U_5(p_1)\ &:=\ 1-h(p_1)+(1-p_1)\,K ,\label{eq:U5}\\
  U_6(p_1)\ &:=\ \frac{5-p_1}{3}+\frac{2-p_1}{3}\,\CC\nonumber\\
  &\qquad-\Ent\Bigl(p_1,\frac{1-2p_1}{3},\frac{2-p_1}{3}\Bigr),\label{eq:U6}\\
  U_7(p_1)\ &:=\ 2-p_1-h(p_1) .\label{eq:U7}
\end{align}
\end{subequations}
The pieces change at $\tfrac14$, $\tfrac13$ and $\tfrac12$, where the top
of the tree changes (Lemmas~\ref{lem:iso} and~\ref{lem:noburial-q}), and
at the crossing points $b_1$, $b_2$, $b_3$ of adjacent pieces.

\begin{lemma}\label{lem:boundaries}
There are unique numbers $b_1,b_2\in(\tfrac14,\tfrac13)$ and
$b_3\in(\tfrac13,\tfrac12)$ with
\begin{equation}\label{eq:boundaries}
  U_2(b_1)=U_3(b_1),\quad U_3(b_2)=U_4(b_2),\quad U_5(b_3)=U_6(b_3),
\end{equation}
and $b_1<b_2$. In addition, $U_2\ge U_3$ on $[\tfrac14,b_1]$ and
$U_2\le U_3$ on $[b_1,\tfrac13]$. Also $U_3\ge U_4$ on $[\tfrac14,b_2]$ and
$U_3\le U_4$ on $[b_2,\tfrac13]$. Finally $U_5\ge U_6$ on
$[\tfrac13,b_3]$ and $U_5\le U_6$ on $[b_3,\tfrac12]$. To seven decimals,
\begin{equation}\label{eq:bvalues}
  b_1\approx0.3190251,\qquad b_2\approx0.3196324,\qquad b_3\approx0.4157982 .
\end{equation}
\end{lemma}

\begin{proof}
This follows from Lemmas~\ref{lem:cross} and~\ref{lem:cross-q}.
\end{proof}

\begin{table*}[t]
\renewcommand{\arraystretch}{1.25}
\caption{The piecewise bounds of Theorem~\ref{thm:main}. ``Exact'' means that
the piece equals the supremum of $R$ over all sources with the given $p_1$
(Theorem~\ref{thm:exact}). Values are rounded.}
\label{tab:pieces}
\centering
\begin{tabular}{@{}llcccl@{}}
\toprule
piece & interval & value at left end & value at right end & exact & proved in\\
\midrule
$U_1$ & $(0,\tfrac14)$ & $0.130120$ ($p_1\to0$) & $0.380120$ & no & Section~\ref{sec:affine}\\
$U_2$ & $[\tfrac14,b_1]$ & $0.315060$ & $0.273890$ & no & Section~\ref{sec:quarter}\\
$U_3$ & $[b_1,b_2]$ & $0.273890$ & $0.273961$ & no & Section~\ref{sec:quarter}\\
$U_4$ & $[b_2,\tfrac13)$ & $0.273961$ & $0.415037$ ($p_1\to\tfrac13$) & yes & Section~\ref{sec:quarter}\\
$U_5$ & $[\tfrac13,b_3]$ & $0.458411$ & $0.350664$ & no & Section~\ref{sec:midcap}\\
$U_6$ & $[b_3,\tfrac12)$ & $0.350664$ & $0.565060$ ($p_1\to\tfrac12$) & no & Section~\ref{sec:midcap}\\
$U_7$ & $[\tfrac12,1)$ & $0.500000$ & $1$ ($p_1\to1$) & yes & Section~\ref{sec:half}\\
\bottomrule
\end{tabular}
\end{table*}

\subsection{Main theorem}

\begin{theorem}\label{thm:main}
Let $\bp$ be a source with $n\ge2$ symbols and largest probability $p_1$.
Then
\begin{subequations}\label{eq:main}
\begin{alignat}{2}
  &\text{if }0<p_1<\tfrac14,\quad
    &R(\bp)&<U_1(p_1)=p_1+\CC ,\label{eq:main-a}\\
  &\text{if }\tfrac14\le p_1<\tfrac13,\quad
    &R(\bp)&<\begin{cases}
    U_2(p_1) & \text{if }\tfrac14\le p_1\le b_1,\\
    U_3(p_1) & \text{if }b_1\le p_1\le b_2,\\
    U_4(p_1) & \text{if }b_2\le p_1<\tfrac13,
    \end{cases}\label{eq:main-b}\\
  &\text{if }\tfrac13\le p_1<\tfrac12,\quad
    &R(\bp)&<\begin{cases}
    U_5(p_1) & \text{if }\tfrac13\le p_1\le b_3,\\
    U_6(p_1) & \text{if }b_3\le p_1<\tfrac12,
    \end{cases}\label{eq:main-c}\\
  &\text{if }\tfrac12\le p_1<1,\quad
    &R(\bp)&<U_7(p_1)=2-p_1-h(p_1) .\label{eq:main-d}
\end{alignat}
\end{subequations}
Since adjacent pieces agree at $b_1$, $b_2$ and $b_3$, it does not
matter which interval contains these points. We write $U(p_1)$ for the
right-hand side (Table~\ref{tab:pieces}).
\end{theorem}

\begin{corollary}\label{cor:kcap}
For every source with $p_1<\tfrac12$, $R(\bp)<K=\tfrac52-\tfrac56\lgg5$.
For every source, $R(\bp)<p_1+\CC$. For every $p_1\in(0,1)$,
$U(p_1)<1$.
\end{corollary}

\begin{theorem}\label{thm:exact}
For $t\in(0,1)$ let
\begin{equation}\label{eq:Rmax}
  \Rmax(t)\ :=\ \sup\bigl\{R(\bp):\ \bp\text{ a source with }p_1=t\bigr\},
\end{equation}
where the supremum is over all alphabet sizes. Then
\begin{equation}\label{eq:exact}
\begin{aligned}
  \Rmax(t)\ &=\ U_7(t)\ \text{ for }t\in[\tfrac12,1),\\
  \Rmax(t)\ &=\ U_4(t)\ \text{ for }t\in[b_2,\tfrac13),
\end{aligned}
\end{equation}
and neither supremum is attained. Also $\Rmax(t)\ge U_4(t)$ for all
$t\in[\tfrac14,\tfrac13)$, and $\sup_{\bp}R(\bp)=1$ over all sources.
\end{theorem}

\subsection{Proof of Theorem~\ref{thm:main}}\label{ss:mainproof}

Theorem~\ref{thm:main} combines four theorems that prove the bound in the four ranges of $p_1$.
\begin{itemize}
\item Theorem~\ref{thm:half} (Section~\ref{sec:half}): if $p_1\ge\tfrac12$,
  then $R<U_7(p_1)$. This is \eqref{eq:main-d}.
\item Theorem~\ref{thm:aff} (Section~\ref{sec:affine}): every source with
  $n\ge3$ symbols has $R\le p_1+\CC-\lam p_n$. So $R<p_1+\CC$ for every
  source (Corollary~\ref{cor:aff}). This gives \eqref{eq:main-a}.
\item Theorem~\ref{thm:mid} (Section~\ref{sec:midcap}): if
  $\tfrac13\le p_1<\tfrac12$, then $R<\max\bigl(U_5(p_1),U_6(p_1)\bigr)$.
  By Lemma~\ref{lem:boundaries} this maximum is the right side of
  \eqref{eq:main-c}.
\item Theorem~\ref{thm:quarter} (Section~\ref{sec:quarter}): if
  $\tfrac14\le p_1<\tfrac13$, then
  $R<\max\bigl(U_2(p_1),U_3(p_1),U_4(p_1)\bigr)$. By
  Lemma~\ref{lem:boundaries} this maximum is the right side of
  \eqref{eq:main-b}.
\end{itemize}
The proof of Theorem~\ref{thm:quarter} uses $b_1$ and $b_2$ only to write
the maximum piecewise (Lemma~\ref{lem:cross-q}). Corollary~\ref{cor:kcap}
follows from these four theorems and Proposition~\ref{prop:dom-q}. Its
proof is at the end of Section~\ref{sec:midcap}.
Theorem~\ref{thm:exact} follows from Theorem~\ref{thm:main} and the
sources of Section~\ref{sec:disc}.

On $[\tfrac14,\tfrac12)$ every piece is smaller than $U_1$
(Propositions~\ref{prop:dom} and~\ref{prop:dom-q}).

\subsection{Outline of the proofs}\label{ss:plan}

For $p_1\ge\tfrac12$ the bound follows from Corollary~\ref{cor:isob}.
Each of the other proofs writes the redundancy by Lemma~\ref{lem:rec} as
the terms of the top one or two levels of the tree plus the redundancies
of the blocks below them. It bounds the redundancy of these blocks by an
induction hypothesis or by results proved earlier.

Theorem~\ref{thm:aff}, which gives $R<p_1+\CC$ for every source, is
proved by induction on the number of symbols, in which the induction
hypothesis bounds the redundancy of the children of the root. Let the
left child of the root be $\{1,\dots,k\}$, of mass $a$. If both children have at least
two symbols, the hypothesis $R<p_1+\CC$ for the children bounds $R$ by
$1-h(a)+p_1+p_{k+1}+\CC$, where $p_{k+1}$ is the largest probability of
the right child, and this exceeds $p_1+\CC$. So the induction is carried
out for the stronger bound $R\le p_1+\CC-\lam p_n$. Applied to the left
child, this bound contains the term $-\lam p_k$. When the left child has at least three symbols, this
term reduces the induction step to $1-h(a)+p_{k+1}\le\lam p_k$. In one
case the hypothesis is not strong enough, and it is applied to the
children of a child instead (Section~\ref{sec:onethree}).

Theorem~\ref{thm:mid} proves the bounds $U_5$ and $U_6$ on
$[\tfrac13,\tfrac12)$. It bounds the children of
the root by Theorem~\ref{thm:aff}, by Corollary~\ref{cor:isob} and by
an induction hypothesis.

Theorem~\ref{thm:quarter} proves the bounds $U_2$, $U_3$, $U_4$ on
$[\tfrac14,\tfrac13)$. In this range
Lemma~\ref{lem:noburial-q} determines the first two levels of the tree,
so the children and grandchildren of the root can be bounded by
Corollary~\ref{cor:isob} and Theorems~\ref{thm:aff} and~\ref{thm:mid}.

Fig.~\ref{fig:dep} shows a dependency graph of the results.

\begin{figure*}[t]
\centering
\begin{tikzpicture}[>=Latex,
  box/.style={draw,rounded corners=2pt,align=center,font=\footnotesize,
    inner sep=3pt, minimum height=10mm, text width=30mm}]
\node[box] (in)   at (0,0)     {Sections~\ref{sec:model} and~\ref{sec:struct}\\ with Corollary~\ref{cor:isob}\\ and Lemma~\ref{lem:univ}};
\node[box] (aff)  at (4.3,1.1)  {$R<p_1+\CC$\\ Theorem~\ref{thm:aff}};
\node[box] (half) at (4.3,-1.1) {$p_1\ge\tfrac12$\\ Theorem~\ref{thm:half}};
\node[box] (mid)  at (8.6,0)    {$\tfrac13\le p_1<\tfrac12$\\ Theorem~\ref{thm:mid}};
\node[box] (quar) at (12.9,0)   {$\tfrac14\le p_1<\tfrac13$\\ Theorem~\ref{thm:quarter}};
\draw[->] (in.east) -- (aff.west);
\draw[->] (in.east) -- (half.west);
\draw[->] (aff.east) -- (mid.west);
\draw[->] (in.east) -- (mid.west);
\draw[->] (mid.east) -- (quar.west);
\draw[->] (aff.north) -- ++(0,0.5) -| (quar.north);
\draw[->] (in.south) |- (8.6,-2.4) -| (quar.south);
\end{tikzpicture}
\caption{Dependencies among the results. An arrow from one result to
another means that the proof of the second uses the first. Theorem~\ref{thm:main}
combines Theorems~\ref{thm:half} to~\ref{thm:quarter}.}
\label{fig:dep}
\end{figure*}

\section{Largest Probability at Least One Half}\label{sec:half}

For $p_1\ge\tfrac12$ we have $\ell_1=1$, so the bound follows from
Corollary~\ref{cor:isob}.

\begin{theorem}\label{thm:half}
If $2p_1+p_2\ge1$, in particular if $p_1\ge\tfrac12$, then
$R(\bp)<U_7(p_1)=2-p_1-h(p_1)$.
\end{theorem}

\begin{proof}
This is Corollary~\ref{cor:isob}.
\end{proof}

\begin{proposition}\label{prop:tight}\leavevmode
\begin{enumerate}
\item[\textup{(a)}] For every $t\in[\tfrac12,1)$, $\Rmax(t)=U_7(t)$, and
the supremum is not attained.
\item[\textup{(b)}] $\sup_{\bp}R(\bp)=1$ over all sources, and the
supremum is not attained.
\end{enumerate}
\end{proposition}

\begin{proof}
(a) The upper bound is Theorem~\ref{thm:half}. For the lower bound we
give sources whose redundancy tends to $U_7(t)$. Fix $t\in[\tfrac12,1)$
and $w\in[\tfrac12,1)$, and consider the three-symbol source
\begin{equation}\label{eq:fam7}
  \bp\ =\ \bigl(t,\ (1-t)w,\ (1-t)(1-w)\bigr).
\end{equation}
It is sorted because $(1-t)w\le1-t\le t$ and $1-w\le w$. Since
$t\ge\tfrac12$, we have $\ell_1=1$ (Lemma~\ref{lem:iso}), and the
remaining pair contributes $(1-t)\bigl(1-h(w)\bigr)$ by \eqref{eq:rec}.
So
\begin{equation}\label{eq:fam7R}
\begin{aligned}
  R(\bp)\ &=\ \bigl(1-h(t)\bigr)+(1-t)\bigl(1-h(w)\bigr)\\
  &\longrightarrow\ 2-t-h(t)\qquad(w\to1^-).
\end{aligned}
\end{equation}
The supremum is not attained because $R<U_7(t)$ for every source
(Theorem~\ref{thm:half}).

(b) By Lemma~\ref{lem:univ}, $R<1$. The source $(1-\varepsilon,\varepsilon)$
has $R=1-h(\varepsilon)$, which tends to $1$ as $\varepsilon\to0^+$.
\end{proof}

\section{The Affine Bound}\label{sec:affine}

\subsection{Statement}

\begin{theorem}\label{thm:aff}
For every source with $n\ge3$ symbols,
\begin{equation}\label{eq:Psi}
  R(\bp)\ \le\ p_1+\CC-\lam\,p_n ,
\end{equation}
with equality if and only if $\bp=(\tfrac13,\tfrac13,\tfrac13)$.
\end{theorem}

\begin{corollary}\label{cor:aff}
For every source with $n\ge2$ symbols, $R(\bp)<p_1+\CC$. In particular
\eqref{eq:main-a} holds.
\end{corollary}

\begin{proof}[Proof of the corollary]
For $n\ge3$, Theorem~\ref{thm:aff} gives $R\le p_1+\CC-\lam p_n<p_1+\CC$,
since $\lam>0$ and $p_n>0$. For $n=2$ we have $p_1\ge\tfrac12$, and
\begin{equation}\label{eq:aff-n2}
  R\ \overset{(a)}{=}\ 1-h(p_1)\ \overset{(b)}{\le}\ 2p_1-1
  \ \overset{(c)}{<}\ p_1\ \overset{(d)}{<}\ p_1+\CC ,
\end{equation}
where (a) follows from \eqref{eq:rec}, (b) from \eqref{eq:chord}, (c)
from $p_1<1$, and (d) from $\CC>0$.
\end{proof}

\begin{example}\label{ex:unif3}
Let $\bp=(\tfrac13,\tfrac13,\tfrac13)$. The cuts after positions $1$ and
$2$ both have imbalance $\tfrac13$. The tie rule selects the cut after
position $1$, so $\ell_1=1$, and the remaining pair is
balanced. Then $R=1-h(\tfrac13)$. The right side of \eqref{eq:Psi} equals
$p_1+\CC-\tfrac\lam3=\tfrac13+(1-h(\tfrac13))+\tfrac{\lam-1}3-\tfrac\lam3
=1-h(\tfrac13)$ by \eqref{eq:C}. So equality holds in \eqref{eq:Psi}.
\end{example}

\subsection{Bounds on the children of a block}\label{ss:children}

Theorem~\ref{thm:aff} is proved by strong induction on $n$, over all
sources with $n\ge3$ symbols. The \emph{induction hypothesis} (IH) is the full
statement of Theorem~\ref{thm:aff}, including the case of equality, for
all sources with $n'$ symbols, $3\le n'<n$.

\begin{lemma}\label{lem:children}
Assume the IH. Let $B=\{i,\dots,j\}$ be a block of a source with $n$
symbols, with $|B|<n$. If $|B|=1$, then $R_B=0$. If $|B|=2$, then
\begin{subequations}\label{eq:children}
\begin{equation}\label{eq:child-pair}
  R_B\ \le\ p_i-p_j ,
\end{equation}
with equality if and only if $p_i=p_j$. If $|B|\ge2$, then
\begin{equation}\label{eq:child-gen}
  R_B\ \le\ p_i+\CC\,M_B-\lam\,p_j ,
\end{equation}
\end{subequations}
with equality only if $|B|=3$ and $\bp^{(B)}=(\tfrac13,\tfrac13,\tfrac13)$.
\end{lemma}

\begin{proof}
The case $|B|=1$ is \eqref{eq:RBsingle}. For $|B|=2$,
\begin{equation}\label{eq:pairchain}
\begin{aligned}
  R_B\ &\overset{(a)}{=}\ M_B\Bigl(1-h\Bigl(\frac{p_i}{M_B}\Bigr)\Bigr)\\
  &\overset{(b)}{\le}\ M_B\,\Bigl|\frac{2p_i}{M_B}-1\Bigr|
  \ \overset{(c)}{=}\ p_i-p_j ,
\end{aligned}
\end{equation}
where (a) follows from \eqref{eq:rec} and \eqref{eq:RBsingle}, (b) from
\eqref{eq:chord}, and (c) from $M_B=p_i+p_j$ and $p_i\ge p_j$. Equality
requires $p_i/M_B\in\{\tfrac12,1\}$, and the value $1$ is unattainable
because $p_j>0$. For $|B|=2$, \eqref{eq:child-gen} follows from \eqref{eq:child-pair},
from $p_j\le\tfrac12M_B$ and from \eqref{eq:constfacts}, since
\begin{equation}\label{eq:pairform}
\begin{aligned}
  R_B+\lam p_j\ &\le\ p_i+(\lam-1)\,p_j\\
  &\le\ p_i+\tfrac{\lam-1}2\,M_B\ <\ p_i+\CC\,M_B .
\end{aligned}
\end{equation}
Let $|B|\ge3$. The source $\bp^{(B)}$ has fewer than $n$ symbols,
largest probability $p_i/M_B$ and smallest probability $p_j/M_B$, so the
IH applied to it gives $R(\bp^{(B)})\le p_i/M_B+\CC-\lam p_j/M_B$.
Multiplying by $M_B$,
\begin{equation}\label{eq:subst}
\begin{aligned}
  R_B\ &=\ M_B\,R\bigl(\bp^{(B)}\bigr)
  \ \le\ M_B\Bigl(\frac{p_i}{M_B}+\CC-\lam\frac{p_j}{M_B}\Bigr)\\
  &=\ p_i+\CC\,M_B-\lam\,p_j ,
\end{aligned}
\end{equation}
with equality only in the case of equality of the IH.
\end{proof}

\subsection{Scalar inequalities}

The induction uses the following three inequalities.

\begin{lemma}\label{lem:scalar}\leavevmode
\begin{enumerate}
\item[\textup{(S1)}] The function $y\mapsto1-h(y)+(\lam-1)(1-2y)$ is
strictly decreasing on $[\tfrac13,\tfrac12]$, and its value at
$y=\tfrac13$ is $\CC$.
\item[\textup{(S2)}] For $y\in[\tfrac12,1]$,
\begin{equation}\label{eq:S2}
  \bigl(2-2y-h(y)\bigr)+\tfrac{\lam-1}{2}(1-y)\ \le\ \tfrac{\lam-1}{4}\ <\ \CC .
\end{equation}
\item[\textup{(S3)}] For $a\in[\tfrac12,\tfrac23]$,
\begin{equation}\label{eq:S3}
  1-h(a)\ \le\ 6\,\hk\,\bigl(a-\tfrac12\bigr).
\end{equation}
\end{enumerate}
\end{lemma}

\begin{proof}
(S1) The derivative is $\lgg\tfrac{y}{1-y}-2(\lam-1)$, which is negative
for $y\le\tfrac12$ because the logarithm is at most $0$ and $\lam>1$. At
$y=\tfrac13$ the value is $(1-h(\tfrac13))+\tfrac{\lam-1}3=\CC$ by
\eqref{eq:C}.

(S2) The left side is an affine function minus $h(y)$, so it is convex.
Its maximum on $[\tfrac12,1]$ is therefore at an endpoint. The values are $\tfrac{\lam-1}4$ at
$y=\tfrac12$ and $0$ at $y=1$. By \eqref{eq:lamrange} and (N4),
$\tfrac{\lam-1}4\le\tfrac1{20}<\tfrac1{10}<\CC$.

(S3) Since $1-h$ is convex, on $[\tfrac12,\tfrac23]$ it lies below its
chord through $(\tfrac12,0)$ and $(\tfrac23,1-h(\tfrac23))$. Since
$h(\tfrac23)=h(\tfrac13)$, the chord has slope $6(1-h(\tfrac13))$.
\end{proof}

\subsection{The induction}\label{ss:induct}

\noindent\textit{Proof of Theorem~\ref{thm:aff}.} Let $\bp$ have $n\ge3$ symbols, and
assume the IH. In this proof write $\Psi:=p_1+\CC-\lam p_n$ for the
bound to be proved. Let the root have left child $B_L$ of mass $a$ and
right child $B_R$, let $p_k$ be the smallest symbol of $B_L$, and let
\begin{equation}\label{eq:E}
  E\ :=\ 1-h(a)
\end{equation}
be the term of the root. By \eqref{eq:chord}, \eqref{eq:key} and
\eqref{eq:quad},
\begin{equation}\label{eq:Ebounds}
  E\ \le\ |2a-1|\ \le\ p_k,\qquad E\ \le\ (2a-1)^2 .
\end{equation}
By \eqref{eq:rec}, $R=E+R_{B_L}+R_{B_R}$. Table~\ref{tab:cases} lists
the cases.

\begin{table*}[t]
\renewcommand{\arraystretch}{1.25}
\caption{The cases of the induction. The second column gives the bounds used
for the two children.}
\label{tab:cases}
\centering
\begin{tabular}{@{}llll@{}}
\toprule
case & bounds for $B_L$, $B_R$ & $R-\Psi$ is at most & proof uses\\
\midrule
$p_1\ge\tfrac12$ & single symbol; Lemma~\ref{lem:univ} & $2-2p_1-h(p_1)+(\lam-1)p_n-\CC$ & (S2)\\
$(1,2)$ & single symbol; \eqref{eq:child-pair} & $1-h(p_1)+p_2+(\lam-1)p_3-p_1-\CC$ & (S1)\\
$(2,\ge2)$ & \eqref{eq:child-pair}; \eqref{eq:child-gen} & $E-a\CC$ & \eqref{eq:quad}, (S3)\\
$(\ge3,\ge2)$ & \eqref{eq:child-gen}; \eqref{eq:child-gen} & $E+p_{k+1}-\lam p_k$ & \eqref{eq:ratio}\\
$(1,\ge3)$ & single symbol; next split & see Section~\ref{sec:onethree} & Lemmas~\ref{lem:caseA}, \ref{lem:B2}\\
\bottomrule
\end{tabular}
\end{table*}

\medskip\noindent\textit{Case $p_1\ge\tfrac12$.}
Here $\ell_1=1$ (Lemma~\ref{lem:iso}), and the right child
$\{2,\dots,n\}$ has $n-1\ge2$ symbols and mass $1-p_1$. Its rescaled
source has smallest probability $p_n/(1-p_1)$, so Lemma~\ref{lem:univ}
gives $R_{B_R}\le(1-p_1)-p_n$. Hence
\begin{equation}\label{eq:dom-case}
\begin{aligned}
  R-\Psi\ &\overset{(a)}{\le}\ \bigl(1-h(p_1)\bigr)+(1-p_1)-p_n-\Psi\\
  &=\ \bigl(2-2p_1-h(p_1)\bigr)+(\lam-1)p_n-\CC\\
  &\overset{(b)}{\le}\ \bigl(2-2p_1-h(p_1)\bigr)+\tfrac{\lam-1}2(1-p_1)-\CC\\
  &\overset{(c)}{<}\ 0 ,
\end{aligned}
\end{equation}
where (a) follows from \eqref{eq:isorec} and the bound on $R_{B_R}$, (b)
from $p_n\le\tfrac12(1-p_1)$ and $\lam>1$, and (c) from (S2). In the remaining cases $p_1<\tfrac12$, and the root split is one
of those of Lemma~\ref{lem:shapes}.

\medskip\noindent\textit{Case $(1,2)$, three symbols.}
Here $n=3$ and $R\le(1-h(p_1))+(p_2-p_3)$, so
\begin{equation}\label{eq:case12}
  R-\Psi\ \le\ \bigl(1-h(p_1)\bigr)+p_2+(\lam-1)p_3-p_1-\CC .
\end{equation}
After substituting $p_3=1-p_1-p_2$, the terms in $p_2$ and $p_3$ become
$(2-\lam)p_2+(\lam-1)(1-p_1)$, which increases with $p_2$ because
$\lam<2$ by \eqref{eq:lamrange}. So the right side of \eqref{eq:case12}
is largest at $p_2=p_1$, where $p_3=1-2p_1$ and the right side becomes
$1-h(p_1)+(\lam-1)(1-2p_1)-\CC$. By (S1) and $p_1\ge\tfrac13$, this is at
most $0$. Equality requires $p_1=\tfrac13$, $p_2=p_1$, and equality in
\eqref{eq:child-pair}, that is, $p_2=p_3$. Equality holds for this
uniform source on three symbols (Example~\ref{ex:unif3}). The remaining
cases are strict, so this is the only case of equality.

\medskip\noindent\textit{Case $(2,\ge2)$.}
Here $a=p_1+p_2\in[\tfrac25,\tfrac23)$ by Lemma~\ref{lem:shapes}(ii). The
bound \eqref{eq:child-pair} for $B_L$ and \eqref{eq:child-gen} for $B_R$
give $R\le E+(p_1-p_2)+\bigl(p_3+(1-a)\CC-\lam p_n\bigr)$. So
$R-\Psi\le E-p_2+p_3-a\CC\le E-a\CC$, since $p_3\le p_2$. It suffices to
show $E<a\CC$. On $[\tfrac25,\tfrac12]$ we have $E\le(2a-1)^2$, and
$(2a-1)^2/a$ is decreasing there, because its derivative has the sign of
$2a-1$. So $E/a\le\tfrac{(1/5)^2}{2/5}=\tfrac1{10}<\CC$ by (N4). On
$[\tfrac12,\tfrac23)$,
\begin{equation}\label{eq:case22b}
  \frac Ea\ \overset{(a)}{\le}\ 6\,\hk\Bigl(1-\frac1{2a}\Bigr)
  \ \overset{(b)}{<}\ \tfrac32\,\hk\ \overset{(c)}{<}\ \CC ,
\end{equation}
where (a) follows from (S3), (b) from $a<\tfrac23$, and (c) from
\eqref{eq:constfacts}. This gives $R<\Psi$.

\medskip\noindent\textit{Case $(\ge3,\ge2)$.}
This is the main case. Both children are bounded by \eqref{eq:child-gen}, which gives
$R\le E+\bigl(p_1+a\CC-\lam p_k\bigr)+\bigl(p_{k+1}+(1-a)\CC-\lam p_n\bigr)$.
The terms $p_1$, $\CC$ and $\lam p_n$ cancel against $\Psi$, and
\begin{equation}\label{eq:case33}
  R-\Psi\ \le\ E+p_{k+1}-\lam\,p_k .
\end{equation}
Since $p_{k+1}\le p_k$, it suffices to show
\begin{equation}\label{eq:E-target}
  E\ \le\ (\lam-1)\,p_k .
\end{equation}
The left child has at least three
symbols of total mass $a$, so $p_k\le\tfrac a3$. If $a\ge\tfrac12$, then
$a\le\tfrac35$ by Lemma~\ref{lem:shapes}(iii), and $p_k\le\tfrac15$. If
$a<\tfrac12$, then $p_k\le\tfrac a3<\tfrac16$. In both cases
$|2a-1|\le p_k\le\tfrac15$, so the monotonicity in
Lemma~\ref{lem:kappa}(iii), followed by \eqref{eq:ratio}, gives
\begin{equation}\label{eq:case33b}
\begin{aligned}
  E\ &=\ 1-h\Bigl(\frac{1+|2a-1|}2\Bigr)\ \le\ 1-h\Bigl(\frac{1+p_k}2\Bigr)\\
  &\le\ (\lam-1)\,p_k .
\end{aligned}
\end{equation}
This proves \eqref{eq:E-target}. With $p_{k+1}\le p_k$ it gives
$E+p_{k+1}\le\lam p_k$, so $R\le\Psi$.

Suppose $R=\Psi$ in this case. Then every estimate above is an equality.
Equality in the last step of \eqref{eq:case33b} with $0<p_k\le\tfrac15$
forces $p_k=\tfrac15$, because the quotient in Lemma~\ref{lem:kappa}(iii)
is strictly increasing. This rules out $a<\tfrac12$, where $p_k<\tfrac16$.
So $a\ge\tfrac12$, and equality in the first step forces $2a-1=\tfrac15$,
that is, $a=\tfrac35$. The left child then consists of at least three
symbols, each at least $\tfrac15$, of total mass $\tfrac35$, so it is
$(\tfrac15,\tfrac15,\tfrac15)$, and $p_{k+1}=p_k=\tfrac15$. Finally,
equality in the bound for $B_R$ requires $|B_R|=3$ and
$\bp^{(B_R)}=(\tfrac13,\tfrac13,\tfrac13)$ by Lemma~\ref{lem:children}.
Then $p_{k+1}=\tfrac{1-a}3=\tfrac2{15}$, which contradicts
$p_{k+1}=\tfrac15$. This shows $R<\Psi$.

\subsection{The case $(1,\ge3)$}\label{sec:onethree}

This is the remaining case of the induction. We have $\ell_1=1$,
the right child $\{2,\dots,n\}$ has at least three symbols and mass
$1-p_1$, and $p_1\in[\tfrac13,\tfrac12)$. In this subsection let
$q:=p_2/(1-p_1)$ be the largest probability of the rescaled right child.
Fig.~\ref{fig:cases} shows the three cases. In Case A the split of the
child separates symbol $2$. In Cases B1 and B2 it does not, and
$q\le\tfrac13$ or $q>\tfrac13$, respectively. If $q\ge\tfrac12$, the split of the child separates symbol $2$ by Lemma~\ref{lem:iso}. So
$q<\tfrac12$ is assumed in Cases B1 and B2.

\begin{figure*}[t]
\centering
\begin{tikzpicture}[level distance=9mm, level 1/.style={sibling distance=22mm},
  level 2/.style={sibling distance=22mm}, level 3/.style={sibling distance=15mm}]
\begin{scope}[xshift=0cm]
\node[inode] {}
  child {node[leaf]{$p_1$}}
  child {node[inode] (a) {}
    child {node[leaf]{$p_2$}}
    child {node[sub]{$F$\\ $M_F$}}};
\node[nlab,anchor=west] at ($(a)+(1.5mm,0)$) {$1-p_1$};
\node[align=center,font=\footnotesize] at (0.6,-4.7) {Case A\\ the child separates symbol $2$};
\end{scope}
\begin{scope}[xshift=5.4cm]
\node[inode] {}
  child {node[leaf]{$p_1$}}
  child {node[inode] (b) {}
    child {node[sub]{$\ni p_2$}}
    child {node[sub]{}}};
\node[nlab,anchor=west] at ($(b)+(1.5mm,0)$) {$1-p_1$};
\node[align=center,font=\footnotesize] at (0.6,-4.7) {Case B1\\ $\tfrac{p_2}{1-p_1}\le\tfrac13$};
\end{scope}
\begin{scope}[xshift=10.6cm]
\node[inode] {}
  child {node[leaf]{$p_1$}}
  child {node[inode] (c) {}
    child[sibling distance=34mm] {node[inode] (d) {}
      child[sibling distance=17mm] {node[leaf]{$p_2$}}
      child[sibling distance=17mm] {node[sub,minimum width=15mm]{$A-p_2$}}}
    child[sibling distance=34mm] {node[sub,minimum width=17mm]{$F$\\ $1-p_1-A$}}};
\node[nlab,anchor=west] at ($(c)+(1.5mm,0)$) {$1-p_1$};
\node[nlab,anchor=east] at ($(d)+(-1.5mm,0)$) {$J$, mass $A$};
\node[align=center,font=\footnotesize] at (0.6,-4.7) {Case B2\\ $\tfrac13<\tfrac{p_2}{1-p_1}<\tfrac12$};
\end{scope}
\end{tikzpicture}
\caption{The three cases of Section~\ref{sec:onethree}. They differ in the
split of the child $\{2,\dots,n\}$ of mass $1-p_1$. Vertex labels give masses, except that $\ni p_2$ marks the
subtree that contains $p_2$.}
\label{fig:cases}
\end{figure*}

\medskip\noindent\textit{Case A: the split of the child separates symbol $2$.}
Let $F:=\{3,\dots,n\}$,
of mass $M_F=1-p_1-p_2$ and with at least two symbols.
Lemma~\ref{lem:rec} at the root and at the child, with the grouping
identity $h(p_1)+(1-p_1)\,h\bigl(p_2/(1-p_1)\bigr)=\Ent(p_1,p_2,M_F)$, gives
\begin{equation}\label{eq:2level}
  R\ =\ \bigl(2-p_1-\Ent(p_1,p_2,M_F)\bigr)+R_F .
\end{equation}
Then \eqref{eq:child-gen} gives $R_F\le p_3+M_F\CC-\lam p_n$.
So $R-\Psi\le G(p_1,p_2,p_3)$, where
\begin{equation}\label{eq:G}
\begin{aligned}
  G(p_1,p_2,p_3)\ &:=\ 2-2p_1-\Ent(p_1,p_2,M_F)\\
  &\qquad+p_3-\CC\,(p_1+p_2),\\
  M_F\ &=\ 1-p_1-p_2 .
\end{aligned}
\end{equation}

\begin{lemma}[Case A]\label{lem:caseA}
On the region
$\{\tfrac13\le p_1\le\tfrac12$,\ $1-2p_1\le p_2\le p_1$,\
$0<p_3\le\min(p_2,M_F)\}$,
\begin{equation}\label{eq:caseA}
  G(p_1,p_2,p_3)\ \le\ \tfrac53-\lgg3-\tfrac23\CC\ <\ 0 .
\end{equation}
\end{lemma}

\begin{proof}
The function $G$ increases with $p_3$, so it suffices to bound $G$ at
$p_3=\min(p_2,M_F)$. The two arguments of the minimum are equal at
$p_2=\tfrac{1-p_1}2$. This point lies in $[1-2p_1,p_1]$ because
$p_1\ge\tfrac13$. On each of the two resulting intervals of $p_2$, $G$ is
an affine function minus $\Ent$ of affine arguments, so $G$ is
convex in $p_2$ (Lemma~\ref{lem:endpoint}). Hence $G$ is at most its
largest value at $p_2\in\{1-2p_1,\ \tfrac{1-p_1}2,\ p_1\}$. Each of these
values is a convex function of $p_1$, so it is at most the larger of its
values at $p_1=\tfrac13$ and $p_1=\tfrac12$:
\begin{itemize}
\item $p_2=1-2p_1$, so $M_F=p_1$ and $p_3=1-2p_1$. The value is
  $3-4p_1-\Ent(p_1,p_1,1-2p_1)-\CC(1-p_1)$. It equals
  $\tfrac53-\lgg3-\tfrac23\CC$ at $p_1=\tfrac13$ and $-\tfrac{\CC}2$ at
  $p_1=\tfrac12$.
\item $p_2=\tfrac{1-p_1}2$, so $p_3=M_F=p_2$. The value is
  $2-2p_1-\Ent\bigl(p_1,\tfrac{1-p_1}2,\tfrac{1-p_1}2\bigr)+\tfrac{1-p_1}2-\CC\,\tfrac{1+p_1}2$.
  It equals $\tfrac53-\lgg3-\tfrac23\CC$ at $p_1=\tfrac13$, where all three
  symbols are $\tfrac13$, and $-\tfrac14-\tfrac34\CC$ at $p_1=\tfrac12$.
\item $p_2=p_1$, so $p_3=1-2p_1$. The value is
  $3-4p_1-\Ent(p_1,p_1,1-2p_1)-2\CC p_1$. This is at most the first
  function, since $2p_1\ge1-p_1$ for $p_1\ge\tfrac13$.
\end{itemize}
The two values at $p_1=\tfrac12$ are at most the value at
$p_1=\tfrac13$. Indeed $-\tfrac14-\tfrac34\CC<-\tfrac{\CC}2$, and
$\tfrac53-\lgg3-\tfrac23\CC+\tfrac{\CC}2=\tfrac53-\lgg3-\tfrac{\CC}6>0$ by
(N5). So $G\le\tfrac53-\lgg3-\tfrac23\CC$ on the region. This number is
negative by the second inequality in \eqref{eq:constfacts}.
\end{proof}

\medskip\noindent\textit{Case B1: $q\le\tfrac13$.}
The IH applies to the child, which has at least three symbols. It gives
$R_{\{2,\dots,n\}}\le p_2+(1-p_1)\CC-\lam p_n$. With
$p_2\le\tfrac{1-p_1}3$,
\begin{equation}\label{eq:caseB1}
\begin{aligned}
  R-\Psi\ &\le\ \bigl(1-h(p_1)\bigr)+p_2-p_1-\CC p_1\\
  &\le\ \bigl(1-h(p_1)\bigr)+\tfrac{1-p_1}3-(1+\CC)\,p_1 .
\end{aligned}
\end{equation}
The right side has derivative $\lgg\tfrac{p_1}{1-p_1}-\tfrac43-\CC<0$ on
$[\tfrac13,\tfrac12)$, because the logarithm is negative there. So it is
at most its value at $p_1=\tfrac13$, which is
$(1-h(\tfrac13))-\tfrac19-\tfrac{\CC}3<0$ by (N6).

\medskip\noindent\textit{Case B2: $\tfrac13<q<\tfrac12$.}
Let the child split into $J=\{2,\dots,j\}$, $j\ge3$, of mass $A$, and
$F=\{j+1,\dots,n\}$, of mass $1-p_1-A$.

(i) \emph{The split of $J$ separates symbol $2$.} Lemma~\ref{lem:key} for the block
$\{2,\dots,n\}$ gives $|2A-(1-p_1)|\le p_j$. Since $q>\tfrac13$ means
$1-p_1<3p_2$,
\begin{equation}\label{eq:B2i}
  A\ \le\ \tfrac{1-p_1+p_j}2\ \le\ \tfrac{1-p_1+p_2}2\ <\ 2p_2 .
\end{equation}
So $p_2/A>\tfrac12$, and the split of $J$ separates symbol $2$ by Lemma~\ref{lem:iso}.
Corollary~\ref{cor:isob} for the block $J$ and the grouping identity
give
\begin{equation}\label{eq:B2ib}
\begin{aligned}
  &R_J\ \le\ 2A-p_2-A\,h\bigl(\tfrac{p_2}A\bigr),\\
  &\Ent(p_1,A,1{-}p_1{-}A)+A\,h\bigl(\tfrac{p_2}A\bigr)\\
  &\qquad=\Ent(p_1,p_2,A{-}p_2,1{-}p_1{-}A).
\end{aligned}
\end{equation}

(ii) \emph{$F$ is not a single symbol.} Suppose $F=\{n\}$. Then the child
splits into $\{2,\dots,n-1\}$ and $\{n\}$. The cut after symbol $2$ is available, because
$n-1\ge3$. Its imbalance $1-p_1-2p_2$ is at most $1-p_1-2p_n$, the
imbalance of the cut after symbol $n-1$, and it lies further left. So the
tie rule does not select the cut after symbol $n-1$, a contradiction.
Therefore $|F|\ge2$, and
\eqref{eq:child-gen} gives
$R_F+\lam p_n\le p_{j+1}+(1-p_1-A)\CC$.

(iii) Lemma~\ref{lem:rec} at the root and at the child gives
$R=\bigl(2-p_1-\Ent(p_1,A,1-p_1-A)\bigr)+R_J+R_F$, as in
\eqref{eq:2level}. With (i) and (ii), $R\le\Psi$ follows from
$\Delta\le0$, where
\begin{equation}\label{eq:B2}
\begin{aligned}
  \Delta\ :=\ {}&2-2p_1+2A-p_2-\Ent\bigl(p_1,p_2,A{-}p_2,1{-}p_1{-}A\bigr)\\
  &+p_{j+1}-\CC\,(p_1+A).
\end{aligned}
\end{equation}

\begin{lemma}[Case B2]\label{lem:B2}
$\Delta<0$ for every tuple $(p_1,p_2,A,p_j,p_{j+1})$ that satisfies the
constraints of Table~\ref{tab:P}.
\end{lemma}

\begin{table}[t]
\renewcommand{\arraystretch}{1.25}
\caption{The constraints of Lemma~\ref{lem:B2}.}
\label{tab:P}
\centering
\footnotesize
\setlength{\tabcolsep}{4pt}
\begin{tabular}{@{}ll>{\raggedright\arraybackslash}p{0.43\columnwidth}@{}}
\toprule
label & constraint & origin\\
\midrule
(P1) & $\tfrac13\le p_1\le\tfrac12$ & case $(1,\ge3)$ has $p_1\in[\tfrac13,\tfrac12)$\\
(P2) & $2p_1+p_2\ge1$ & $\ell_1=1$ (Lemma~\ref{lem:iso})\\
(P3) & $\tfrac{1-p_1}3\le p_2\le\tfrac{1-p_1}2$ & $\tfrac13<q<\tfrac12$\\
(P4) & $0\le p_j\le A-p_2$ & $p_j$ is a term of $A-p_2=p_3+\cdots+p_j$\\
(P5) & $|2A-(1-p_1)|\le p_j$ & Lemma~\ref{lem:key} for the block $\{2,\dots,n\}$\\
(P6) & $0\le p_{j+1}\le p_j$ & sorted order\\
\bottomrule
\end{tabular}
\end{table}

\begin{proof}
Put $z:=A-p_2$ and
$v:=1-p_1-p_2-z=1-p_1-A$. By (P4) and (P5),
$2A-(1-p_1)\le p_j\le z$, which gives $v\ge p_2$. Also
$(1-p_1)-2A\le p_j\le z$, which gives $v\le p_2+2z$. Since $\Delta$
increases with $p_{j+1}$, and $p_{j+1}\le p_j\le z$ by (P6) and (P4),
\begin{equation}\label{eq:Dpoly}
\begin{aligned}
  \Delta\ \le\ \Phi(p_1,p_2,z)\ &:=\ 2-2p_1+p_2+3z-\Ent(p_1,p_2,z,v)\\
  &\qquad-\CC\,(p_1+p_2+z),\\
  v&=1-p_1-p_2-z .
\end{aligned}
\end{equation}

Constraints (P1) to (P3) and the two
consequences of (P5) give
\begin{equation}\label{eq:polytope}
\begin{gathered}
  p_1\le\tfrac12,\qquad 2p_1+p_2\ge1,\\
  p_1+3p_2\ge1,\qquad p_1+2p_2\le1,\\
  \tfrac13(1-p_1-2p_2)\ \le\ z\ \le\ 1-p_1-2p_2 .
\end{gathered}
\end{equation}
Every tuple that comes from a source satisfies \eqref{eq:polytope}, so it
suffices to bound $\Phi$ on the set defined by \eqref{eq:polytope}.

The function $\Phi$ is an affine function
minus $\Ent(p_1,p_2,z,v)$, a sum of entropy terms of affine arguments, so
it is convex by Lemma~\ref{lem:endpoint}. For fixed $(p_1,p_2)$ its
maximum over $z$ is at one of the two ends
$z^-:=\tfrac13(1-p_1-2p_2)$ and $z^+:=1-p_1-2p_2$.

The two functions
$(p_1,p_2)\mapsto\Phi(p_1,p_2,z^\pm)$ are convex, because $z^\pm$ is
affine in $(p_1,p_2)$. The four inequalities in the first line of
\eqref{eq:polytope} cut out the quadrilateral with corners
\begin{equation}\label{eq:W}
  (\tfrac13,\tfrac13),\qquad(\tfrac25,\tfrac15),\qquad
  (\tfrac12,\tfrac16),\qquad(\tfrac12,\tfrac14)
\end{equation}
(Appendix~\ref{app:polytope}). A convex
function on a convex polygon is at most its largest value at a corner.
At each corner $\Phi$ is negative for both $z=z^-$ and $z=z^+$
(Table~\ref{tab:W}, computed in Appendix~\ref{app:polytope}), so
$\Delta<0$.
\end{proof}

\begin{table}[t]
\renewcommand{\arraystretch}{1.25}
\caption{The values of $\Phi$ at the corners in the proof of
Lemma~\ref{lem:B2}. At the first and last corner $z^-=z^+=0$. The first
value is negative by (N2), the others by Appendix~\ref{app:polytope}.}
\label{tab:W}
\centering
\begin{tabular}{@{}lll@{}}
\toprule
$(p_1,p_2,z)$ & $\Phi(p_1,p_2,z)$ & bound\\
\midrule
$(\tfrac13,\tfrac13,0)$ & $-1-\lgg3+\tfrac{10}9\lgg5$ & $<0$\\
$(\tfrac25,\tfrac15,\tfrac1{15})$ & $-\tfrac23-\tfrac25\lgg3+\tfrac49\lgg5$ & $\le-\tfrac{31}{135}$\\
$(\tfrac25,\tfrac15,\tfrac15)$ & $-\tfrac45+\tfrac13\lgg5$ & $\le-\tfrac1{45}$\\
$(\tfrac12,\tfrac16,\tfrac1{18})$ & $-\tfrac{23}9-\tfrac56\lgg3+\tfrac{40}{27}\lgg5$ & $\le-\tfrac{113}{324}$\\
$(\tfrac12,\tfrac16,\tfrac16)$ & $-\tfrac83-\tfrac12\lgg3+\tfrac{25}{18}\lgg5$ & $\le-\tfrac{19}{108}$\\
$(\tfrac12,\tfrac14,0)$ & $-\tfrac{13}4+\tfrac54\lgg5$ & $\le-\tfrac13$\\
\bottomrule
\end{tabular}
\end{table}

The cases of Section~\ref{ss:induct} and Cases A, B1 and B2 cover all
splits of the root, which proves Theorem~\ref{thm:aff}. \qed

\section{The Range $\tfrac13\le p_1<\tfrac12$}\label{sec:midcap}

Let $\tfrac13\le p_1<\tfrac12$. If $\ell_1=1$, the
redundancy is $1-h(p_1)$ plus that of the right child. Otherwise the left
child of the root has at least two symbols, and its split separates
symbol $1$.

\subsection{A left child with at least two symbols}

Suppose that the left child of the root has at least two symbols, and
let $x$ be the mass of the left child without $p_1$. For $0<p_1<\tfrac12$
and $0<x<1-p_1$ let
\begin{equation}\label{eq:D2}
\begin{aligned}
  V(p_1,x)\ :=\ {}&1+p_1+2x+(1-p_1-x)\,\CC\\
  &-\Ent(p_1,\,x,\,1-p_1-x).
\end{aligned}
\end{equation}
We use the values of $V$ at $x=\tfrac{1-2p_1}3$ and $x=1-2p_1$, and of
$V+x\CC$ at $x=\tfrac{2(1-2p_1)}5$ and $x=\tfrac{2(1-2p_1)}3$, the ends of
the intervals in Fig.~\ref{fig:windows}. For $\tfrac13\le p_1\le\tfrac12$
they are
\begin{subequations}\label{eq:midcurves}
\begin{align}
  V\bigl(p_1,\tfrac{1-2p_1}3\bigr)\ &=\ U_6(p_1),\label{eq:D2third}\\
  V_1(p_1)\ &:=\ V(p_1,1-2p_1)\nonumber\\
  &=\ 3-3p_1+\CC\,p_1\nonumber\\
  &\qquad-\Ent(p_1,p_1,1-2p_1) ,\label{eq:D2delta}\\
  V_2(p_1)\ &:=\ \tfrac{9-3p_1}{5}+(1-p_1)\,\CC\nonumber\\
  &\qquad-\Ent\Bigl(p_1,\tfrac{2(1-2p_1)}5,\tfrac{3-p_1}5\Bigr),\label{eq:D3a}\\
  V_3(p_1)\ &:=\ \tfrac{7-5p_1}{3}+(1-p_1)\,\CC\nonumber\\
  &\qquad-\Ent\Bigl(p_1,\tfrac{2(1-2p_1)}3,\tfrac{1+p_1}3\Bigr).\label{eq:D3b}
\end{align}
\end{subequations}

\begin{figure}[t]
\centering
\begin{tikzpicture}[font=\footnotesize,x=6.6cm]
\draw[->,>=Latex] (0,0) -- (1.08,0) node[right] {$\dfrac{x}{1-2p_1}$};
\foreach \v/\l in {0/0,0.33333/\tfrac13,0.4/\tfrac25,0.66667/\tfrac23,1/1}
  \draw (\v,0.08) -- (\v,-0.08) node[below] {$\l$};
\draw[line width=1.6pt] (0.33333,1.55) -- (1,1.55);
\node[anchor=south] at (0.66667,1.6) {two symbols, $R<V(p_1,x)$};
\draw[line width=1.6pt] (0.4,0.8) -- (0.66667,0.8);
\node[anchor=south] at (0.53333,0.85) {three or more symbols, $R<V(p_1,x)+x\CC$};
\foreach \v in {0.33333,0.4,0.66667,1}
  \draw[densely dotted,black!50] (\v,0.12) -- (\v,0.42);
\end{tikzpicture}
\caption{The possible values of the mass $x$ of the left child of the root
without $p_1$, in units of $1-2p_1$ (Lemma~\ref{lem:fine}). The upper bar
is for a left child with two symbols and the lower bar for one with at
least three.}
\label{fig:windows}
\end{figure}

\begin{lemma}\label{lem:fine}
Assume $\tfrac13\le p_1<\tfrac12$. Let the left child $B_L$ of the root have
at least two symbols and mass $a$, and put $x:=a-p_1$. Then the split of
$B_L$ separates symbol~$1$, so $\ell_1=2$. Moreover:
\begin{enumerate}
\item[\textup{(i)}] if $|B_L|=2$, then $R<V(p_1,x)$ and
$\tfrac{1-2p_1}3\le x\le1-2p_1$;
\item[\textup{(ii)}] if $|B_L|\ge3$, then $R<V(p_1,x)+x\CC$ and
$\tfrac{2(1-2p_1)}5\le x\le\tfrac{2(1-2p_1)}3$;
\item[\textup{(iii)}] in both cases
\begin{equation}\label{eq:fourcurves}
  R\ <\ \max\bigl\{U_6(p_1),\ V_1(p_1),\ V_2(p_1),\ V_3(p_1)\bigr\}.
\end{equation}
\end{enumerate}
\end{lemma}

\begin{proof}
Let $B_L=\{1,\dots,k\}$, and let $L'=\{2,\dots,k\}$ be the rest of
$B_L$. The mass of $L'$ is $x$. Its smallest symbol $p_k$ is one of the
symbols of $L'$, so $p_k\le x$.

The two children of the root have masses $a=p_1+x$ and $1-a$. Their
difference is $2a-1=2x-(1-2p_1)$. By Lemma~\ref{lem:key}, the absolute
value of this difference is at most $p_k$. Together with $p_k\le x$
this gives
\begin{equation}\label{eq:fine-window}
  |2x-(1-2p_1)|\ \le\ p_k\ \le\ x .
\end{equation}
In particular $2x-(1-2p_1)\le x$, that is, $x\le1-2p_1$. Hence
$a=p_1+x\le1-p_1$, and $1-p_1\le2p_1$ because $p_1\ge\tfrac13$. So
$a\le2p_1$, which means that symbol $1$ has at least half of the mass
of $B_L$. By Lemma~\ref{lem:iso}, the split of $B_L$ therefore separates
symbol $1$, and so $\ell_1=2$. Now Lemma~\ref{lem:twolevel} gives
\begin{equation}\label{eq:fine-master}
  R\ =\ (1+a)-\Ent(p_1,x,1-a)+R_{L'}+R_{B_R}.
\end{equation}
The right child $B_R$ has mass $1-a$, and its largest symbol is
$p_{k+1}$. By Corollary~\ref{cor:aff}, $R_{B_R}<p_{k+1}+(1-a)\CC$. This
also holds when $B_R$ is a single symbol, because then $R_{B_R}=0$.

(i) In this case $B_L=\{1,2\}$ and $L'=\{2\}$. So $x=p_2=p_k$, and
$R_{L'}=0$ because $L'$ is a single symbol. Since $p_{k+1}\le p_k=x$, the
bound on $R_{B_R}$ gives $R_{B_R}<x+(1-a)\CC$. Substituting this into
\eqref{eq:fine-master} gives
$R<(1+a)-\Ent(p_1,x,1-a)+x+(1-a)\CC=V(p_1,x)$. By
\eqref{eq:fine-window}, $|2x-(1-2p_1)|\le x$, which is
equivalent to $\tfrac{1-2p_1}3\le x\le1-2p_1$.

(ii) In this case $L'$ has at least two symbols. By
Corollary~\ref{cor:aff} applied to $L'$, whose largest symbol is $p_2$,
we have $R_{L'}<p_2+x\CC$. The symbols $p_2$ and $p_k$ are different
symbols of $L'$, so $p_2+p_k\le x$, and hence $p_2+p_{k+1}\le x$.
Substituting the bounds on $R_{L'}$ and $R_{B_R}$ into
\eqref{eq:fine-master} gives $R<V(p_1,x)+x\CC$. For the range of $x$,
note that $p_k\le p_2$ and $p_k\le x-p_2$, so $p_k\le\tfrac x2$. Then
\eqref{eq:fine-window} gives $|2x-(1-2p_1)|\le\tfrac x2$, which is
equivalent to $\tfrac{2(1-2p_1)}5\le x\le\tfrac{2(1-2p_1)}3$.

(iii) The partial derivatives
\begin{equation}\label{eq:dDdx}
\begin{aligned}
  \frac{\partial V}{\partial x}\ &=\ 2-\CC+\lgg\frac{x}{1-p_1-x},\\
  \frac{\partial (V+x\CC)}{\partial x}\ &=\ 2+\lgg\frac{x}{1-p_1-x}
\end{aligned}
\end{equation}
are strictly increasing in $x$. So, for fixed $p_1$, both $V$ and
$V+x\CC$ are convex in $x$, and on an interval of $x$ each of them is at
most its larger value at the two ends. In case (i), $R$ is therefore
less than the larger of the values of $V$ at $x=\tfrac{1-2p_1}3$ and at
$x=1-2p_1$. In case (ii), $R$ is less than the larger of the values of
$V+x\CC$ at $x=\tfrac{2(1-2p_1)}5$ and at $x=\tfrac{2(1-2p_1)}3$. By
\eqref{eq:midcurves}, these four values are $U_6(p_1)$, $V_1(p_1)$,
$V_2(p_1)$ and $V_3(p_1)$.
\end{proof}

\subsection{The bounds at the ends of the intervals}

It remains to show that $V_1$, $V_2$, $V_3$ and $U_6$ are at most $U_5$
on $[\tfrac13,b_3]$, and that $U_5$, $V_1$, $V_2$ and $V_3$ are at most
$U_6$ on $[b_3,\tfrac12)$ (Lemma~\ref{lem:middom}). By \eqref{eq:second},
$1-2p_1<\tfrac13$ and $\tfrac1{1-p_1}<2$, each second derivative in
Table~\ref{tab:midcurv} is negative on $(\tfrac13,\tfrac12)$, so each
difference is concave on $[\tfrac13,\tfrac12]$.

\begin{table*}[t]
\renewcommand{\arraystretch}{1.25}
\caption{Values of the functions of Section~\ref{sec:midcap}, obtained
from \eqref{eq:pieces} and \eqref{eq:midcurves}. At $p_1=\tfrac12$ the
four functions other than $U_5$ equal $K$.}
\label{tab:midvalues}
\centering
\footnotesize
\setlength{\tabcolsep}{4pt}
\begin{tabular}{@{}lllll@{}}
\toprule
function & $p_1=\tfrac13$ & $p_1=\tfrac25$ & $p_1=\tfrac7{16}$ & $p_1=\tfrac12$\\
\midrule
$U_5$ & $\tfrac{10}3-\lgg3-\tfrac59\lgg5$ & $\tfrac{29}{10}+\tfrac35\lgg3-\tfrac32\lgg5$ & $-\tfrac{51}{32}+\tfrac98\lgg3-\tfrac{15}{32}\lgg5+\tfrac7{16}\lgg7$ & $\tfrac K2$\\
$V_1$ & $\tfrac{10}3-\lgg3-\tfrac59\lgg5$ & $\tfrac{21}5-\tfrac53\lgg5$ & $-\tfrac7{16}-\tfrac{35}{48}\lgg5+\tfrac78\lgg7$ & $K$\\
$U_6$ & $\tfrac{34}9-\tfrac53\lgg3-\tfrac{10}{27}\lgg5$ & $\tfrac{17}3-\tfrac35\lgg3-\tfrac{17}9\lgg5$ & $-\tfrac{17}{48}-\tfrac9{16}\lgg3+\tfrac{25}{144}\lgg5+\tfrac7{16}\lgg7$ & $K$\\
$V_2$ & $6-\lgg3-\tfrac{16}9\lgg5$ & $\tfrac{111}{25}-\tfrac{13}5\lgg5+\tfrac{13}{25}\lgg13$ & not used & $K$\\
$V_3$ & $\tfrac{50}9-\tfrac53\lgg3-\tfrac{10}9\lgg5$ & $\tfrac{23}5-\tfrac35\lgg3-2\lgg5+\tfrac7{15}\lgg7$ & not used & $K$\\
\bottomrule
\end{tabular}
\end{table*}

\begin{table*}[t]
\renewcommand{\arraystretch}{1.25}
\caption{Second derivatives, times $\ln2$, of the differences used in
Lemma~\ref{lem:middom}, Proposition~\ref{prop:dom} and
Lemma~\ref{lem:msup-q}, from \eqref{eq:second}. The last column shows
why each second derivative is negative for $\tfrac13<p_1<\tfrac12$: every
positive term is smaller than a negative term.}
\label{tab:midcurv}
\centering
\begin{tabular}{@{}lll@{}}
\toprule
difference & $(\cdot)''\ln2$ & reason it is negative\\
\midrule
$U_5-V_1$ & $-\dfrac1{p_1}+\dfrac1{1-p_1}-\dfrac4{1-2p_1}$ & $\tfrac1{1-p_1}<2$ and $\tfrac4{1-2p_1}>12$\\[6pt]
$U_6-V_1$ & $-\dfrac1{p_1}-\dfrac8{3(1-2p_1)}+\dfrac1{3(2-p_1)}$ & $\tfrac1{3(2-p_1)}<\tfrac29$ and $\tfrac8{3(1-2p_1)}>8$\\[6pt]
$U_6-V_2$ & $-\dfrac4{15(1-2p_1)}+\dfrac{9-2p_1}{15(2-p_1)(3-p_1)}$ & second term $<\tfrac4{27}$ and $\tfrac4{15(1-2p_1)}>\tfrac45$\\[6pt]
$U_6-V_3$ & $-\dfrac4{3(1-2p_1)}-\dfrac{1-2p_1}{3(2-p_1)(1+p_1)}$ & both terms are negative\\[6pt]
$U_5-V_2$ & $\dfrac1{1-p_1}-\dfrac8{5(1-2p_1)}-\dfrac1{5(3-p_1)}$ & $\tfrac1{1-p_1}<2$ and $\tfrac8{5(1-2p_1)}>\tfrac{24}5$\\[6pt]
$U_5-V_3$ & $\dfrac1{1-p_1}-\dfrac8{3(1-2p_1)}-\dfrac1{3(1+p_1)}$ & $\tfrac1{1-p_1}<2$ and $\tfrac8{3(1-2p_1)}>8$\\[6pt]
$U_1-U_6$ & $-\dfrac1{p_1}-\dfrac4{3(1-2p_1)}-\dfrac1{3(2-p_1)}$ & all terms are negative\\
\bottomrule
\end{tabular}
\end{table*}

\subsection{The crossing $b_3$}

\begin{lemma}[Crossing of $U_5$ and $U_6$]\label{lem:cross}
On $[\tfrac13,\tfrac12]$:
\begin{enumerate}
\item[\textup{(a)}] $U_5$ is strictly decreasing, with
$U_5'(p_1)=\lgg\tfrac{p_1}{1-p_1}-K\le-K$, and $U_6$ is strictly
increasing, with $U_6'(p_1)\ge\tfrac{9\lgg3+2\lgg5-15}9>0$;
\item[\textup{(b)}] $(U_5-U_6)(\tfrac13)=-\tfrac49+\tfrac23\lgg3-\tfrac5{27}\lgg5>0$
and $(U_5-U_6)(\tfrac12)=-\tfrac K2<0$. So there is a unique
$b_3\in(\tfrac13,\tfrac12)$ with $U_5(b_3)=U_6(b_3)$. Also
$U_6\le U_5$ on $[\tfrac13,b_3]$ and $U_5\le U_6$ on $[b_3,\tfrac12]$;
\item[\textup{(c)}] $\tfrac25<b_3<\tfrac7{16}$, because
$(U_5-U_6)(\tfrac25)=-\tfrac{83}{30}+\tfrac65\lgg3+\tfrac7{18}\lgg5>0$ and
$(U_6-U_5)(\tfrac7{16})=\tfrac{119}{96}-\tfrac{27}{16}\lgg3+\tfrac{185}{288}\lgg5>0$.
\end{enumerate}
\end{lemma}

\begin{proof}
(a) The derivative of $U_5$ is $U_5'=-h'(p_1)-K=\lgg\tfrac{p_1}{1-p_1}-K$,
which is at most $-K$ because $p_1\le1-p_1$. For $U_6$, differentiate
\eqref{eq:U6}. The terms $\lgg e$ cancel, because the slopes
$1,-\tfrac23,-\tfrac13$ of the entropy arguments sum to zero. So
\begin{equation}\label{eq:U6prime}
\begin{aligned}
  U_6'(p_1)\ &=\ -\tfrac{1+\CC}3+\lgg3+\lgg p_1\\
  &\qquad-\tfrac23\lgg(1-2p_1)-\tfrac13\lgg(2-p_1)\\
  &\ge\ -\tfrac{1+\CC}3+\lgg3-\tfrac{\lgg5}3\\
  &=\ \tfrac{9\lgg3+2\lgg5-15}9,
\end{aligned}
\end{equation}
using $\lgg p_1\ge-\lgg3$, $\lgg(1-2p_1)\le-\lgg3$,
$\lgg(2-p_1)\le\lgg5-\lgg3$ and $\CC=4-\tfrac53\lgg5$. The last expression
is positive by (N7).

(b), (c) The values are differences of entries of
Table~\ref{tab:midvalues}, and their signs are (N8), (N9) and (N10). By
(a), $U_5-U_6$ is strictly decreasing, so it has exactly one zero and the
stated signs on either side of it. The signs at $\tfrac25$ and
$\tfrac7{16}$ locate the zero.
\end{proof}

\subsection{Comparison with $U_5$ and $U_6$}

\begin{lemma}\label{lem:middom}\leavevmode
\begin{enumerate}
\item[\textup{(i)}] On $[\tfrac13,b_3]$ each of $V_1$, $U_6$,
$V_2$, $V_3$ is at most $U_5$.
\item[\textup{(ii)}] On $[b_3,\tfrac12)$ each of $U_5$, $V_1$,
$V_2$, $V_3$ is at most $U_6$.
\end{enumerate}
\end{lemma}

\begin{proof}
Every difference below is concave (Table~\ref{tab:midcurv}), so by
Lemma~\ref{lem:endpoint} it suffices to check its values at the ends of an
interval. These values are differences of entries of
Table~\ref{tab:midvalues} (Appendix~\ref{app:mid}), and they are
nonnegative by (N11) to (N15).

(i) The difference $U_5-V_1$ vanishes at $\tfrac13$, where both
values are $\tfrac{10}3-\lgg3-\tfrac59\lgg5$. At $\tfrac7{16}$ it equals
$-\tfrac{37}{32}+\tfrac98\lgg3+\tfrac{25}{96}\lgg5-\tfrac7{16}\lgg7>0$. So
$U_5-V_1\ge0$ on $[\tfrac13,\tfrac7{16}]$, which contains
$[\tfrac13,b_3]$ by Lemma~\ref{lem:cross}(c). The bound $U_6\le U_5$ on
$[\tfrac13,b_3]$ is Lemma~\ref{lem:cross}(b). The difference
$U_5-V_2$ equals $-\tfrac83+\tfrac{11}9\lgg5>0$ at $\tfrac13$. At
$\tfrac25$ it equals $(U_5-U_6)(\tfrac25)+(U_6-V_2)(\tfrac25)$, a
sum of two positive numbers by Lemma~\ref{lem:cross}(c) and part (ii)
below. So $U_5-V_2\ge0$ on $[\tfrac13,\tfrac25]$. On
$[\tfrac25,b_3]$ we have $V_2\le U_6\le U_5$. The first bound is
proved in part (ii) on all of $[\tfrac25,\tfrac12]$, and the second is
Lemma~\ref{lem:cross}(b). The same argument applies to $U_5-V_3$,
whose value at $\tfrac13$ is $-\tfrac{20}9+\tfrac23\lgg3+\tfrac59\lgg5>0$.

(ii) The bound $U_5\le U_6$ on $[b_3,\tfrac12]$ is
Lemma~\ref{lem:cross}(b). The difference $U_6-V_1$ vanishes at
$\tfrac12$, where both values are $K$. At $b_3$ it equals
$(U_5-V_1)(b_3)\ge0$ by part (i). So $U_6-V_1\ge0$ on
$[b_3,\tfrac12]$. The differences $U_6-V_2$ and $U_6-V_3$
vanish at $\tfrac12$. At $\tfrac25$ they equal
\begin{equation}\label{eq:at25}
\begin{aligned}
  &\tfrac{92}{75}-\tfrac35\lgg3+\tfrac{32}{45}\lgg5-\tfrac{13}{25}\lgg13\\
  &\text{and}\quad
  \tfrac{16}{15}+\tfrac19\lgg5-\tfrac7{15}\lgg7,
\end{aligned}
\end{equation}
which are positive. So both differences are nonnegative on
$[\tfrac25,\tfrac12]$, which contains $[b_3,\tfrac12]$.
\end{proof}

\subsection{Proof of Theorem~\ref{thm:mid} and of the bound $R<K$}

\begin{theorem}\label{thm:mid}
For every source with $\tfrac13\le p_1<\tfrac12$,
\begin{equation}\label{eq:mid}
\begin{aligned}
  R(\bp)\ &<\ \max\bigl(U_5(p_1),U_6(p_1)\bigr)\\
  &=\ \begin{cases}
    U_5(p_1) & \text{if }\tfrac13\le p_1\le b_3,\\
    U_6(p_1) & \text{if }b_3\le p_1<\tfrac12 .
  \end{cases}
\end{aligned}
\end{equation}
\end{theorem}

\begin{proof}
The piecewise form is Lemma~\ref{lem:cross}(b). We prove
$R<\max(U_5,U_6)(p_1)$ by strong induction on $n$, over all sources with
$\tfrac13\le p_1<\tfrac12$. Such a source has $n\ge3$.

Let $B$ be a block with fewer than $n$ symbols, and let $q<\tfrac12$ be
the largest probability of $\bp^{(B)}$. Then
\begin{equation}\label{eq:RBK}
  \frac{R_B}{M_B}\ <\
  \begin{cases}
    q+\CC\ \overset{(a)}{<}\ \tfrac13+\CC\ \overset{(b)}{<}\ K
      & \text{if }q<\tfrac13,\\[2pt]
    \max\bigl(U_5(q),U_6(q)\bigr)
      & \text{if }\tfrac13\le q<\tfrac12,\\
    \rlap{$\quad\overset{(c)}{\le}\
    \max\bigl(U_5(\tfrac13),U_6(q)\bigr)\ \overset{(d)}{<}\ K$} &
  \end{cases}
\end{equation}
where the first inequality follows from Corollary~\ref{cor:aff} if
$q<\tfrac13$ and from the induction hypothesis otherwise, (a) from
$q<\tfrac13$, (b) from \eqref{eq:constfacts}, (c) from
Lemma~\ref{lem:cross}(a), and (d) from
$K-U_5(\tfrac13)=\lgg3-\tfrac56-\tfrac5{18}\lgg5>0$, which is (N16), and
from $U_6(q)<U_6(\tfrac12)=K$, which holds by Lemma~\ref{lem:cross}(a).
So $R_B<M_BK$.

\emph{Case: $\ell_1=1$.} By \eqref{eq:isorec},
$R=(1-h(p_1))+R_{\{2,\dots,n\}}$. The rescaled child has largest
probability $q=p_2/(1-p_1)$, and $q\le y:=p_1/(1-p_1)$, where
$y\in[\tfrac12,1)$. If $q<\tfrac12$, then $R_{\{2,\dots,n\}}<(1-p_1)K$,
and $R<U_5(p_1)$. If $q\ge\tfrac12$, then Corollary~\ref{cor:isob} gives
$R_{\{2,\dots,n\}}<(1-p_1)U_7(q)\le(1-p_1)\max\bigl(U_7(\tfrac12),U_7(y)\bigr)$,
because $U_7$ is convex and $q\in[\tfrac12,y]$. If the maximum is
$U_7(\tfrac12)=\tfrac12$, then $R<(1-h(p_1))+\tfrac{1-p_1}2<U_5(p_1)$,
since $\tfrac12<K$. Otherwise the grouping identity gives
\begin{equation}\label{eq:T}
\begin{aligned}
  R\ &<\ \bigl(1-h(p_1)\bigr)+(1-p_1)\,U_7(y)\\
  &=\ 3-3p_1-\Ent(p_1,p_1,1-2p_1)\ \le\ V_1(p_1).
\end{aligned}
\end{equation}

\emph{Case: the left child of the root has at least two symbols.}
By Lemma~\ref{lem:fine}, $R$ is less than the largest of $U_6(p_1)$,
$V_1(p_1)$, $V_2(p_1)$ and $V_3(p_1)$. This case needs no induction
hypothesis.

So in every case $R$ is less than $U_5(p_1)$ or one of $V_1$, $U_6$,
$V_2$, $V_3$ at $p_1$. By Lemma~\ref{lem:middom}, each of these is at
most $U_5(p_1)$ on $[\tfrac13,b_3]$ and at most $U_6(p_1)$ on
$[b_3,\tfrac12)$, that is, at most $\max(U_5,U_6)(p_1)$.
\end{proof}

\begin{proof}[Proof of Corollary~\ref{cor:kcap}]
Let $p_1<\tfrac12$. Then
\begin{equation}\label{eq:Kcases}
  R\ <\
  \begin{cases}
    p_1+\CC\ \overset{(a)}{<}\ \tfrac13+\CC\ \overset{(b)}{<}\ K
      & \text{if }p_1<\tfrac13,\\[2pt]
    U_5(p_1)\ \overset{(c)}{\le}\ U_5(\tfrac13)\ \overset{(d)}{<}\ K
      & \text{if }\tfrac13\le p_1\le b_3,\\[2pt]
    U_6(p_1)\ \overset{(e)}{<}\ U_6(\tfrac12)\ =\ K
      & \text{if }b_3\le p_1<\tfrac12,
  \end{cases}
\end{equation}
where the first inequality follows from Corollary~\ref{cor:aff} if
$p_1<\tfrac13$ and from Theorem~\ref{thm:mid} otherwise, (a) from
$p_1<\tfrac13$, (b) from \eqref{eq:constfacts}, (c) and (e) from
Lemma~\ref{lem:cross}(a), and (d) from (N16). The bound $R<p_1+\CC$ for
every source is Corollary~\ref{cor:aff}. It remains to show $U<1$. On
$(0,\tfrac12)$ every piece is below $K<1$, by the right sides of
\eqref{eq:Kcases} and by Proposition~\ref{prop:dom-q} on
$[\tfrac14,\tfrac13)$. On $[\tfrac12,1)$,
$U_7(p_1)=2-p_1-h(p_1)<1$. Indeed the concave function $h$ lies above its
chord $2(1-p_1)$ on $[\tfrac12,1]$, so $h(p_1)>1-p_1$ for $p_1<1$.
\end{proof}

\begin{proposition}\label{prop:dom}
On $[\tfrac13,\tfrac12)$, $\max(U_5,U_6)<U_1=p_1+\CC$. The difference is
at least $\tfrac23\CC-(1-h(\tfrac13))$ on $[\tfrac13,b_3]$ and at least
$\tfrac{\CC}2$ on $[b_3,\tfrac12)$.
\end{proposition}

\begin{proof}
On $[\tfrac13,b_3]$ the difference $U_1-U_5$ has derivative
$1-U_5'=1+\lgg\tfrac{1-p_1}{p_1}+K>0$. So it is smallest at $\tfrac13$,
where it equals $\tfrac13+\CC-U_5(\tfrac13)=\tfrac23\CC-(1-h(\tfrac13))$.
This is positive by \eqref{eq:constfacts}. On $[b_3,\tfrac12)$ the
difference $U_1-U_6$ is concave (Table~\ref{tab:midcurv}). On
$[\tfrac25,\tfrac12]$, which contains $[b_3,\tfrac12)$, it is
at least the smaller of its values $\tfrac{\CC}2$ at $\tfrac12$ and
$-\tfrac{19}{15}+\tfrac35\lgg3+\tfrac29\lgg5$ at $\tfrac25$. The latter
exceeds $\tfrac{\CC}2$ by (N17).
\end{proof}

\section{The Range $\tfrac14\le p_1<\tfrac13$}\label{sec:quarter}

For $\tfrac14\le p_1<\tfrac13$ the first two levels of the tree
are known (Lemma~\ref{lem:noburial-q}). So the redundancy is a
function of $p_1$ and of the mass $a$ of the left child, plus the
redundancies of two blocks.

Let $\tfrac14\le p_1<\tfrac13$. By Lemma~\ref{lem:noburial-q} the left
child $B_L=\{1,\dots,k\}$ of the root has $k\ge2$ symbols and mass $a$,
and its split separates symbol $1$. The block $L'=\{2,\dots,k\}$ has mass $x=a-p_1$,
and the right child $B_R$ has mass $1-a$ and largest symbol
$p_{k+1}\le p_k$. In this notation \eqref{eq:twolevel} reads
\begin{equation}\label{eq:twolevel-q}
  R\ =\ (1+a)-\Ent(p_1,\,a-p_1,\,1-a)+R_{L'}+R_{B_R}.
\end{equation}
For fixed $p_1$, the first two terms satisfy
\begin{subequations}\label{eq:Jprime}
\begin{align}
  &\frac{d}{da}\Bigl((1+a)-\Ent(p_1,a-p_1,1-a)\Bigr)\nonumber\\
  &\qquad=\ 1+\lgg\frac{a-p_1}{1-a},\label{eq:Jprime1}\\
  &\frac{d^2}{da^2}\Bigl((1+a)-\Ent(p_1,a-p_1,1-a)\Bigr)\nonumber\\
  &\qquad=\ \frac1{\ln2}\Bigl(\frac1{a-p_1}+\frac1{1-a}\Bigr)\ >\ 0 .\label{eq:Jprime2}
\end{align}
\end{subequations}
So this function is convex in $a$.

\begin{lemma}\label{lem:msup-q}
Let $B$ be a block, and let $q$ be the largest probability of
$\bp^{(B)}$. Then:
\begin{enumerate}
\item[\textup{(a)}] if $q<\tfrac13$, then $R_B<M_B(q+\CC)$;
\item[\textup{(b)}] if $\tfrac13\le q<\tfrac12$, then
$R_B<M_B\max\bigl(U_5(q),U_6(q)\bigr)\le M_B\max\bigl(\tfrac13+\CC,U_6(q)\bigr)$;
\item[\textup{(c)}] if $q\ge\tfrac12$, then $R_B<M_B\,U_7(q)$;
\item[\textup{(d)}] if $q\le Q<\tfrac12$, then $R_B<M_B(Q+\CC)$. If
also $\tfrac13\le Q$, then $R_B<M_B\max\bigl(\tfrac13+\CC,U_6(Q)\bigr)$;
\item[\textup{(e)}] if $q\le Q$ with $\tfrac12\le Q<1$, then
$R_B<M_B\max\bigl(K,U_7(Q)\bigr)$.
\end{enumerate}
\end{lemma}

\begin{proof}
(a), (b) and (c) are Corollary~\ref{cor:aff}, Theorem~\ref{thm:mid} and
Corollary~\ref{cor:isob} applied to $\bp^{(B)}$. For the second bound in
(b), $U_5(q)\le U_5(\tfrac13)$ by Lemma~\ref{lem:cross}(a), and
$U_5(\tfrac13)=2-\lgg3+\tfrac{\CC}3\le\tfrac13+\CC$ is equivalent to the
second inequality in \eqref{eq:constfacts}.

(d) If $q<\tfrac13$, then (a) gives $R_B<M_B(q+\CC)$, which is at most
both $M_B(Q+\CC)$ and $M_B(\tfrac13+\CC)$. If $q\ge\tfrac13$, then
$Q\ge\tfrac13$, and (b) with $U_6(q)\le U_6(Q)$
(Lemma~\ref{lem:cross}(a)) gives the second claim. The second claim
implies the first, because $\tfrac13+\CC\le Q+\CC$ and
$U_6(Q)\le Q+\CC$. The last inequality holds because $U_1-U_6$ is
concave on $[\tfrac13,\tfrac12]$ (Table~\ref{tab:midcurv}) and positive
at both ends, by (N21) and since $\tfrac12+\CC-K=\tfrac{\CC}2$.

(e) For $q\in[\tfrac12,Q]$, the convexity of $U_7$ gives
$U_7(q)\le\max\bigl(U_7(\tfrac12),U_7(Q)\bigr)=\max\bigl(\tfrac12,U_7(Q)\bigr)\le\max\bigl(K,U_7(Q)\bigr)$,
since $\tfrac12<K$. For $q<\tfrac12$ the bounds of (a) and (b) are at
most $\max(\tfrac13+\CC,K)=K$ by \eqref{eq:constfacts}.
\end{proof}

\subsection{Auxiliary bounds}\label{ss:aux}

The proofs below bound $R$ by $U_2$, $U_3$, $U_4$, $U_6$ and the
functions (Table~\ref{tab:corners})
\begin{subequations}\label{eq:corners}
\begin{align}
  W_1(p_1)\ &:=\ 1+2p_1+(1-2p_1)K\nonumber\\
  &\qquad-\Ent(p_1,p_1,1-2p_1),\label{eq:g1}\\
  W_2(p_1)\ &:=\ 2-p_1+\tfrac{3-p_1}5\,\CC\nonumber\\
  &\qquad-\Ent\Bigl(p_1,\tfrac{3-p_1}5,\tfrac{1-2p_1}5,\tfrac{1-2p_1}5\Bigr),\label{eq:g2}\\
  W_3(p_1)\ &:=\ \tfrac{8-7p_1}3+\tfrac{1+p_1}3\,\CC\nonumber\\
  &\qquad-\Ent\Bigl(p_1,\tfrac{1+p_1}3,\tfrac{1-2p_1}3,\tfrac{1-2p_1}3\Bigr),\label{eq:g3}\\
  W_4(p_1)\ &:=\ \tfrac{12-3p_1}7+(1-p_1)\,\CC\nonumber\\
  &\qquad-\Ent\Bigl(p_1,\tfrac{3(1-2p_1)}7,\tfrac{4-p_1}7\Bigr),\label{eq:g4}\\
  W_5(p_1)\ &:=\ 2-p_1+(1-p_1)\,\CC\nonumber\\
  &\qquad-\Ent\Bigl(p_1,\tfrac{3(1-2p_1)}5,\tfrac{2+p_1}5\Bigr),\label{eq:g5}\\
  W_6(p_1)\ &:=\ \tfrac{12-3p_1}7+\tfrac{4(1-2p_1)}7\,K+\tfrac{3+p_1}7\,\CC\nonumber\\
  &\qquad-\Ent\Bigl(p_1,\tfrac{4(1-2p_1)}7,\tfrac{3+p_1}7\Bigr),\label{eq:g6}\\
  W_7(p_1)\ &:=\ \tfrac{14-p_1}9+\tfrac{4(1-2p_1)}9\,K+\tfrac{5-p_1}9\,\CC\nonumber\\
  &\qquad-\Ent\Bigl(p_1,\tfrac{4(1-2p_1)}9,\tfrac{5-p_1}9\Bigr),\label{eq:g7}\\
  W_8(p_1)\ &:=\ V\Bigl(p_1,\tfrac{1-p_1}4\Bigr).\label{eq:g8}
\end{align}
\end{subequations}

\begin{table}[t]
\renewcommand{\arraystretch}{1.25}
\caption{The functions that bound $R$ in Section~\ref{sec:quarter}: the
number of symbols of $L'$, the left mass $a$ at which each arises, the
function that is at least as large on $[\tfrac14,\tfrac13]$, and the
numerical inequalities used (Appendix~\ref{app:cert}).}
\label{tab:corners}
\centering
\footnotesize
\setlength{\tabcolsep}{3pt}
\begin{tabular}{@{}ll>{\raggedright\arraybackslash}p{0.27\columnwidth}ll@{}}
\toprule
function & $|L'|$ & arises at & at most & inequalities\\
\midrule
$U_6$ & $1$ & $a=\tfrac{1+p_1}3$ & $U_3$ & (N22)\\
$W_8$ & $1$ & $a=\tfrac{1+3p_1}4$ & $U_3$ & (N23), (N24)\\
$U_3$ & $1$ & $a=\tfrac{1+2p_1}3$ & $U_3$ & none\\
$W_1$ & $1$ & $a=2p_1$, right child bounded by $K$ & $U_3$ & (N25)\\
$U_4$ & $1$ & $a=2p_1$, right child bounded by $U_7$ & $U_4$ & none\\
$U_2$ & $2$ and $\ge3$ & $a=\tfrac12$ & $U_2$ & none\\
$W_2$ & $2$ & $a=\tfrac{2+p_1}5$ & $U_2(\tfrac13)$ & (N26), (N27)\\
$W_3$ & $2$ & $a=\tfrac{2-p_1}3$ & $U_2(\tfrac13)$ & (N28), (N29)\\
$W_4$ & $\ge3$ & $a=\tfrac{3+p_1}7$ & $U_2(\tfrac13)$ & (N30), (N31)\\
$W_7$ & $\ge3$ & $a=\tfrac{4+p_1}9$ & $U_2(\tfrac13)$ & (N32), (N33)\\
$W_5$ & $\ge3$ & $a=\tfrac{3-p_1}5$ & $U_2$ & (N34), (N35)\\
$W_6$ & $\ge3$ & $a=\tfrac{4-p_1}7$ & $U_2$ & (N36), (N37)\\
\bottomrule
\end{tabular}
\end{table}

\begin{proposition}\label{prop:L1-q}
If $|L'|=1$, then
\begin{equation}\label{eq:L1}
\begin{aligned}
  R\ <\ \max\bigl\{&U_6(p_1),\,W_8(p_1),\,U_3(p_1),\\
  &W_1(p_1),\,U_4(p_1)\bigr\}.
\end{aligned}
\end{equation}
\end{proposition}

\begin{proof}
In this proof let $f(a):=(1+a)-\Ent(p_1,a-p_1,1-a)$ and
$Q(a):=\tfrac{a-p_1}{1-a}$. Here $x=p_2=p_k\le p_1$ and $R_{L'}=0$. The
rescaled right child has largest probability
$p_{k+1}/(1-a)\le p_k/(1-a)=Q(a)$. So \eqref{eq:twolevel-q} and
Lemma~\ref{lem:msup-q} bound $R$ by $f(a)$ plus $(1-a)$ times the bound
of Lemma~\ref{lem:msup-q}(d) or (e) with $Q=Q(a)$. The constraint
$s\le p_k=x$ of Lemma~\ref{lem:noburial-q}(i) reads
$a\in[\tfrac{1+p_1}3,\,1-p_1]$, and $x\le p_1$ gives $a\le2p_1$. Since
$p_1<\tfrac13$, we have $2p_1<1-p_1$, so the range is $a\in[\tfrac{1+p_1}3,\,2p_1]$. On this
range $Q$ increases from $Q(\tfrac{1+p_1}3)=\tfrac{1-2p_1}{2-p_1}<\tfrac13$
through $Q=\tfrac13$ at $a':=\tfrac{1+3p_1}4$ and $Q=\tfrac12$ at
$a'':=\tfrac{1+2p_1}3$ to $Q(2p_1)=\tfrac{p_1}{1-2p_1}\ge\tfrac12$. Also
$a'\le a''$, and $a''\le2p_1$ because $p_1\ge\tfrac14$.

\emph{On $[\tfrac{1+p_1}3,a']$, where $Q\le\tfrac13$.} The bound is
$f(a)+x+(1-a)\CC=V(p_1,x)$ with $x\in[\tfrac{1-2p_1}3,\tfrac{1-p_1}4]$.
Since $V(p_1,\cdot)$ is convex by \eqref{eq:dDdx}, that gives
\begin{equation}\label{eq:VW8}
\begin{aligned}
  R\ &<\ \max\bigl(V(p_1,\tfrac{1-2p_1}3),V(p_1,\tfrac{1-p_1}4)\bigr)\\
  &=\ \max\bigl(U_6(p_1),W_8(p_1)\bigr).
\end{aligned}
\end{equation}

\emph{On $[a',a'']$, where $\tfrac13\le Q\le\tfrac12$.} The bound is the
larger of
\begin{equation}\label{eq:flco}
  f(a)+(1-a)\bigl(\tfrac13+\CC\bigr)
  \qquad\text{and}\qquad
  f(a)+(1-a)\,U_6\bigl(Q(a)\bigr).
\end{equation}
By \eqref{eq:Jprime}, $f'=1+\lgg Q\le0$ here, so the first expression
decreases in $a$ and is at most its value at $a'$. Since $Q(a')=\tfrac13$,
this value is $f(a')+x+(1-a')\CC=V(p_1,\tfrac{1-p_1}4)=W_8(p_1)$. For
the second expression, the scaling identity \eqref{eq:scaling} in the form
$(1-a)\,\Ent(\mathbf u)=\Ent\bigl((1-a)\mathbf u\bigr)-\Ent(1-a)$, for a
source $\mathbf u$, gives
\begin{equation}\label{eq:co}
\begin{aligned}
  &f(a)+(1-a)\,U_6(Q(a))\\
  &\quad=\ (1+a)+\frac{5(1-a)-x}3+\frac{2(1-a)-x}3\,\CC\\
  &\qquad-\Ent(p_1)-2\Ent(x)\\
  &\qquad-\Ent\Bigl(\frac{1+2p_1-3a}3\Bigr)-\Ent\Bigl(\frac{2+p_1-3a}3\Bigr).
\end{aligned}
\end{equation}
This is an affine function of $a$ minus entropy terms of affine
arguments, so it is convex (Lemma~\ref{lem:endpoint}) and at most
the larger of its values at $a'$ and $a''$. At $a'$ it is at most the first
expression of \eqref{eq:flco}, because $U_6(\tfrac13)\le\tfrac13+\CC$ by
(N21). At $a''$, where $Q=\tfrac12$ and $U_6(\tfrac12)=K$, it equals
$f(a'')+(1-a'')K=U_3(p_1)$.

\emph{On $[a'',2p_1]$, where $Q\ge\tfrac12$.} The bound is the larger of
\begin{equation}\label{eq:flKf}
  f(a)+(1-a)K
  \qquad\text{and}\qquad
  f(a)+(1-a)\,U_7\bigl(Q(a)\bigr).
\end{equation}
The first is convex by \eqref{eq:Jprime}, so it is at most the larger of
its values at $a''$ and $2p_1$, which are $U_3(p_1)$ and $W_1(p_1)$. For
the second, $(1-a)h(Q)=\Ent(x)+\Ent(1+p_1-2a)-\Ent(1-a)$ gives
\begin{equation}\label{eq:f}
\begin{aligned}
  &f(a)+(1-a)\,U_7(Q(a))\\
  &\quad=\ 3-2a+p_1-\Ent(p_1)\\
  &\qquad-2\Ent(a-p_1)-\Ent(1+p_1-2a).
\end{aligned}
\end{equation}
This is again convex, so it is at most the larger of its values at $a''$
and $2p_1$. At $a''$ it equals $f(a'')+(1-a'')U_7(\tfrac12)\le U_3(p_1)$,
because $U_7(\tfrac12)=\tfrac12<K$. At $2p_1$ it equals
$3-3p_1-\Ent(p_1,p_1,p_1,1-3p_1)=U_4(p_1)$.
\end{proof}

\begin{proposition}\label{prop:L2-q}
If $|L'|=2$, then
\begin{equation}\label{eq:L2}
  R\ <\ \max\bigl\{U_2(p_1),\,W_2(p_1),\,W_3(p_1)\bigr\}.
\end{equation}
\end{proposition}

\begin{proposition}\label{prop:L3-q}
If $|L'|\ge3$, then
\begin{equation}\label{eq:L3}
\begin{aligned}
  R\ <\ \max\bigl\{&U_2(p_1),\,W_4(p_1),\,W_5(p_1),\\
  &W_6(p_1),\,W_7(p_1)\bigr\}.
\end{aligned}
\end{equation}
\end{proposition}

Propositions~\ref{prop:L2-q} and~\ref{prop:L3-q} are proved like
Proposition~\ref{prop:L1-q}, and the proofs are in
Appendix~\ref{app:quarter}.

\subsection{Comparison with $U_2$, $U_3$ and $U_4$}

It remains to show that every function in Table~\ref{tab:corners} is at
most $\max(U_2,U_3,U_4)$ on $[\tfrac14,\tfrac13]$.
Table~\ref{tab:qvalues} gives the values at $\tfrac14$ and $\tfrac13$. In
particular
\begin{equation}\label{eq:qends}
\begin{gathered}
  U_2(\tfrac14)=\tfrac14+\tfrac{\CC}2,\quad
  U_3(\tfrac14)=W_1(\tfrac14)=\tfrac{1+\CC}4,\\
  U_4(\tfrac14)=\tfrac14,\quad
  U_3(\tfrac13)=U_6(\tfrac13),\quad
  U_2(\tfrac13)=W_8(\tfrac13).
\end{gathered}
\end{equation}
Also $U_3-W_8$ is affine in $p_1$, because by the grouping identity its
entropy terms are
\begin{equation}\label{eq:U3W8}
\begin{aligned}
  &-\Ent\Bigl(\frac{1-p_1}3,\frac{2(1-p_1)}3\Bigr)+\Ent\Bigl(\frac{1-p_1}4,\frac{3(1-p_1)}4\Bigr)\\
  &\qquad=\ (1-p_1)\bigl(h(\tfrac14)-h(\tfrac13)\bigr).
\end{aligned}
\end{equation}

\begin{table*}[t]
\renewcommand{\arraystretch}{1.25}
\caption{Values at the ends of $[\tfrac14,\tfrac13]$, from
\eqref{eq:pieces} and \eqref{eq:corners}. Decimals are in parentheses.}
\label{tab:qvalues}
\centering
\footnotesize
\setlength{\tabcolsep}{3pt}
\begin{tabular}{@{}lll@{}}
\toprule
function & $p_1=\tfrac14$ & $p_1=\tfrac13$\\
\midrule
$U_2$ & $\tfrac94-\tfrac56\lgg5$ ($0.315060$) & $3-\tfrac12\lgg3-\tfrac56\lgg5$ ($0.272579$)\\
$U_3$ & $\tfrac54-\tfrac5{12}\lgg5$ ($0.282530$) & $\tfrac{34}9-\tfrac53\lgg3-\tfrac{10}{27}\lgg5$ ($0.276200$)\\
$U_4$ & $\tfrac14$ ($0.250000$) & $2-\lgg3$ ($0.415037$)\\
$U_6$ & $\tfrac{25}{12}-\tfrac34\lgg3-\tfrac{35}{36}\lgg5+\tfrac7{12}\lgg7$ ($0.274805$) & $\tfrac{34}9-\tfrac53\lgg3-\tfrac{10}{27}\lgg5$ ($0.276200$)\\
$W_1$ & $\tfrac54-\tfrac5{12}\lgg5$ ($0.282530$) & $\tfrac52-\lgg3-\tfrac5{18}\lgg5$ ($0.270057$)\\
$W_2$ & $\tfrac{43}{20}-\tfrac53\lgg5+\tfrac{11}{20}\lgg{11}$ ($0.182807$) & $\tfrac{27}5-\lgg3-\tfrac{14}9\lgg5$ ($0.203149$)\\
$W_3$ & $\tfrac{25}{12}-\tfrac34\lgg3-\tfrac5{18}\lgg5$ ($0.249631$) & $\tfrac{41}9-\tfrac53\lgg3-\tfrac{20}{27}\lgg5$ ($0.194005$)\\
$W_4$ & $\tfrac{79}{28}+\tfrac34\lgg3-\tfrac57\lgg5-\tfrac34\lgg7$ ($0.246114$) & $\tfrac{89}{21}-\tfrac67\lgg3-\tfrac{10}9\lgg5-\tfrac23\lgg7+\tfrac{11}{21}\lgg{11}$ ($0.240149$)\\
$W_5$ & $\tfrac{61}{20}+\tfrac65\lgg3-2\lgg5$ ($0.308099$) & $\tfrac{13}3-\tfrac45\lgg3-\tfrac{16}9\lgg5+\tfrac7{15}\lgg7$ ($0.247590$)\\
$W_6$ & $\tfrac{85}{28}-\tfrac{85}{84}\lgg5-\tfrac34\lgg7+\tfrac{13}{28}\lgg{13}$ ($0.298689$) & $\tfrac{101}{21}-\lgg3-\tfrac{10}{21}\lgg5-\tfrac23\lgg7$ ($0.247311$)\\
$W_7$ & $\tfrac{103}{36}-\tfrac32\lgg3-\tfrac{115}{108}\lgg5+\tfrac{19}{36}\lgg{19}$ ($0.253206$) & $\tfrac{43}9-\tfrac73\lgg3-\tfrac{80}{81}\lgg5+\tfrac{14}{27}\lgg7$ ($0.241935$)\\
$W_8$ & $\tfrac38+\tfrac{21}{16}\lgg3-\tfrac{15}{16}\lgg5$ ($0.278456$) & $3-\tfrac12\lgg3-\tfrac56\lgg5$ ($0.272579$)\\
\bottomrule
\end{tabular}
\end{table*}

\begin{lemma}\label{lem:dom-q}
On $[\tfrac14,\tfrac13]$:
\begin{enumerate}
\item[\textup{(a)}] $U_6\le U_3$, $W_8\le U_3$, and $W_1\le U_3$;
\item[\textup{(b)}] $U_2$ is strictly decreasing, and
$W_2,W_3,W_4,W_7\le U_2(\tfrac13)\le U_2$;
\item[\textup{(c)}] $W_5\le U_2$ and $W_6\le U_2$.
\end{enumerate}
\end{lemma}

\begin{proof}
Every difference is an affine function plus terms $\pm\Ent(u(p_1))$ with
$u$ affine. The second derivatives come from \eqref{eq:second}, the
values at the ends from Table~\ref{tab:qvalues}, and the signs of these
values from the numerical inequalities of Appendix~\ref{app:cert}.

(a) The difference $U_3-U_6$ has second derivative
$\bigl[\tfrac1{1-p_1}-\tfrac4{3(1-2p_1)}-\tfrac1{3(2-p_1)}\bigr]/\ln2$,
which is negative because $1-2p_1\le\tfrac12$ gives
$\tfrac4{3(1-2p_1)}\ge\tfrac83>\tfrac32\ge\tfrac1{1-p_1}$. So the
difference is concave. It is positive at $\tfrac14$ by (N22) and $0$ at
$\tfrac13$ by \eqref{eq:qends}, so it is nonnegative between. The difference
$U_3-W_8$ is affine and positive at both ends by
(N23) and (N24). The difference $U_3-W_1$ has second derivative
$\bigl[\tfrac1{1-p_1}-\tfrac1{p_1}-\tfrac4{1-2p_1}\bigr]/\ln2<0$, so it is
concave. It vanishes at $\tfrac14$ by \eqref{eq:qends} and is positive at
$\tfrac13$ by (N25).

(b) The derivative $U_2'(p_1)=-1+\lgg\tfrac{p_1}{1/2-p_1}$ is at most
$0$ if and only if $p_1\le2(\tfrac12-p_1)$, that is, $p_1\le\tfrac13$.
Each of $W_2$, $W_3$, $W_4$, $W_7$ is an affine function minus entropy
terms of affine arguments, so it is convex and at most the larger
of its values at $\tfrac14$ and $\tfrac13$. These eight values are below
$U_2(\tfrac13)$ by (N26) to (N33).

(c) In $U_2-W_5$ the terms $\Ent(p_1)$ cancel. The second derivative is
\begin{equation}\label{eq:U2W5}
\begin{aligned}
  &\frac1{\ln2}\Bigl[\frac2{1-2p_1}-\frac{12}{5(1-2p_1)}-\frac1{5(2+p_1)}\Bigr]\\
  &\quad=\ -\frac1{\ln2}\Bigl[\frac2{5(1-2p_1)}+\frac1{5(2+p_1)}\Bigr]\ <\ 0 .
\end{aligned}
\end{equation}
So the difference is concave, and it is positive at both ends by (N34)
and (N35). The difference $U_2-W_6$ has second derivative
$\bigl[\tfrac2{1-2p_1}-\tfrac{16}{7(1-2p_1)}-\tfrac1{7(3+p_1)}\bigr]/\ln2<0$,
so it is concave, and it is positive at both ends by (N36) and (N37).
\end{proof}

\begin{lemma}\label{lem:cross-q}
On $[\tfrac14,\tfrac13]$:
\begin{enumerate}
\item[\textup{(i)}] $U_2-U_3$ is convex, equals $\tfrac{\CC}4>0$ at
$\tfrac14$, and is negative at $\tfrac13$. So it has a unique zero
$b_1$. It is nonnegative before $b_1$ and nonpositive after $b_1$.
\item[\textup{(ii)}] $U_3-U_4$ is concave, equals $\tfrac{\CC}4>0$ at
$\tfrac14$, and is negative at $\tfrac13$. So it has a unique zero
$b_2$. It is nonnegative before $b_2$ and nonpositive after $b_2$.
\item[\textup{(iii)}] $b_1<\tfrac{23}{72}<b_2$.
\end{enumerate}
\end{lemma}

\begin{proof}
By \eqref{eq:second}, $U_2''\ln2=\tfrac1{p_1}+\tfrac1{1/2-p_1}$,
$U_3''\ln2=\tfrac1{p_1}+\tfrac1{1-p_1}$ and
$U_4''\ln2=\tfrac3{p_1}+\tfrac9{1-3p_1}$. So
$(U_2-U_3)''\ln2=\tfrac1{1/2-p_1}-\tfrac1{1-p_1}>0$ and
$(U_3-U_4)''\ln2=\tfrac1{1-p_1}-\tfrac2{p_1}-\tfrac9{1-3p_1}<0$, because
$\tfrac1{1-p_1}<2<\tfrac2{p_1}$. The values
at $\tfrac14$ are given by \eqref{eq:qends}, and the values at $\tfrac13$
are negative by (N18) and (N19). A convex function that is positive at
the left end and negative at the right end of an interval has exactly
one zero there, with the stated signs, and the same holds for a concave
function. For (iii), $U_3-U_2>0$ and $U_3-U_4>0$ at $p_1=\tfrac{23}{72}$
by (N38) and (N39), so this point lies strictly between the two zeros.
\end{proof}

\subsection{Proof of Theorem~\ref{thm:quarter} and exactness of $U_4$}

\begin{theorem}\label{thm:quarter}
For every source with $\tfrac14\le p_1<\tfrac13$,
\begin{equation}\label{eq:quarter}
\begin{aligned}
  R(\bp)\ &<\ \max\bigl(U_2(p_1),U_3(p_1),U_4(p_1)\bigr)\\
  &=\ \begin{cases}
    U_2(p_1) & \text{if }\tfrac14\le p_1\le b_1,\\
    U_3(p_1) & \text{if }b_1\le p_1\le b_2,\\
    U_4(p_1) & \text{if }b_2\le p_1<\tfrac13 .
  \end{cases}
\end{aligned}
\end{equation}
\end{theorem}

\begin{proof}
By Lemma~\ref{lem:noburial-q} every such source satisfies
\eqref{eq:twolevel-q}. By Propositions~\ref{prop:L1-q}
to~\ref{prop:L3-q}, which cover the three possible sizes of $L'$,
\begin{equation}\label{eq:allcorners}
  R\ <\ \max\bigl\{U_2,U_3,U_4,U_6,W_1,\dots,W_8\bigr\}(p_1).
\end{equation}
By Lemma~\ref{lem:dom-q}, every function in \eqref{eq:allcorners} other
than $U_2$, $U_3$, $U_4$ is at most $U_3$ or at most $U_2$ on all of
$[\tfrac14,\tfrac13]$, so $R<\max(U_2,U_3,U_4)(p_1)$.
Lemma~\ref{lem:cross-q} gives the piecewise form, since $U_2\ge U_3$ if
and only if $p_1\le b_1$, and $U_3\ge U_4$ if and only if $p_1\le b_2$.
\end{proof}

\begin{proposition}\label{prop:tight-q}
For every $t\in[b_2,\tfrac13)$, $\Rmax(t)=U_4(t)$, and the supremum is not
attained. For every $t\in[\tfrac14,\tfrac13)$, $\Rmax(t)\ge U_4(t)$.
\end{proposition}

\begin{proof}
We give sources whose redundancy tends to $U_4(t)$. Fix
$t\in[\tfrac14,\tfrac13)$ and $0<\varepsilon<\tfrac12$, and consider
the five-symbol source
\begin{equation}\label{eq:fam4}
  \bp^{(\varepsilon)}\ =\ \bigl(t,\,t,\,t,\,(1-3t)(1-\varepsilon),\,(1-3t)\varepsilon\bigr).
\end{equation}
It is sorted because $1-3t\le t$. The cuts of the root after positions
$1$, $2$ and $3$ have imbalances $1-2t$, $|4t-1|$ and $6t-1$, and the
second is strictly the smallest for $\tfrac14\le t<\tfrac13$. The cut
after position $4$ has imbalance $1-2(1-3t)\varepsilon\ge\tfrac12>|4t-1|$.
So the root splits into $\{1,2\}$ and $\{3,4,5\}$. The left child
contributes $0$. The right child rescales to a source with largest
probability $t/(1-2t)\ge\tfrac12$, so its split separates its first symbol
(Lemma~\ref{lem:iso}), and the remaining pair contributes
$(1-3t)\bigl(1-h(\varepsilon)\bigr)$. By \eqref{eq:rec} and the grouping
identity,
\begin{equation}\label{eq:fam4R}
\begin{aligned}
  R(\bp^{(\varepsilon)})\ &=\ \bigl(1-h(2t)\bigr)
  +(1-2t)\Bigl(1-h\Bigl(\frac{t}{1-2t}\Bigr)\Bigr)\\
  &\qquad+(1-3t)\bigl(1-h(\varepsilon)\bigr)\\
  &=\ U_4(t)-(1-3t)\,h(\varepsilon).
\end{aligned}
\end{equation}
Letting $\varepsilon\downarrow0$ gives $\Rmax(t)\ge U_4(t)$. For
$t\in[b_2,\tfrac13)$, Theorem~\ref{thm:quarter} gives $R<U_4(t)$ for
every source. So $\Rmax(t)=U_4(t)$, and the supremum is not attained.
\end{proof}

\begin{proposition}\label{prop:dom-q}
On $[\tfrac14,\tfrac13)$, $\max(U_2,U_3,U_4)<p_1+\CC$.
\end{proposition}

\begin{proof}
On $[\tfrac14,b_1]$ the difference $U_1-U_2$ has derivative
$2-\lgg\tfrac{p_1}{1/2-p_1}\ge2-\lgg2>0$, so it is smallest at
$\tfrac14$, where it equals $\tfrac{\CC}2$. The difference $U_1-U_3$ is
concave on $[\tfrac14,\tfrac13]$, because $U_3$ is convex. It equals
$\tfrac34\CC>0$ at $\tfrac14$. At $\tfrac13$,
\begin{equation}\label{eq:U1U3}
\begin{aligned}
  (U_1-U_3)(\tfrac13)\ &\overset{(a)}{>}\ \tfrac13+\CC-U_4(\tfrac13)\\
  &=\ \tfrac13+\CC-(2-\lgg3)\ \overset{(b)}{>}\ 0 ,
\end{aligned}
\end{equation}
where (a) follows from (N19) and (b) from (N20). So $U_1-U_3>0$ on
$[b_1,b_2]$. On $[b_2,\tfrac13)$ the difference $U_1-U_4$ has derivative
\begin{equation}\label{eq:U1U4prime}
\begin{aligned}
  1-U_4'(p_1)\ &=\ 4+3\lgg\frac{1-3p_1}{p_1}
  \ \overset{(a)}{\le}\ 4+3\lgg\frac3{23}\\
  &\overset{(b)}{<}\ 4+3\lgg2^{-4/3}\ =\ 0 ,
\end{aligned}
\end{equation}
where (a) follows from $p_1\ge\tfrac{23}{72}$ on this interval
(Lemma~\ref{lem:cross-q}(iii)), and (b) from
$(\tfrac3{23})^3=\tfrac{27}{12167}<\tfrac1{16}$. So the difference decreases to its limit
$\tfrac13+\CC-(2-\lgg3)$ at $\tfrac13$, which is positive by (N20).
\end{proof}

\section{Lower Bounds and Discussion}\label{sec:disc}

\subsection{Proof of Theorem~\ref{thm:exact}}

The upper bounds $\Rmax(t)\le U(t)$ are Theorem~\ref{thm:main}. The
equalities in \eqref{eq:exact} are Proposition~\ref{prop:tight}(a) and
Proposition~\ref{prop:tight-q}. The lower bound $\Rmax(t)\ge U_4(t)$ on
$[\tfrac14,\tfrac13)$ is the second part of
Proposition~\ref{prop:tight-q}. Finally $\sup_{\bp}R(\bp)=1$ is
Proposition~\ref{prop:tight}(b). \qed

\subsection{Sources with large redundancy}\label{ss:family}

Propositions~\ref{prop:tight} and~\ref{prop:tight-q} prove the
equalities in \eqref{eq:exact} with sources of one form. Each consists of
$m$ copies of $t$, followed by the rest $1-mt$ split into a large and a
small part. A source of this form exists for every $t\in(0,1)$, with
$m=\lceil1/t\rceil-1$. Letting the small part tend to $0$ gives a lower
bound on $\Rmax(t)$. The Fano tree of such a source contains the balanced binary tree
defined below.

For $N\ge1$ let $\mathcal B_N$ be the binary tree with $N$ leaves defined
as follows. The tree $\mathcal B_1$ is a single leaf. For $N\ge2$, the
tree $\mathcal B_N$ is a root whose left subtree is $\mathcal B_{\lfloor N/2\rfloor}$ and
whose right subtree is $\mathcal B_{\lceil N/2\rceil}$. Put
\begin{align}
  S_N\ &:=\ N\bigl(\lfloor\lgg N\rfloor+2\bigr)-2^{\lfloor\lgg N\rfloor+1},\label{eq:SNdef}\\
  d_N\ &:=\ \lceil\lgg N\rceil .\label{eq:ellN}
\end{align}

\begin{lemma}\label{lem:balanced}
Let $N\ge1$ and $k:=\lfloor\lgg N\rfloor$.
\begin{enumerate}
\item[\textup{(i)}] $\mathcal B_N$ has $2^{k+1}-N$ leaves at depth $k$
and $2(N-2^{k})$ leaves at depth $k+1$, and no others. The sum of the
depths of its leaves is $S_N$.
\item[\textup{(ii)}] The last leaf of $\mathcal B_N$, in left-to-right
order, has depth $d_N$.
\item[\textup{(iii)}] $S_{N+1}-d_{N+1}=S_N+1$.
\item[\textup{(iv)}] Let $m\ge1$, $0<r\le t$ and $0<\varepsilon<\tfrac12$,
and let $\bp=(t,\dots,t,r(1-\varepsilon),r\varepsilon)$, with $m$ copies
of $t$, be a source. Its Fano tree is $\mathcal B_{m+1}$ with the last
leaf replaced by a node whose two children are the leaves
$r(1-\varepsilon)$ and $r\varepsilon$.
\end{enumerate}
\end{lemma}

\begin{proof}
Parts (i) to (iii) follow by induction on $N$
(Appendix~\ref{app:balanced}).

(iv) Every block splits as in the definition of
$\mathcal B_{m+1}$, if the pair of the two last symbols is treated as one
leaf. Every block of the tree is of one of three kinds. It is a run of
$j\ge2$ copies of $t$, or a run of $j\ge1$ copies of $t$ followed by the
two last symbols, or the pair of the two last symbols. A run of $j$
copies of $t$ has the imbalances $|2i-j|\,t$ for $i=1,\dots,j-1$. They are
least at $i=j/2$ for even $j$ and at $i\in\{\tfrac{j-1}2,\tfrac{j+1}2\}$
for odd $j$, where the tie rule selects $i=\tfrac{j-1}2$. So the run
splits into $\lfloor j/2\rfloor$ and $\lceil j/2\rceil$ copies. A run of
$j$ copies followed by the two last symbols has mass $jt+r$. The cut
after $i$ copies has imbalance $|2it-jt-r|$. The cut between the two last
symbols has imbalance $jt+r-2r\varepsilon>jt-r$, which is the imbalance
of the cut after $i=j$, so it is never selected. For even $j$ the
imbalances are $r$ at $i=\tfrac j2$, $2t-r\ge r$ at $i=\tfrac j2+1$,
$2t+r$ at $i=\tfrac j2-1$, and larger elsewhere. For odd $j$ they are
$t-r$ at $i=\tfrac{j+1}2$, $t+r$ at $i=\tfrac{j-1}2$, and larger
elsewhere. Here the cut $i=j$ has imbalance $jt-r$, which equals $t-r$
only for $j=1$. In every case, with ties resolved toward the smaller $i$,
the cut is after $\lceil j/2\rceil$ copies. Finally the pair splits once.
So, with the pair counted as one leaf, a block of $j+1$ leaves splits into
$\lceil j/2\rceil=\lfloor(j+1)/2\rfloor$ and $\lceil(j+1)/2\rceil$ leaves,
and a run of $j$ leaves splits into $\lfloor j/2\rfloor$ and
$\lceil j/2\rceil$ leaves. This is the recursion that defines
$\mathcal B_{m+1}$.
\end{proof}

\begin{proposition}\label{prop:family}
Let $m\ge1$ and $\tfrac1{m+1}\le t<\tfrac1m$. Put $r:=1-mt\in(0,t]$,
$N:=m+1$, and
\begin{equation}\label{eq:Tm}
  L_m(t)\ :=\ (S_N-d_N)\,t+(d_N+1)\,r-\Ent(\underbrace{t,\dots,t}_{m},r).
\end{equation}
For $0<\varepsilon<\tfrac12$ the source
$\bp^{(\varepsilon)}=(t,\dots,t,r(1-\varepsilon),r\varepsilon)$ with $m$
copies of $t$ has largest probability $t$ and
\begin{equation}\label{eq:TmR}
  R\bigl(\bp^{(\varepsilon)}\bigr)\ =\ L_m(t)-r\,h(\varepsilon).
\end{equation}
So $\Rmax(t)\ge L_m(t)$.
\end{proposition}

\begin{proof}
The source is sorted because $r(1-\varepsilon)\le r\le t$, and it sums to
one. By Lemma~\ref{lem:balanced}(iv), (i) and (ii), the $m$ copies of $t$
are leaves of $\mathcal B_N$ whose depths sum to $S_N-d_N$, and the two
last symbols have depth $d_N+1$. So
$\bar\ell(\bp^{(\varepsilon)})=(S_N-d_N)\,t+(d_N+1)\,r$. The grouping identity
\eqref{eq:grouping} gives
$\Ent(\bp^{(\varepsilon)})=\Ent(t,\dots,t,r)+r\,h(\varepsilon)$, which
proves \eqref{eq:TmR}. Letting $\varepsilon\downarrow0$ gives
$\Rmax(t)\ge L_m(t)$.
\end{proof}

For $m=1$ and $m=3$ the lower bound is a piece of $U$, which recovers the
exact pieces. For $m=1$ we have $N=2$, $S_2=2$, $d_2=1$, and
$L_1(t)=t+2(1-t)-h(t)=U_7(t)$. For $m=3$ we have $N=4$, $S_4=8$,
$d_4=2$, and $L_3(t)=6t+3(1-3t)-\Ent(t,t,t,1-3t)=U_4(t)$. For $m=2$ we
have $N=3$, $S_3=5$, $d_3=2$, and $L_2(t)=3-3t-\Ent(t,t,1-2t)$.

On part of $[\tfrac13,\tfrac12)$ a source with four symbols has larger
redundancy than the source of Proposition~\ref{prop:family}.

\begin{proposition}\label{prop:family4}
Let $\tfrac13\le t<\tfrac12$ and $0<\varepsilon<\tfrac12$. The source
\begin{equation}\label{eq:fourfam}
  \bp^{(\varepsilon)}\ =\ \Bigl(t,\ \frac{1-t}2,\ \frac{1-t}2(1-\varepsilon),\
  \frac{1-t}2\,\varepsilon\Bigr)
\end{equation}
has largest probability $t$ and
\begin{equation}\label{eq:fourfamR}
  R\bigl(\bp^{(\varepsilon)}\bigr)\ =\ \frac32-\frac t2-h(t)-\frac{1-t}2\,h(\varepsilon).
\end{equation}
So $\Rmax(t)\ge\tfrac32-\tfrac t2-h(t)$.
\end{proposition}

\begin{proof}
The source sums to one, and it is sorted because $\tfrac{1-t}2\le t$ for
$t\ge\tfrac13$ and $\varepsilon<1-\varepsilon$ for $\varepsilon<\tfrac12$.
Since $2t+\tfrac{1-t}2\ge1$ if and only if $t\ge\tfrac13$, $\ell_1=1$
by Lemma~\ref{lem:iso}. The right child is the block
$\{2,3,4\}$ of mass $1-t$. Its cut after the first symbol has imbalance
$0$, so this cut is chosen, and the last two symbols form a block of
mass $\tfrac{1-t}2$. So the codeword lengths are $1,2,3,3$. By
Lemma~\ref{lem:rec} the redundancy is the sum of three terms. The root
contributes $1-h(t)$, the right child contributes
$(1-t)\bigl(1-h(\tfrac12)\bigr)=0$, and the last block contributes
$\tfrac{1-t}2\bigl(1-h(\varepsilon)\bigr)$. The sum is
\eqref{eq:fourfamR}. Letting $\varepsilon\downarrow0$ gives the bound on
$\Rmax(t)$.
\end{proof}

Together the two propositions give one lower bound on $(0,1)$. Define
\begin{equation}\label{eq:Tdef}
  L(t)\ :=\
  \begin{cases}
    L_m(t),\quad\text{if }\tfrac1{m+1}\le t<\tfrac1m\text{ and }m\ne2,\\[2pt]
    \max\bigl(L_2(t),\ \tfrac32-\tfrac t2-h(t)\bigr),\\
    \qquad\text{if }\tfrac13\le t<\tfrac12.
  \end{cases}
\end{equation}

\begin{corollary}\label{cor:L}
Every $t\in(0,1)$ satisfies $\Rmax(t)\ge L(t)$. The function $L$ is
continuous on $(0,1)$. There is a number $t^*\in(\tfrac25,\tfrac12)$ such
that for $\tfrac14\le t<1$
\begin{equation}\label{eq:Tpieces}
  L(t)\ =\
  \begin{cases}
    U_4(t), & \text{if }\tfrac14\le t<\tfrac13,\\
    \tfrac32-\tfrac t2-h(t), & \text{if }\tfrac13\le t\le t^*,\\
    3-3t-\Ent(t,t,1-2t), & \text{if }t^*<t<\tfrac12,\\
    U_7(t), & \text{if }\tfrac12\le t<1.
  \end{cases}
\end{equation}
\end{corollary}

\begin{proof}
The bound $\Rmax\ge L$ follows from Propositions~\ref{prop:family}
and~\ref{prop:family4}, and the first and last rows of \eqref{eq:Tpieces}
are the values of $L_3$ and $L_1$ computed above.

In this proof let $f(t):=\tfrac32-\tfrac t2-h(t)-L_2(t)$ for
$\tfrac13\le t<\tfrac12$. Then
\begin{equation}\label{eq:fourdiff}
\begin{aligned}
  f(t)\ =\ {}&-\frac32+\frac{5t}2-t\lgg t+(1-t)\lgg(1-t)\\
  &-(1-2t)\lgg(1-2t),
\end{aligned}
\end{equation}
and $f$ extends continuously to $t=\tfrac12$. Its second derivative is
\begin{equation}\label{eq:fourdiff2}
  f''(t)\ =\ \frac1{\ln2}\Bigl(-\frac1t+\frac1{1-t}-\frac4{1-2t}\Bigr),
\end{equation}
which is negative on $(\tfrac13,\tfrac12)$ because $\tfrac1{1-t}<\tfrac1t$
there. So $f$ is concave. Its values at the ends and at $\tfrac25$ are
\begin{equation}\label{eq:fourvals}
\begin{gathered}
  f(\tfrac13)\ =\ 0,\qquad
  f(\tfrac25)\ =\ \tfrac35\bigl(\lgg3-\tfrac32\bigr)\ >\ 0,\\
  f(\tfrac12)\ =\ -\tfrac14 .
\end{gathered}
\end{equation}
The middle value is positive because $3^2>2^3$. By concavity $f>0$ on
$(\tfrac13,\tfrac25]$. Since $f(\tfrac25)>0>f(\tfrac12)$, the
concave function $f$ has exactly one zero $t^*$ in $(\tfrac25,\tfrac12)$,
with $f>0$ before it and $f<0$ after it. So $f>0$ on $(\tfrac13,t^*)$ and
$f<0$ on $(t^*,\tfrac12)$. This gives the middle rows of \eqref{eq:Tpieces}.

The formula \eqref{eq:Tm} for $L_m$ is continuous on
$[\tfrac1{m+1},\tfrac1m]$, and so is the
maximum in \eqref{eq:Tdef}. So it suffices to check the switching points
$t=\tfrac1m$, $m\ge2$. At $t=\tfrac1m$ the
rest $r$ of $L_m$ vanishes, and
$L_m(\tfrac1m)=\tfrac1m(S_{m+1}-d_{m+1})-\lgg m$. The function $L_{m-1}$
has $r=t=\tfrac1m$ there, and
$L_{m-1}(\tfrac1m)=\tfrac1m(S_m-d_m)+\tfrac1m(d_m+1)-\lgg m=\tfrac1m(S_m+1)-\lgg m$.
The two values agree by Lemma~\ref{lem:balanced}(iii). At $t=\tfrac13$
the maximum equals $L_2(\tfrac13)$ because $f(\tfrac13)=0$. Near
$t=\tfrac12$ the maximum equals $L_2$ because $f<0$ there. So $L$ is also
continuous at $\tfrac13$ and $\tfrac12$.
\end{proof}

For $t=\tfrac1N$ the source of Proposition~\ref{prop:family} is the uniform source on
$N$ symbols with its last symbol split. The argument of
Lemma~\ref{lem:balanced}(iv) shows that the uniform source on $N$
symbols has the Fano tree $\mathcal B_N$, so the uniform source on $N$
symbols has redundancy $\tfrac{S_N}N-\lgg N$. Splitting its last symbol puts a mass $r=\tfrac1N$ one
level deeper, which adds $\tfrac1N$ in \eqref{eq:Tm}. Hence
$L(\tfrac1N)=\tfrac{S_N}N-\lgg N+\tfrac1N$.
In this paragraph write $N=2^{k}(1+y)$ with $0\le y<1$. Then
$\tfrac{S_N}N-\lgg N=2-\tfrac2{1+y}-\lgg(1+y)$. This is a concave function
of $y$. It vanishes at $y=0$ and as $y\to1$, and it has its maximum at
$1+y=2\ln2$, where it equals the constant $\sigma$ of Gallager's bound
\eqref{eq:gallager}. So $L(\tfrac1N)$ lies between $\tfrac1N$ and
$\sigma+\tfrac1N$. It does not converge as $N\to\infty$, because it
equals $\tfrac1N$ when $N$ is a power of two, while
$\tfrac{S_N}N-\lgg N\to\sigma$ along $N=\lfloor2^{k+1}\ln2\rfloor$. For
comparison, the upper bound $U_1(t)=t+\CC$ tends to $\CC$ as $t\to0$.

\subsection{The differences $U-L$}\label{ss:gaps}

Table~\ref{tab:gaps} lists $U$, $L$ and their difference at selected
points, and Fig.~\ref{fig:envelope} shows $L$ as a dotted curve. The
difference $U-L$ bounds the possible error of $U$ as an estimate of
$\Rmax$.
\begin{itemize}
\item On $[b_2,\tfrac13)$ and on $[\tfrac12,1)$ the difference is $0$.
\item On $[\tfrac14,b_2)$ the difference is $U_2-U_4$ up to $b_1$ and
$U_3-U_4$ from $b_1$ on. It vanishes at $b_2$.
\item On $[\tfrac13,\tfrac12)$ the difference is $U_5-L$ up to $b_3$ and
$U_6-L$ from $b_3$ on. Among the points of Table~\ref{tab:gaps} in
$[\tfrac13,\tfrac12)$ it is largest at $p_1=0.45$, where it is about
$0.127$.
\item On $(0,\tfrac14)$ the lower bound has one piece $L_m$ on each
interval $[\tfrac1{m+1},\tfrac1m)$, and $U_1-L$ is not monotone. At
$t=\tfrac1N$ the difference is
\begin{equation}\label{eq:gap-recip}
\begin{aligned}
  U_1(\tfrac1N)-L(\tfrac1N)\ &=\ \CC-\Bigl(\frac{S_N}N-\lgg N\Bigr)\\
  &\in\ [\CC-\sigma,\ \CC].
\end{aligned}
\end{equation}
For example, it equals $\tfrac85-\tfrac23\lgg5$ at $\tfrac15$,
$\tfrac1{10}$ and $\tfrac1{20}$. Between these points it can be larger,
as the rows $0.15$ and $0.22$ of Table~\ref{tab:gaps} show.
\end{itemize}

\begin{table}[t]
\renewcommand{\arraystretch}{1.25}
\caption{The bound $U$, the lower bound $L$ of
Corollary~\ref{cor:L}, and their difference rounded to six decimals. The
column ``source'' gives the number $m$ of copies of $p_1$ in the source of
Proposition~\ref{prop:family}, or it names
Proposition~\ref{prop:family4}.}
\label{tab:gaps}
\centering
\footnotesize
\begin{tabular}{@{}lcllll@{}}
\toprule
$p_1$ & piece of $U$ & source & $U(p_1)$ & $L(p_1)$ & $U-L$\\
\midrule
$0.05$ & $U_1$ & $19$ & $0.180120$ & $0.128072$ & $0.052048$\\
$0.10$ & $U_1$ & $9$ & $0.230120$ & $0.178072$ & $0.052048$\\
$0.15$ & $U_1$ & $6$ & $0.280120$ & $0.154538$ & $0.125582$\\
$0.20$ & $U_1$ & $4$ & $0.330120$ & $0.278072$ & $0.052048$\\
$0.22$ & $U_1$ & $4$ & $0.350120$ & $0.170639$ & $0.179481$\\
$0.25$ & $U_2$ & $3$ & $0.315060$ & $0.250000$ & $0.065060$\\
$0.28$ & $U_2$ & $3$ & $0.290266$ & $0.194322$ & $0.095944$\\
$0.31$ & $U_2$ & $3$ & $0.276039$ & $0.230061$ & $0.045978$\\
$0.3193$ & $U_3$ & $3$ & $0.273922$ & $0.272025$ & $0.001896$\\
$b_2$ & $U_3=U_4$ & $3$ & $0.273961$ & $0.273961$ & $0$\\
$\tfrac13$ & $U_5$ & $2$ & $0.458411$ & $0.415037$ & $0.043373$\\
$0.36$ & $U_5$ & Prop.~\ref{prop:family4} & $0.418955$ & $0.377317$ & $0.041638$\\
$0.40$ & $U_5$ & Prop.~\ref{prop:family4} & $0.368085$ & $0.329049$ & $0.039036$\\
$b_3$ & $U_5=U_6$ & Prop.~\ref{prop:family4} & $0.350664$ & $0.312656$ & $0.038008$\\
$0.45$ & $U_6$ & Prop.~\ref{prop:family4} & $0.409705$ & $0.282226$ & $0.127480$\\
$0.48$ & $U_6$ & $2$ & $0.484289$ & $0.357708$ & $0.126581$\\
$0.499$ & $U_6$ & $2$ & $0.558109$ & $0.484186$ & $0.073923$\\
$0.60$ & $U_7$ & $1$ & $0.429049$ & $0.429049$ & $0$\\
$0.80$ & $U_7$ & $1$ & $0.478072$ & $0.478072$ & $0$\\
\bottomrule
\end{tabular}
\end{table}

\subsection{Discontinuities}\label{ss:jumps}

The bound $U$ jumps at $\tfrac14$, $\tfrac13$ and $\tfrac12$. It jumps
down at $\tfrac14$, from $\tfrac14+\CC$ to $\tfrac14+\tfrac{\CC}2$, and at
$\tfrac12$, from $K$ to $\tfrac12$. It jumps up at $\tfrac13$, from
$2-\lgg3$ to $2-\lgg3+\tfrac{\CC}3$. Each jump is at a point where the
proof changes. At $\tfrac14$ Lemma~\ref{lem:noburial-q} becomes
available. At $\tfrac13$ the first symbol may get a codeword of length one, and if
$p_2<\tfrac12(1-p_1)$, the proof then bounds the right child of the root by
$(1-p_1)K$. At $\tfrac12$ we always have $\ell_1=1$, and the right
child is bounded by Lemma~\ref{lem:univ}. Since the lower bound $L$ is
continuous at all three points, nothing proved here says that $\Rmax$
jumps. We conjecture that $\Rmax$ is continuous on $(0,1)$, that
$\Rmax(\tfrac13)=2-\lgg3$, and that $\Rmax(t)\to\tfrac12$ as
$t\uparrow\tfrac12$.

\subsection{Exactness}\label{ss:exactness}

\begin{enumerate}
\item[\textup{(i)}] Equality in Theorem~\ref{thm:aff} holds for the
uniform source on three symbols and for no other source
(Example~\ref{ex:unif3}). Given $\lam$, the constant $\CC$ is the
smallest for which the case $(1,2)$ of the induction holds, by (S1).
\item[\textup{(ii)}] The coefficient $\lam$ is the smallest for which
\eqref{eq:ratio} holds (Lemma~\ref{lem:kappa}(iii)). This inequality
completes the case $(\ge3,\ge2)$, and it is an equality when the left child
is $(\tfrac15,\tfrac15,\tfrac15)$. Statements (i) and (ii) concern this
proof only. They do not show that $\CC$ is the smallest constant in a
valid bound $R<p_1+\CC$.
\item[\textup{(iii)}] Let $\CC^*$ be the supremum of $R(\bp)-p_1$ over all
sources. Its value is not known. It is at most $\CC$ by
Corollary~\ref{cor:aff}. It is at least $\sigma$, because the uniform
source on $N$ symbols gives $\CC^*\ge\tfrac{S_N}N-\lgg N-\tfrac1N$, and
$\tfrac{S_N}N-\lgg N$ tends to $\sigma$ along a sequence of $N$
(Section~\ref{ss:family}). On these uniform sources the Huffman and Fano
codes have the same codeword lengths, so the constant $\sigma$ of
Gallager's bound \eqref{eq:gallager} appears here as well.
\item[\textup{(iv)}] The supremum $\Rmax(t)$ equals the bound on
$[b_2,\tfrac13)$ and on $[\tfrac12,1)$ (Theorem~\ref{thm:exact}). It lies
between $L(t)$ and $U(t)$ elsewhere, and it is not known there.
\end{enumerate}
Since $K=\tfrac{1+\CC}2$, a smaller constant in Theorem~\ref{thm:aff}
would give a smaller $K$, provided the proofs of
Section~\ref{sec:midcap} remain valid with the new constant.

\subsection{Possible extensions}\label{ss:ext}

The proofs of this paper use three properties of the Fano code. The
split rule of Definition~\ref{def:fano} itself, applied to a sorted list,
determines the top of the tree (Lemmas~\ref{lem:iso}
and~\ref{lem:shapes}). Every split is nearly balanced
(Lemma~\ref{lem:key}). The redundancy is a sum of one term per node
(Lemma~\ref{lem:rec}). So we expect the method to apply, with
other constants, to the weight-balanced trees of Rissanen
\cite{Rissanen73} and Horibe \cite{Horibe77} and to $D$-ary Fano codes
\cite{CicaleseRossi20},
but not to alphabetic codes on an unsorted list \cite{Nakatsu91}, where
Lemmas~\ref{lem:iso} and~\ref{lem:shapes} have no counterpart. We have not proved this.

\subsection{Comparison with Gallager's bound}\label{ss:gallager}

Gallager's bound \eqref{eq:gallager} for Huffman codes and our bound
$R<p_1+\CC$ for Fano codes differ only in the constant, and
$\CC-\sigma\approx0.0440$. So the redundancy of every Fano code is less
than Gallager's bound for the Huffman code of the same source plus
$0.0441$. This compares two upper bounds. It does not bound the
difference between the redundancies of the two codes, because a Huffman
code can have redundancy well below $p_1+\sigma$.

\subsection{Open questions}\label{ss:open}

\begin{enumerate}
\item[\textup{(a)}] Is $\Rmax$ continuous on $(0,1)$? The lower bound $L$
is continuous, but the upper bound $U$ is not.
\item[\textup{(b)}] What is $\Rmax(t)$ on $[\tfrac14,b_2)$ and on
$[\tfrac13,\tfrac12)$? The bounds $U_2$, $U_3$, $U_5$ and $U_6$ exceed $L$ by the amounts
in Table~\ref{tab:gaps}.
\item[\textup{(c)}] What is $\CC^*$, the supremum of $R(\bp)-p_1$ over all
sources? We know that it lies in $[\sigma,\CC]$.
\item[\textup{(d)}] What is $\Rmax(t)$ for $t<\tfrac14$? In particular,
what are $\liminf$ and $\limsup$ of $\Rmax(t)$ as $t\to0$? The uniform
sources give $\limsup\ge\sigma$, and the bound gives $\limsup\le\CC$.
\end{enumerate}

\appendices

\section{Case B2}\label{app:polytope}

This appendix supports the proof of Lemma~\ref{lem:B2}. It checks the
corners \eqref{eq:W} and computes the values of $\Phi$ in
Table~\ref{tab:W}.
Fig.~\ref{fig:polytope} shows the quadrilateral. Each corner in
\eqref{eq:W} is the intersection of two of the four lines of
\eqref{eq:polytope} and satisfies the other two inequalities. The other
two intersections, $(\tfrac12,0)$ and $(1,0)$, violate $p_1+3p_2\ge1$ or
$p_1\le\tfrac12$. The four inequalities imply $p_1\ge\tfrac13$, so (P1)
is implied by the others.

\begin{figure}[t]
\centering
\begin{tikzpicture}
\begin{axis}[paper,
  width=\columnwidth, height=0.64\columnwidth,
  xlabel={$p_1$}, ylabel={$p_2$},
  xmin=0.24, xmax=0.6, ymin=0.12, ymax=0.38,
  grid=none,
  xtick={0.33333,0.4,0.5}, xticklabels={$\tfrac13$,$\tfrac25$,$\tfrac12$},
  ytick={0.16667,0.2,0.25,0.33333}, yticklabels={$\tfrac16$,$\tfrac15$,$\tfrac14$,$\tfrac13$}]
\addplot[black,thick,fill=black!10] coordinates {(0.33333,0.33333) (0.4,0.2) (0.5,0.16667) (0.5,0.25) (0.33333,0.33333)};
\addplot[only marks,mark=*,mark size=1.5pt,black] coordinates {(0.33333,0.33333) (0.4,0.2) (0.5,0.16667) (0.5,0.25)};
\node[font=\scriptsize,anchor=east] at (axis cs:0.352,0.25) {$2p_1+p_2=1$};
\node[font=\scriptsize,anchor=north east] at (axis cs:0.445,0.178) {$p_1+3p_2=1$};
\node[font=\scriptsize,anchor=west] at (axis cs:0.508,0.208) {$p_1=\tfrac12$};
\node[font=\scriptsize,anchor=south west] at (axis cs:0.425,0.298) {$p_1+2p_2=1$};
\end{axis}
\end{tikzpicture}
\caption{The quadrilateral \eqref{eq:W} in the $(p_1,p_2)$-plane, bounded
by the four lines of \eqref{eq:polytope}.}
\label{fig:polytope}
\end{figure}

In $\Phi=(2-2p_1+p_2+3z)-\Ent(p_1,p_2,z,v)-\CC(p_1+p_2+z)$ we write
$\CC(p_1+p_2+z)=4(p_1+p_2+z)-\tfrac53(p_1+p_2+z)\lgg5$. We use
$\lgg6=1+\lgg3$, $\lgg{15}=\lgg3+\lgg5$, $\lgg{18}=1+2\lgg3$,
$\lgg\tfrac52=\lgg5-1$ and $\lgg\tfrac{18}5=1+2\lgg3-\lgg5$. At
$(p_1,p_2,z)=(\tfrac13,\tfrac13,0)$ we have $v=\tfrac13$. The affine part
is $\tfrac53$, $\Ent=\lgg3$ and $p_1+p_2+z=\tfrac23$, so
$\Phi=\tfrac53-\lgg3-\tfrac83+\tfrac{10}9\lgg5=-1-\lgg3+\tfrac{10}9\lgg5$.
At $(\tfrac25,\tfrac15,\tfrac1{15})$ we have $v=\tfrac13$. The affine part
is $\tfrac85$, and
$\Ent(\tfrac25,\tfrac15,\tfrac1{15},\tfrac13)
=\tfrac25(\lgg5-1)+\tfrac15\lgg5+\tfrac1{15}(\lgg3+\lgg5)+\tfrac13\lgg3
=-\tfrac25+\tfrac25\lgg3+\tfrac23\lgg5$. With $p_1+p_2+z=\tfrac23$ this
gives $\Phi=-\tfrac23-\tfrac25\lgg3+\tfrac49\lgg5$. The remaining four
values are the same computation with $v=\tfrac15,\tfrac{5}{18},\tfrac16,\tfrac14$:
\begin{equation}\label{eq:Wvalues}
\begin{aligned}
  \Phi(\tfrac25,\tfrac15,\tfrac15)&=-\tfrac45+\tfrac13\lgg5,\\
  \Phi(\tfrac12,\tfrac16,\tfrac1{18})&=-\tfrac{23}9-\tfrac56\lgg3+\tfrac{40}{27}\lgg5,\\
  \Phi(\tfrac12,\tfrac16,\tfrac16)&=-\tfrac83-\tfrac12\lgg3+\tfrac{25}{18}\lgg5,\\
  \Phi(\tfrac12,\tfrac14,0)&=-\tfrac{13}4+\tfrac54\lgg5 .
\end{aligned}
\end{equation}
An upper bound on each value follows by substituting a bound on each
logarithm in the direction given by the sign of its coefficient. We substitute $\lgg3>\tfrac32$ in the
three values with a term in $\lgg3$, and $\lgg5<\tfrac73$ in all five. This
gives the bounds of Table~\ref{tab:W}:
$-\tfrac23-\tfrac35+\tfrac{28}{27}=-\tfrac{31}{135}$,
$-\tfrac45+\tfrac79=-\tfrac1{45}$,
$-\tfrac{23}9-\tfrac54+\tfrac{280}{81}=-\tfrac{113}{324}$,
$-\tfrac83-\tfrac34+\tfrac{175}{54}=-\tfrac{19}{108}$ and
$-\tfrac{13}4+\tfrac{35}{12}=-\tfrac13$.

\section{Details for the Range $\tfrac13\le p_1<\tfrac12$}\label{app:mid}

This appendix supports Section~\ref{sec:midcap}. It gives the entropy
arguments behind Tables~\ref{tab:midvalues} and~\ref{tab:midcurv}.

\emph{Entropy arguments.} Each function is an affine function of $p_1$
minus $\Ent$ of the affine arguments in the following table. The slope
of each argument is given in parentheses. The function $U_1$ has no
entropy term.
\begin{center}
\begin{tabular}{@{}ll@{}}
\toprule
function & arguments (slopes)\\
\midrule
$U_5$ & $p_1$ $(1)$, $1-p_1$ $(-1)$\\
$V_1$ & $p_1$ $(1)$, $p_1$ $(1)$, $1-2p_1$ $(-2)$\\
$U_6$ & $p_1$ $(1)$, $\tfrac{1-2p_1}3$ $(-\tfrac23)$, $\tfrac{2-p_1}3$ $(-\tfrac13)$\\
$V_2$ & $p_1$ $(1)$, $\tfrac{2(1-2p_1)}5$ $(-\tfrac45)$, $\tfrac{3-p_1}5$ $(-\tfrac15)$\\
$V_3$ & $p_1$ $(1)$, $\tfrac{2(1-2p_1)}3$ $(-\tfrac43)$, $\tfrac{1+p_1}3$ $(-\tfrac13)$\\
\bottomrule
\end{tabular}
\end{center}
By \eqref{eq:second}, the
second derivative of a difference $X-Y$ is $\tfrac1{\ln2}\sum_j s_j^2/u_j$
over the arguments $u_j$ of $X$, with $s_j$ the slope of $u_j$, minus the
same sum over the arguments of $Y$. This gives every row of
Table~\ref{tab:midcurv}. For example, for $U_6-V_2$ the argument
$p_1$ cancels, and
\begin{equation}\label{eq:appB-example}
\begin{aligned}
  &\frac{4/9}{(1-2p_1)/3}+\frac{1/9}{(2-p_1)/3}\\
  &\qquad-\frac{16/25}{2(1-2p_1)/5}-\frac{1/25}{(3-p_1)/5}\\
  &\quad=\ \frac4{3(1-2p_1)}-\frac8{5(1-2p_1)}\\
  &\qquad+\frac1{3(2-p_1)}-\frac1{5(3-p_1)}\\
  &\quad=\ -\frac4{15(1-2p_1)}+\frac{9-2p_1}{15(2-p_1)(3-p_1)} .
\end{aligned}
\end{equation}

\emph{Values.} At $p_1=\tfrac13$: $\Ent(\tfrac13)=\tfrac13\lgg3$,
$\Ent(\tfrac23)=\tfrac23(\lgg3-1)$, $\Ent(\tfrac19)=\tfrac29\lgg3$,
$\Ent(\tfrac59)=\tfrac59(2\lgg3-\lgg5)$,
$\Ent(\tfrac2{15})=\tfrac2{15}(\lgg3+\lgg5-1)$,
$\Ent(\tfrac8{15})=\tfrac8{15}(\lgg3+\lgg5-3)$,
$\Ent(\tfrac29)=\tfrac29(2\lgg3-1)$, $\Ent(\tfrac49)=\tfrac49(2\lgg3-2)$.
At $p_1=\tfrac25$: $\Ent(\tfrac25)=\tfrac25(\lgg5-1)$,
$\Ent(\tfrac35)=\tfrac35(\lgg5-\lgg3)$, $\Ent(\tfrac15)=\tfrac15\lgg5$,
$\Ent(\tfrac1{15})=\tfrac1{15}(\lgg3+\lgg5)$, $\Ent(\tfrac2{15})$ and
$\Ent(\tfrac8{15})$ as above,
$\Ent(\tfrac2{25})=\tfrac2{25}(2\lgg5-1)$,
$\Ent(\tfrac{13}{25})=\tfrac{13}{25}(2\lgg5-\lgg{13})$,
$\Ent(\tfrac7{15})=\tfrac7{15}(\lgg3+\lgg5-\lgg7)$. At $p_1=\tfrac7{16}$:
$\Ent(\tfrac7{16})=\tfrac7{16}(4-\lgg7)$,
$\Ent(\tfrac9{16})=\tfrac9{16}(4-2\lgg3)$, $\Ent(\tfrac18)=\tfrac38$,
$\Ent(\tfrac1{24})=\tfrac1{24}(3+\lgg3)$,
$\Ent(\tfrac{25}{48})=\tfrac{25}{48}(4+\lgg3-2\lgg5)$. At $p_1=\tfrac12$
we have $1-2p_1=0$, so $V(\tfrac12,0)=\tfrac32+\tfrac{\CC}2-h(\tfrac12)=K$.
Substituting these values into \eqref{eq:pieces} and \eqref{eq:midcurves},
with $\CC=4-\tfrac53\lgg5$ and $K=\tfrac52-\tfrac56\lgg5$, gives
Table~\ref{tab:midvalues}. For instance
$U_6(\tfrac25)=\tfrac{23}{15}+\tfrac8{15}\CC-\Ent(\tfrac25)-\Ent(\tfrac1{15})-\Ent(\tfrac8{15})
=\tfrac{17}3-\tfrac35\lgg3-\tfrac{17}9\lgg5$.

\section{Proofs of Propositions~\ref{prop:L2-q} and~\ref{prop:L3-q}}\label{app:quarter}

The proofs follow the proof of Proposition~\ref{prop:L1-q}.
Throughout this appendix, $\tfrac14\le p_1<\tfrac13$, $a$ is the mass of the left child of the
root, $x=a-p_1$, and $s=|1-2a|$.

\begin{proof}[Proof of Proposition~\ref{prop:L2-q}]
Here $L'=\{2,3\}$ and $p_3=x-p_2=p_k$. The block $L'$ contributes exactly
$R_{L'}=x\bigl(1-h(p_2/x)\bigr)$, and $x\,h(p_2/x)=\Ent(p_2)+\Ent(p_3)-\Ent(x)$
by the scaling identity. The rescaled right child has largest probability
$p_4/(1-a)\le p_3/(1-a)$. The constraints are $s\le p_3\le\tfrac x2$, by
Lemma~\ref{lem:key} and because $p_3$ is the smaller symbol of the pair,
and they force $a\in[\tfrac{2+p_1}5,\tfrac{2-p_1}3]$. Since $a\le\tfrac{2-p_1}3<\tfrac{1+p_1}2$, as
$p_1>\tfrac15$, we have $x<1-a$. So
$p_3/(1-a)\le x/(2(1-a))<\tfrac12$, and Lemma~\ref{lem:msup-q}(d) gives
$R_{B_R}<p_3+(1-a)\CC$. In this proof let
\begin{equation}\label{eq:w}
\begin{aligned}
  f(a,p_3)\ :=\ {}&(1+a)+x+p_3+(1-a)\,\CC\\
  &-\Ent(p_1)-\Ent(1-a)-\Ent(x-p_3)-\Ent(p_3).
\end{aligned}
\end{equation}
By \eqref{eq:twolevel-q}, $R<f(a,p_3)$. For fixed $a$, $f$ is convex in
$p_3$ (Lemma~\ref{lem:endpoint}), so
$f\le\max\bigl(f(a,s),f(a,\tfrac x2)\bigr)$.

\emph{The end $p_3=\tfrac x2$.} The function $f(a,\tfrac x2)$ is an affine
function of $a$ minus $\Ent(1-a)+2\Ent(\tfrac x2)$. So it is convex in
$a$ on $[\tfrac{2+p_1}5,\tfrac{2-p_1}3]$ and at most the larger of its
values at the two ends. At $a=\tfrac{2+p_1}5$ we have
$x=\tfrac{2-4p_1}5$, $\tfrac x2=\tfrac{1-2p_1}5=s$ and
$1-a=\tfrac{3-p_1}5$, and the value is $W_2(p_1)$. At $a=\tfrac{2-p_1}3$
we have $x=\tfrac{2-4p_1}3$, $\tfrac x2=\tfrac{1-2p_1}3=s$ and
$1-a=\tfrac{1+p_1}3$, and the value is $W_3(p_1)$.

\emph{The end $p_3=s$.} For $a\le\tfrac12$, $s=1-2a$ and
$x-s=3a-1-p_1$. For $a\ge\tfrac12$, $s=2a-1$ and $x-s=1-p_1-a$. On each
side $f(a,s)$ is an affine function of $a$ minus entropy terms of affine
arguments, so it is convex. Hence $f$ is at most its largest value at
$a\in\{\tfrac{2+p_1}5,\tfrac12,\tfrac{2-p_1}3\}$. These values are
$W_2(p_1)$, $U_2(p_1)$ and $W_3(p_1)$. The outer two are
the values found for the end $p_3=\tfrac x2$, because $s=\tfrac x2$ there. At $a=\tfrac12$ we
have $s=0$ and $x=\tfrac12-p_1$, and $f(\tfrac12,0)=\tfrac32+x+\tfrac{\CC}2-\Ent(p_1)-\Ent(\tfrac12)-\Ent(x)=U_2(p_1)$.
\end{proof}

\begin{proof}[Proof of Proposition~\ref{prop:L3-q}]
Now $L'\setminus\{2\}$ has at least two symbols, each at least $p_k$. So
$p_k\le\min\bigl(p_2,\tfrac{x-p_2}2\bigr)\le\tfrac x3$. The constraint
$s\le p_k\le\tfrac x3$ gives $a\in[\tfrac{3+p_1}7,\tfrac{3-p_1}5]$. So
$1-a\ge\tfrac{2+p_1}5$ and
$p_k/(1-a)\le\tfrac{x/3}{1-a}\le\tfrac{1-2p_1}{2+p_1}<\tfrac13$, as
$p_1>\tfrac17$. By Lemma~\ref{lem:msup-q}(a),
$R_{B_R}<p_{k+1}+(1-a)\CC\le p_k+(1-a)\CC$. Let $q':=p_2/x$ be the
largest probability of the rescaled block $L'$. In this proof write
\begin{equation}\label{eq:Jtilde}
  g(a)\ :=\ (1+a)-\Ent(p_1,\,a-p_1,\,1-a)+(1-a)\,\CC .
\end{equation}
By \eqref{eq:twolevel-q} and Lemma~\ref{lem:msup-q},
\begin{equation}\label{eq:L3total}
  R\ <\ g(a)+p_k+R_{L'} ,
\end{equation}
where $R_{L'}$ is bounded by Lemma~\ref{lem:msup-q}(a) to (d),
according to $q'$. The right side increases with $p_k$, so we put
$p_k=\min\bigl(p_2,\tfrac{x-p_2}2\bigr)$. The function $g$ is convex by
\eqref{eq:Jprime}.

\emph{Case $p_2\le\tfrac x3$.} Here $p_k=p_2$ and $q'\le\tfrac13$. Lemma~\ref{lem:msup-q}(d) with
$Q=q'$ gives $R_{L'}<x(q'+\CC)=p_2+x\CC$. The right side of \eqref{eq:L3total} is at most
$g(a)+2p_2+x\CC$. This increases in $p_2$, so it is at most
\begin{equation}\label{eq:Lambda}
  g_1(a)\ :=\ g(a)+\tfrac{2x}3+x\,\CC ,
\end{equation}
which is convex in $a$. Since $(1-a)\CC+x\CC=(1-p_1)\CC$, at
$a=\tfrac{3+p_1}7$, where $x=\tfrac{3(1-2p_1)}7$, we have $g_1=W_4$, and
at $a=\tfrac{3-p_1}5$, where $x=\tfrac{3(1-2p_1)}5$, we have $g_1=W_5$. So
$g_1\le\max(W_4,W_5)$.

\emph{Case $p_2\ge\tfrac x3$.} Here $p_k=\tfrac{x-p_2}2$, and the right
side of \eqref{eq:L3total} is
\begin{equation}\label{eq:gp2}
  g(a)+R_{L'}+\tfrac{x-p_2}2,\qquad p_2\in\bigl[\tfrac x3,\ x-2s\bigr].
\end{equation}
The upper end comes from $p_k\ge s$. So $q'\le1-\tfrac{2s}x$, and
$q'\ge\tfrac12$ requires $s\le\tfrac x4$. For $q'\in[\tfrac13,\tfrac12)$,
Lemma~\ref{lem:msup-q}(b) gives
$R_{L'}<x\max\bigl(\tfrac13+\CC,U_6(q')\bigr)$. The expression
$g+x(\tfrac13+\CC)+\tfrac{x-p_2}2$ decreases in $p_2$, so it is at most
its value $g_1(a)$ at $p_2=\tfrac x3$. The expression
$g+x\,U_6(p_2/x)+\tfrac{x-p_2}2$ is convex in $p_2$. Indeed, by the scaling
identity, $x\,U_6(p_2/x)$ is an affine function minus $\Ent$ of the
arguments $p_2$, $\tfrac{x-2p_2}3$, $\tfrac{2x-p_2}3$, which are affine in
$p_2$. Its maximum is at $p_2=\tfrac x3$, where (N21) bounds it by
$g_1(a)$, or at $p_2=\min(\tfrac x2,x-2s)$. For $q'\ge\tfrac12$ and
$s\le\tfrac x4$, Lemma~\ref{lem:msup-q}(c) gives $R_{L'}<x\,U_7(q')$. Since
$x\,U_7(p_2/x)=2x-p_2-\Ent(p_2)-\Ent(x-p_2)+\Ent(x)$, this expression is
convex in $p_2$. Its maximum is at $p_2=\tfrac x2$, where it is at most
$g+xK+\tfrac x4$ because $U_7(\tfrac12)=\tfrac12<K$, or at $p_2=x-2s$.

(1) \emph{The point $p_2=\tfrac x2$.} Let
\begin{equation}\label{eq:kappa0}
  g_2(a)\ :=\ g(a)+xK+\tfrac x4 .
\end{equation}
Since $U_6(\tfrac12)=K$, it also bounds the expression with $U_6$ at
$p_2=\tfrac x2$. It is defined for $s\le\tfrac x4$, that is,
$a\in[\tfrac{4+p_1}9,\tfrac{4-p_1}7]$, and it is convex. So
$g_2\le\max(W_7,W_6)$, its values at the two ends, where
$x=\tfrac{4(1-2p_1)}9$ and $x=\tfrac{4(1-2p_1)}7$.

(2) \emph{The point $p_2=x-2s$ with $\tfrac x4\le s\le\tfrac x3$.} Here
$q'\le\tfrac12$. The expression $g+x\,U_6(1-\tfrac{2s}x)+s$ has the
entropy arguments $x-2s$, $\tfrac{4s-x}3$ and $\tfrac{x+2s}3$, which are
affine in $a$ on either side of $a=\tfrac12$. The condition
$\tfrac x4\le s\le\tfrac x3$ means $a\in[\tfrac{3+p_1}7,\tfrac{4+p_1}9]$
or $a\in[\tfrac{4-p_1}7,\tfrac{3-p_1}5]$, one interval on each side of
$\tfrac12$. On each interval the expression is convex. At
$a=\tfrac{3+p_1}7$ and $a=\tfrac{3-p_1}5$ we have $s=\tfrac x3$ and
$q'=\tfrac13$, and (N21) bounds the expression by $g_1$, that is, by
$W_4$ and $W_5$. At $a=\tfrac{4+p_1}9$ and $a=\tfrac{4-p_1}7$ we have
$s=\tfrac x4$ and $q'=\tfrac12$, and the expression equals $g_2$, that is,
$W_7$ and $W_6$.

(3) \emph{The point $p_2=x-2s$ with $s\le\tfrac x4$.} Here $q'\ge\tfrac12$,
and the expression with $U_7$ is
\begin{equation}\label{eq:m}
\begin{aligned}
  g_3(a)\ :=\ {}&(1+a)-\Ent(p_1)-\Ent(1-a)+(1-a)\,\CC\\
  &+x+3s-\Ent(x-2s)-\Ent(2s).
\end{aligned}
\end{equation}
It is convex on either side of $a=\tfrac12$. At $a=\tfrac12$ we have
$s=0$ and $g_3=U_2(p_1)$. At $s=\tfrac x4$ the identity
$2\Ent(\tfrac x2)=x+\Ent(x)$ gives $g_3=g_2-x\bigl(K-\tfrac12\bigr)\le g_2$.

Collecting the values, $R<\max(U_2,W_4,W_5,W_6,W_7)(p_1)$.
\end{proof}

\section{Numerical Inequalities}\label{app:cert}

Tables~\ref{tab:cert}, \ref{tab:cert2} and~\ref{tab:cert3} list every
inequality between constants used in the paper.

\begin{table*}[!t]
\renewcommand{\arraystretch}{1.25}
\caption{Numerical inequalities, part 1 of 3.}
\label{tab:cert}
\centering
\footnotesize
\setlength{\tabcolsep}{3pt}
\begin{tabular}{@{}l>{\raggedright\arraybackslash}p{0.55\textwidth}l>{\raggedright\arraybackslash}p{0.18\textwidth}@{}}
\toprule
label & inequality & value & used in\\
\midrule
(N1) & $\tfrac{\lam-1}2<\CC$: $\tfrac12-\tfrac32\lgg3+\tfrac56\lgg5>0$ & $> \tfrac{29}{840}$ & \eqref{eq:constfacts}\\
(N2) & $1-h(\tfrac13)<\tfrac23\CC$: $-1-\lgg3+\tfrac{10}9\lgg5<0$ & $< -\tfrac{7}{2124}$ & \eqref{eq:constfacts}\\
(N3) & $\tfrac13+\CC<K$: $-\tfrac{11}6+\tfrac56\lgg5>0$ & $> \tfrac{17}{168}$ & \eqref{eq:constfacts}\\
(N4) & $\tfrac1{10}<\CC$: $\tfrac{39}{10}-\tfrac53\lgg5>0$ & $> \tfrac{1}{90}$ & case $(2,\ge2)$\\
(N5) & $\tfrac53-\lgg3-\tfrac{\CC}6>0$: $1-\lgg3+\tfrac5{18}\lgg5>0$ & $> \tfrac{113}{2520}$ & Lemma~\ref{lem:caseA}\\
(N6) & $(1-h(\tfrac13))-\tfrac19-\tfrac{\CC}3<0$: $\tfrac29-\lgg3+\tfrac59\lgg5<0$ & $< -\tfrac{7}{108}$ & Case B1\\
(N7) & $9\lgg3+2\lgg5-15>0$ & $> \tfrac{5}{2}$ & Lemma~\ref{lem:cross}(a)\\
(N8) & $(U_5-U_6)(\tfrac13)=-\tfrac49+\tfrac23\lgg3-\tfrac5{27}\lgg5>0$ & $> \tfrac{10}{81}$ & Lemma~\ref{lem:cross}(b)\\
(N9) & $(U_5-U_6)(\tfrac25)=-\tfrac{83}{30}+\tfrac65\lgg3+\tfrac7{18}\lgg5>0$ & $> \tfrac{13}{360}$ & Lemma~\ref{lem:cross}(c)\\
(N10) & $(U_6-U_5)(\tfrac7{16})=\tfrac{119}{96}-\tfrac{27}{16}\lgg3+\tfrac{185}{288}\lgg5>0$ & $> \tfrac{1241}{40320}$ & Lemma~\ref{lem:cross}(c)\\
(N11) & $(U_5-V_1)(\tfrac7{16})=-\tfrac{37}{32}+\tfrac98\lgg3+\tfrac{25}{96}\lgg5-\tfrac7{16}\lgg7>0$ & $> \tfrac{41}{34944}$ & Lemma~\ref{lem:middom}\\
(N12) & $(U_5-V_2)(\tfrac13)=-\tfrac83+\tfrac{11}9\lgg5>0$ & $> \tfrac{43}{252}$ & Lemma~\ref{lem:middom}\\
(N13) & $(U_5-V_3)(\tfrac13)=-\tfrac{20}9+\tfrac23\lgg3+\tfrac59\lgg5>0$ & $> \tfrac{17}{252}$ & Lemma~\ref{lem:middom}\\
(N14) & $(U_6-V_2)(\tfrac25)=\tfrac{92}{75}-\tfrac35\lgg3+\tfrac{32}{45}\lgg5-\tfrac{13}{25}\lgg{13}>0$ & $> \tfrac{29489}{14658525}$ & Lemma~\ref{lem:middom}\\
\bottomrule
\end{tabular}
\end{table*}

\begin{table*}[!t]
\renewcommand{\arraystretch}{1.25}
\caption{Numerical inequalities, part 2 of 3.}
\label{tab:cert2}
\centering
\footnotesize
\setlength{\tabcolsep}{3pt}
\begin{tabular}{@{}l>{\raggedright\arraybackslash}p{0.55\textwidth}l>{\raggedright\arraybackslash}p{0.18\textwidth}@{}}
\toprule
label & inequality & value & used in\\
\midrule
(N15) & $(U_6-V_3)(\tfrac25)=\tfrac{16}{15}+\tfrac19\lgg5-\tfrac7{15}\lgg7>0$ & $> \tfrac{47}{3276}$ & Lemma~\ref{lem:middom}\\
(N16) & $K-U_5(\tfrac13)=-\tfrac56+\lgg3-\tfrac5{18}\lgg5>0$ & $> \tfrac{1}{54}$ & Theorem~\ref{thm:mid}\\
(N17) & $(U_1-U_6)(\tfrac25)-\tfrac{\CC}2=-\tfrac{49}{15}+\tfrac35\lgg3+\tfrac{19}{18}\lgg5>0$ & $> \tfrac{211}{2520}$ & Proposition~\ref{prop:dom}\\
(N18) & $(U_2-U_3)(\tfrac13)=-\tfrac79+\tfrac76\lgg3-\tfrac{25}{54}\lgg5<0$ & $< -\tfrac{181}{61992}$ & Lemma~\ref{lem:cross-q}(i)\\
(N19) & $(U_3-U_4)(\tfrac13)=\tfrac{16}9-\tfrac23\lgg3-\tfrac{10}{27}\lgg5<0$ & $< -\tfrac{31}{378}$ & Lemma~\ref{lem:cross-q}(ii)\\
(N20) & $\tfrac13+\CC-(2-\lgg3)=\tfrac73+\lgg3-\tfrac53\lgg5>0$ & $> \tfrac{1}{36}$ & Proposition~\ref{prop:dom-q}\\
(N21) & $\tfrac13+\CC-U_6(\tfrac13)=\tfrac59+\tfrac53\lgg3-\tfrac{35}{27}\lgg5>0$ & $> \tfrac{5}{162}$ & Lemma~\ref{lem:msup-q}, Prop.~\ref{prop:L1-q}\\
(N22) & $(U_3-U_6)(\tfrac14)=-\tfrac56+\tfrac34\lgg3+\tfrac59\lgg5-\tfrac7{12}\lgg7>0$ & $> \tfrac{79}{13104}$ & Lemma~\ref{lem:dom-q}(a)\\
(N23) & $(U_3-W_8)(\tfrac14)=\tfrac78-\tfrac{21}{16}\lgg3+\tfrac{25}{48}\lgg5>0$ & $> \tfrac{181}{55104}$ & Lemma~\ref{lem:dom-q}(a)\\
(N24) & $(U_3-W_8)(\tfrac13)=\tfrac79-\tfrac76\lgg3+\tfrac{25}{54}\lgg5>0$ & $> \tfrac{181}{61992}$ & Lemma~\ref{lem:dom-q}(a)\\
(N25) & $(U_3-W_1)(\tfrac13)=\tfrac{23}{18}-\tfrac23\lgg3-\tfrac5{54}\lgg5>0$ & $> \tfrac{16}{3321}$ & Lemma~\ref{lem:dom-q}(a)\\
(N26) & $U_2(\tfrac13)-W_2(\tfrac14)=\tfrac{17}{20}-\tfrac12\lgg3+\tfrac56\lgg5-\tfrac{11}{20}\lgg{11}>0$ & $> \tfrac{5}{84}$ & Lemma~\ref{lem:dom-q}(b)\\
(N27) & $U_2(\tfrac13)-W_2(\tfrac13)=-\tfrac{12}5+\tfrac12\lgg3+\tfrac{13}{18}\lgg5>0$ & $> \tfrac{67}{2520}$ & Lemma~\ref{lem:dom-q}(b)\\
(N28) & $U_2(\tfrac13)-W_3(\tfrac14)=\tfrac{11}{12}+\tfrac14\lgg3-\tfrac59\lgg5>0$ & $> \tfrac{7}{432}$ & Lemma~\ref{lem:dom-q}(b)\\
\bottomrule
\end{tabular}
\end{table*}

\begin{table*}[!t]
\renewcommand{\arraystretch}{1.25}
\caption{Numerical inequalities, part 3 of 3.}
\label{tab:cert3}
\centering
\footnotesize
\setlength{\tabcolsep}{3pt}
\begin{tabular}{@{}l>{\raggedright\arraybackslash}p{0.55\textwidth}l>{\raggedright\arraybackslash}p{0.18\textwidth}@{}}
\toprule
label & inequality & value & used in\\
\midrule
(N29) & $U_2(\tfrac13)-W_3(\tfrac13)=-\tfrac{14}9+\tfrac76\lgg3-\tfrac5{54}\lgg5>0$ & $> \tfrac{49}{648}$ & Lemma~\ref{lem:dom-q}(b)\\
(N30) & $U_2(\tfrac13)-W_4(\tfrac14)=\tfrac5{28}-\tfrac54\lgg3-\tfrac5{42}\lgg5+\tfrac34\lgg7>0$ & $> \tfrac{1}{1260}$ & Lemma~\ref{lem:dom-q}(b)\\
(N31) & $U_2(\tfrac13)-W_4(\tfrac13)=-\tfrac{26}{21}+\tfrac5{14}\lgg3+\tfrac5{18}\lgg5+\tfrac23\lgg7-\tfrac{11}{21}\lgg{11}>0$ & $> \tfrac{1}{180}$ & Lemma~\ref{lem:dom-q}(b)\\
(N32) & $U_2(\tfrac13)-W_7(\tfrac14)=\tfrac5{36}+\lgg3+\tfrac{25}{108}\lgg5-\tfrac{19}{36}\lgg{19}>0$ & $> \tfrac{25}{1512}$ & Lemma~\ref{lem:dom-q}(b)\\
(N33) & $U_2(\tfrac13)-W_7(\tfrac13)=-\tfrac{16}9+\tfrac{11}6\lgg3+\tfrac{25}{162}\lgg5-\tfrac{14}{27}\lgg7>0$ & $> \tfrac{202}{7371}$ & Lemma~\ref{lem:dom-q}(b)\\
(N34) & $(U_2-W_5)(\tfrac14)=-\tfrac45-\tfrac65\lgg3+\tfrac76\lgg5>0$ & $> \tfrac{29}{4920}$ & Lemma~\ref{lem:dom-q}(c)\\
(N35) & $(U_2-W_5)(\tfrac13)=-\tfrac43+\tfrac3{10}\lgg3+\tfrac{17}{18}\lgg5-\tfrac7{15}\lgg7>0$ & $> \tfrac{391}{16380}$ & Lemma~\ref{lem:dom-q}(c)\\
(N36) & $(U_2-W_6)(\tfrac14)=-\tfrac{11}{14}+\tfrac5{28}\lgg5+\tfrac34\lgg7-\tfrac{13}{28}\lgg{13}>0$ & $> \tfrac{17}{3920}$ & Lemma~\ref{lem:dom-q}(c)\\
(N37) & $(U_2-W_6)(\tfrac13)=-\tfrac{38}{21}+\tfrac12\lgg3-\tfrac5{14}\lgg5+\tfrac23\lgg7>0$ & $> \tfrac{13}{840}$ & Lemma~\ref{lem:dom-q}(c)\\
(N38) & $(U_3-U_2)(\tfrac{23}{72})=-\tfrac{167}{108}-\tfrac{121}{72}\lgg3+\tfrac{295}{648}\lgg5+\tfrac{49}{36}\lgg7-\tfrac{13}{72}\lgg{13}>0$ & $> \tfrac{893459}{9948764412}$ & Lemma~\ref{lem:cross-q}(iii)\\
(N39) & $(U_3-U_4)(\tfrac{23}{72})=\tfrac{59}{54}-\tfrac{13}{18}\lgg3-\tfrac{245}{648}\lgg5+\tfrac{49}{36}\lgg7-\tfrac{23}{36}\lgg{23}>0$ & $> \tfrac{676867}{1196011656}$ & Lemma~\ref{lem:cross-q}(iii)\\
(N40) & $\lam\le\tfrac65$: $\tfrac{34}5+3\lgg3-5\lgg5\le0$ & $< -\tfrac{1}{140}$ & \eqref{eq:lamrange}\\
\bottomrule
\end{tabular}
\end{table*}

\section{The Balanced Tree}\label{app:balanced}

This appendix proves parts (i) to (iii) of Lemma~\ref{lem:balanced}.
As there, $k=\lfloor\lgg N\rfloor$.

(i) We use induction on $N$. For $N=1$ there is one leaf at depth
$0=k$. Let $N\ge2$, $N_1=\lfloor N/2\rfloor$ and $N_2=\lceil N/2\rceil$.
If $N=2^k$, both subtrees are $\mathcal B_{2^{k-1}}$. By the induction
hypothesis their $2^{k-1}$ leaves all have depth $k-1$, so all $N$ leaves
have depth $k$. If $2^k<N<2^{k+1}-1$, then $2^{k-1}\le N_1\le N_2<2^k$.
So $\lfloor\lgg N_1\rfloor=\lfloor\lgg N_2\rfloor=k-1$, and the subtree
$\mathcal B_{N_i}$ has $2^k-N_i$ leaves at depth $k-1$ and
$2(N_i-2^{k-1})$ leaves at depth $k$. One level lower, this gives
$2^{k+1}-N$ leaves at depth $k$ and $2(N-2^k)$ leaves at depth $k+1$. If
$N=2^{k+1}-1$, then $N_1=2^k-1$ and $N_2=2^k$. The right subtree
contributes $2^k$ leaves at depth $k+1$. The left subtree contributes
$2^k-N_1=1$ leaf at depth $k$ and $2(N_1-2^{k-1})=2^k-2$ leaves at depth
$k+1$. In total there are $1=2^{k+1}-N$ leaves at depth $k$ and
$2^{k+1}-2=2(N-2^k)$ leaves at depth $k+1$. The depth sum is
$k(2^{k+1}-N)+(k+1)\,2(N-2^k)=N(k+2)-2^{k+1}$.

(ii) The last leaf lies in the right subtree at every level, so its depth
$e_N$ satisfies $e_1=0$ and $e_N=1+e_{\lceil N/2\rceil}$. Hence $e_N$ is
the least $i$ with $\lceil N/2^i\rceil=1$, that is, with $N\le2^i$. This
is $\lceil\lgg N\rceil$.

(iii) If $N+1<2^{k+1}$, then $\lfloor\lgg(N+1)\rfloor=k$ and $d_{N+1}=k+1$.
So $S_{N+1}-S_N=k+2$ and $S_{N+1}-d_{N+1}=S_N+1$. If $N+1=2^{k+1}$, then
$S_{N+1}=2^{k+1}(k+1)$, $d_{N+1}=k+1$ and
$S_N=(2^{k+1}-1)(k+2)-2^{k+1}=2^{k+1}(k+1)-(k+2)$. So again
$S_{N+1}-d_{N+1}=S_N+1$.

% Balances the two columns of the last page.
\IEEEtriggeratref{12}

\end{document}